\documentclass[11pt]{article}

\usepackage{url,nicefrac,color}
\newcommand{\moni}[1]{\emph{\textcolor{red}{\textbf{Moni:} #1}}}

\newcommand{\barkar}[1]{\emph{\textcolor{blue}{\textbf{Bar:} #1}}}
\usepackage{multirow}
\usepackage{latexsym}
\usepackage{empheq}
\usepackage{amssymb}
\usepackage{amsthm}
\usepackage{amsmath}
\usepackage{lineno}
\usepackage{verbatim}
\usepackage{color}
\usepackage{tabularx}
\usepackage{url}
\usepackage[maxbibnames=99]{biblatex}
\usepackage{array}
\usepackage{epsfig}
\usepackage{hyperref}
\usepackage{cleveref}
\usepackage{bbm}
\usepackage{enumitem}

\usepackage{fullpage}
\usepackage{epigraph}
\AtBeginDocument{}

\usepackage{bm}
\usepackage{bbm}
\usepackage{etoolbox}
\usepackage{blindtext}
\usepackage{titlesec}
\titleformat{\section}{\large\bf}{\thesection}{1em}{}
\titleformat{\subsection}{\normalsize\bf}{\thesubsection}{1em}{}
\usepackage{titlesec}

\titleformat{\paragraph}
{\normalfont\normalsize\bfseries}{\theparagraph}{1em}{}
\titlespacing*{\paragraph}
{0pt}{3.25ex plus 1ex minus .2ex}{1.5ex plus .2ex}

\newtheorem*{assumption}{Assumption}

\newtheorem{theorem}{Theorem}

\newtheorem{definition}{Definition}
\newtheorem{claim}{Claim}[section]

\newtheorem*{conclusion}{Conclusion}

\def\bbE{{\mathbb E}}

\def\bbN{{\mathbb N}}

\def\bbR{{\mathbb R}}

\def\bbOne{{\mathbbm{1}}}
\newcommand{\half}{\frac{1}{2}}

\title{New Algorithms and Applications
for Risk-Limiting Audits\footnote{Research supported in part by grants from the Israel Science Foundation (no.2686/20),
by the Simons Foundation Collaboration on the Theory of Algorithmic Fairness and by the Israeli Council for Higher
Education (CHE) via the Weizmann Data Science Research Center.} \footnote{A shorter version of this paper appears in the Proceeding of the 4th Annual  Symposium on Foundations of Responsible Computing, FORC 2023.} }
\author{Bar Karov\footnote{Department of Computer Science and Applied Mathematics, Weizmann Institute of Science, Rehovot,
Israel. Email: barkrv@gmail.com} \and  Moni Naor\footnote{
Department of Computer Science and Applied Mathematics, Weizmann Institute of Science, Rehovot,
Israel. Incumbent of the Judith Kleeman Professorial Chair. Email: moni.naor@weizmann.ac.il. }}
\begin{document}
\maketitle

\begin{center}
    \textbf{Abstract}
\end{center}
Risk-limiting audits (RLAs) are a significant tool in increasing confidence in the accuracy of elections. They consist of randomized algorithms which check that an election's vote tally, as reported by a vote tabulation system, corresponds to the correct candidates winning. If an initial vote count leads to the wrong election winner, an RLA guarantees to identify the error with high probability over its own randomness. These audits operate by sequentially sampling and examining ballots until they can either confirm the reported winner or identify the true winner.

The first part of this work suggests a new generic method, called ``Batchcomp", for converting classical (ballot-level) RLAs into ones that operate on batches. As a concrete application of the suggested method, we develop the first ballot-level RLA for the Israeli Knesset elections, and convert it to one which operates on batches. We ran the suggested ``Batchcomp" procedure on the results of 22nd, 23rd and 24th Knesset elections, both with and without errors.

The second part of this work suggests a new use-case for RLAs: verifying that a population census leads to the correct allocation of political power to a nation's districts or federal-states. We present an adaptation of ALPHA, an existing RLA method, to a method which applies to censuses. Our census-RLA is applicable in nations where parliament seats are allocated to geographical regions in proportion to their population according to a certain class of functions (highest averages). It relies on data from both the census and from an additional procedure which is already conducted in many countries today, called a post-enumeration survey.
\newpage

\tableofcontents
\newpage

\section{Introduction}
\label{sec:Introduction}

\begin{center}
    \epigraph{``Those Who Vote Decide Nothing. Those Who Count The Votes Decide Everything"}{Attributed to many, including Joseph Stalin\footnotemark[1]} \footnotetext[1]{See the Snopes website \url{https://www.snopes.com/fact-check/stalin-vote-count-quote}}
\end{center}

Running an election is a delicate endeavour, since casting and tallying votes entails seemingly contradictory requirements: counting the votes should be accurate and it must also be confidential. 
A risk-limiting audit (RLA) is a process whose goal is to increase the confidence that results of  an election were tallied appropriately,  or more accurately that the winner/s were chosen correctly. It is usually described for election systems where there is an electronic vote tabulation, whose tally is referred to as the \textbf{reported results}, but also backup paper-ballots, whose tally is assumed to be the \textbf{true results}. The procedure examines what is hopefully a relatively small number of the backup paper-ballots, and comparing them to the full reported results of the electronic voting system. These audits are randomized algorithms, where the randomization is manifested in the choice of ballots to examine, and potentially the order in which they are examined.

A risk-limiting audit ends either when the reported winners of the election are confirmed, or after a full recount of the backup paper-ballots of all voters. The audit's goal is to confirm that the reported winners according to the electronic vote tabulation (the reported tally) match the winners according to the paper-backups (the true tally). Note that RLAs verify that the elections resulted in the correct winners according to the backup paper-ballots, and not that the reported vote tally was completely accurate; an RLA will approve election results that contain counting errors which do not change the winners of the elections. This fact is useful since it would be infeasible to expect the vote tally to be accurate up to every single ballot, but we should avoid at all cost counting errors which change the winners of the elections.

The claimed guarantee of  RLAs is that if the reported winners of the elections are not correct (with regards to the full paper count), then the probability that the audit will mistakenly confirm the results is lower than some predetermined parameter, referred to as the {\em risk-limit} of the audit.

\begin{center}
\fbox{\begin{minipage}{40em} \begin{center}\textbf{\hypertarget{The RLA Guarantee}{The RLA Guarantee}:} 

    If the reported winners of the elections are not correct, the probability that an RLA will approve them is at most $\alpha$, where $\alpha$ is a parameter which is set before the audit begins.
\end{center} \end{minipage}} \end{center}

The efficiency of an RLA is measured by the number of paper-ballots it requires to read, given that the reported tally matches the true one. In most cases, an RLA should remain relatively efficient even if the reported tally isn't completely accurate, as long as it results in the same winners as the true tally. The efficiency of any specific RLA method is limited by the election system it operates on. If a system has a sensitive social choice function, meaning that small tallying errors can often change the election winners, then it is more difficult to audit efficiently. % maybe add something about election system setting - interpetation of specific ballots 

Risk limiting audits provide several advantages over other type of post-election audits. First, they are software independent, meaning that they do not depend on the specific mechanism with which the voters cast their ballots. Additionally, they are publicly verifiable, as the audit can be easily broadcast to the public and verified by third parties. Lastly, they provide a clear statistical guarantee regarding the reported winners of the election. Previous post-election audits relied mostly on manual recounts of randomly selected precincts, without providing any clear statistical statement regarding the election results. 

For these reasons, RLAs are recommended by many bipartisan organizations who deal with election integrity~\cite{american2010american, center2013presidential, national2018securing}.  
They are currently used to audit a number of local and state-wide elections in the United States~\cite{NcslRLA}. Some states require them by law, while others are in more preliminary stages of their implementation. Specific RLA methods were also designed for a number of elections in Europe~\cite{stark2014verifiable} and Australia~\cite{blom2019raire}, but to the best of our knowledge they have yet to be implemented, with the sole exception of a preliminary pilot conducted in Denmark~\cite{schurmann2016risk}. 
 
The goal of the work is to expand the realm where RLAs can be used. First, in~\Cref{sec:Batchcomp}, we present a new generic method of converting many existing RLAs into batch-level RLAs. Afterwards, in~\Cref{sec:census RLA}, we suggest a new use case of RLAs- assessing whether a census correctly allocates political power to different regions. % Such a procedure could be used in the United States, for example, to verify that each state receives the correct number of representatives in the House of Representatives.

\subsection{Definition of Risk-Limiting Audits}
\label{sec:RLA definition}
Explaining RLAs from a mathematical perspective raises the need to formalize the audited election system. Most election systems can be defined by the set of possible votes a single voter may cast and a social choice function which outputs the winners of the elections. More formally, it requires 3 definitions:
\begin{itemize}
    \item Let $\mathcal{C}$ be the set of possible ballots a voter may cast. For example, in standard parliamentary elections (like the Israeli Knesset Elections), $\mathcal{C}$ would be the set of all running parties, plus an invalid ballot. In the US presidential elections, we could define $\mathcal{C}$ as the set of tuples of every running candidate and every federal-state\footnote{Including the 50 States, District of Columbia and the Maine and Nebraska districts.}.
    \item Let $O$ be the set of all possible election outcomes. An election outcome here is a winning candidate/s, or an allocation of seats to the running parties, and not an exact tally of votes.
    \item Let $f\colon \mathcal{C}^*\rightarrow O$ be the social choice function of the said elections. $f$ takes any number of ballots and outputs the appropriate election outcome.
\end{itemize}

After an initial (possibly electronic) tallying of the votes, the reported results are determined. These results are comprised of a reported election outcome, denoted as $o^{rep}\in O$, and some information regarding the ballots that were cast. Usually, this will be the vote tally - a count of how many ballots of each type were cast in the entire elections. In some instances, the certificate may even include the interpretation that the initial vote count (the reported results) gave for each specific paper-backup ballots. We refer to this extended information which is learned about the ballots, beyond only the reported winners of the elections, as the {\em certificate} of the reported results. This certificate can be used by a risk-limiting audit to increase its efficiency on instances where the reported tally is accurate, though it must fulfil the~\hyperlink{The RLA Guarantee}{RLA guarantee} even if it is not.

Using this notation, let $\mathcal{R}$ be a randomized algorithm with query-access to the paper-backup ballots. The algorithm receives, before it begins querying for ballots, a reported outcome $o^{rep} \in O$, a certificate $t$ as described above, and a risk-limit $0\leq \alpha \leq 1$. $\mathcal{R}$ can sequentially query for paper-backup ballots until it either outputs ''approve", or until all ballots were queried. If all ballots were queried, $\mathcal{R}$ knows whether the true outcome of the election according to the paper-backup ballots is $o^{rep}$, and outputs ''approve" or ''reject" accordingly. $\mathcal{R}$ is an RLA if it fulfills the~\hyperlink{The RLA Guarantee}{RLA guarantee}, meaning that it approves a wrongful outcome with probability of at most $\alpha$: 

\begin{definition}
    $\mathcal{R}$ is an RLA for an election system with possible ballots $\mathcal{C}$, a set of possible election outcomes $O$ and a social choice function $f:\mathcal{C}^*\rightarrow O$, if for any reported results $o^{rep}\in O$, certificate $t$, number of ballots $n\in \bbN$ and $q:[n]\rightarrow \mathcal{C}$, we have:
    $$
        f(q(1),q(2),\cdots, q(n)) \neq o^{rep} \Longrightarrow 
        \ \Pr[\mathcal{R}^{q(\cdot)}(o^{rep}, t, \alpha) = \text{approve}] \leq \alpha.
    $$
\end{definition}

Where $\mathcal{R}^{q(\cdot)}$ denotes that $\mathcal{R}$ has oracle access to $q$. $q$ here encodes the paper-backup ballots, such that $q(i)$ returns the $i$th ballot. The expression above essentially means that for any reported outcome of the elections $o^{rep}$, and for any sequence of paper-backup ballots, if tallying the  paper-backup ballots results in an outcome different than $o^{rep}$, then $\mathcal{R}$ approves w.p.\ of at most $\alpha$. This is exactly the~\hyperlink{The RLA Guarantee}{RLA guarantee} defined previously.

Under this definition, the efficiency of an RLA is measured by the query-complexity of $\mathcal{R}$ on $q:[n]\rightarrow \mathcal{C}$ for which $f(q(1),...,q(n))=o^{rep}$. Ideally, we would like this query complexity to be as small as possible when the certificate that $\mathcal{R}$ receives is representative of the paper-backup ballots. 

\subsection{Risk Limiting Audits - Limitations and Assumptions}
While RLAs are considered the golden-standard of post-election audits, they should always be part of a more exhaustive post election audit. This is due to some critical assumptions that are made by existing RLA methods, and need to be verified by some other mean. These assumptions are:

\begin{itemize}
    \item \textbf{The paper-backup ballots represent the true intention of the voters:} An RLA only verifies that the reported winners according to an initial vote tally matches the winners according to the paper-backup ballots tally. For this reason, it is critical that voters read and approve their paper-backup ballots, and that no party can alter these ballots after being cast.
    \item \textbf{The audit has access to random bits:} Most existing audits require sequential access to randomly selected paper-backup ballots. If the randomness used to choose these ballots is compromised, an adversary can cause an RLA to approve the wrong election winners. 
    In practice, random ballots are sampled using a pseudo random number generator with a random seed, which is generated by throwing dice at a public meeting\footnote{See Q5 on \url{https://www.sos.state.co.us/pubs/elections/RLA/faqs.html} and \url{https://arstechnica.com/tech-policy/2016/11/saving-american-elections-with-10-sided-dice-one-stats-profs-quest/}.}.
    \item \textbf{The total number of ballots is known:} If the number of ballots is not known, some ballots may not be considered during the audit. Moreover, if the number of paper-backup ballots does not match the total number of votes according to the electronic tabulation system, the audit may not fulfill its claimed guarantee. Ensuring that these two numbers match is usually done before an RLA begins, as part of a separate process called a compliance audit~\cite{lindeman2012gentle}. In most RLAs, if the two numbers do not match but the size of the discrepancy is known, it can be fixed by adding imaginary invalid ballots to the reported or true results.   
\end{itemize}

In addition to the points above, note that an RLA only verifies the vote tabulation, and not the voting process itself. If ineligible voters are allowed to vote, or if voters are coerced to vote in a specific manner, an RLA does not guarantee the integrity of the elections.

\subsection{Batch-level Risk-Limiting Audits}
\label{sec:intro Batch RLAs}
Most RLA models assume that the auditing party has the ability to repeatedly sample random ballots. However, in many real elections, the ballots are partitioned into batches such that it may be difficult to sample single ballots at random. This inspires the idea of batch risk-limiting audits. In a batch-level RLA, instead of sampling single ballots, we iteratively sample entire batches and then manually count the ballots in these batches. Additionally, we assume knowledge of the reported tally of each batch, according to the reported count. Batch audits are useful when retrieving a sequence of randomly chosen batches of ballots is easier than retrieving a sequence of randomly chosen single ballots.

During batch-level RLAs, we may not assume that ballots were partitioned into batches randomly; in practice, each batch is usually comprised of ballots cast at a different location, meaning different batches typically have different distributions of votes. Additionally, we cannot assume that different batches have similar probabilities of miscounting votes. If, for example, the initial vote count, which produces the reported results, is done at different locations for different batches, a malicious or faulty vote counter could produce many counting errors within the same batch. For these reasons, batch-level RLAs are considered more difficult than single-ballot based RLAs, and typically require a larger number of ballots to audit the same results.

One critical detail which is assumed to be known is the number of ballots contained in each batch. Without this knowledge, it would always be possible for a single unaudited batch to contain an extremely large number of ballots which would change the election winners. It is generally sufficient, however, to assume knowledge of an upper bound on number of ballots in each batch, at the cost of a certain decrease in the audit's efficiency. This could be done by imagining that in the reported tallies, each batch has a certain number of extra invalid ballots, and recalculating the batch totals accordingly. During the audit, if a batch has less than its reported number of ballots, we treat it as if the missing ballots are invalid ones. 

\subsection{Risk Limiting Audits for the Census}
\label{sec:intro RLA census}
Some political systems allocate political power to different regions of a country based on their population, as reported by a country-wide census. For example, in the United States' House of Representatives, each state is allocated a number of representatives in proportion to their population after every census, which is held every 10 years. Such systems are also used, for example, in the Danish Parliament (Folketing)~\cite{denmarkElectorealSystem} , in the Pakistani National Assembly~\cite{pakistanElectoralSystem}, and in the German Bundestag~\cite{bundestagAllocation}. Some nations re-allocate parliament seats to the nation's regions automatically following every census (e.g.\ Germany, USA, Denmark), while others require law amendments for each such update (e.g.\ Argentina~\cite{argentinaConstitution}, Cyprus~\cite{charalambous2008house}). In these systems, an inaccurate population count for one or more of these regions could lead to inadequate allocation of representation. For this reason, it is critical for a country's census to be as accurate as possible.

Many countries today assess the accuracy of their census by conducting an independent mini-census over a small number of randomly chosen households, in a process called a {\em post-enumeration survey} (PES). This survey is then compared to the census to estimate its accuracy. Statistical analyses of the census and the PES provide an estimation on the counting errors of the census regarding different population groups. In most cases, these comparisons do not provide any statistical assurances regarding the accuracy of the census, but help estimate the number of people who were under or overcounted\footnote{Such estimations are performed, for example, as part of the US PES: \url{https://www.census.gov/library/stories/2022/05/2020-census-undercount-overcount-rates-by-state.html}}~\cite{un2010pes}. 

The issue of verifying the original census by comparing it to an independently sampled and independently conducted mini-census is reminiscent of the problem RLAs were designed to solve. We have a reported tally of the number of residents in each region (the census), and we wish to verify that this tally is accurate by taking a small random sample of the households in these regions and re-running a smaller census over them (the PES). These similarities inspire a new use-case for RLAs - verifying that the census leads to the correct allocation of political representatives to federal-states.

\Cref{sec:census RLA} suggests a new RLA method which applies to population censuses. Just like a classical RLA sequentially samples ballots and learns their true content, this census RLA sequentially samples households and learns their ''true" number of residents, according to the PES. It eventually returns a probability $\alpha'$ with which it can approve the results. This returned probability comes with a statistical guarantee regarding the probability with which the census' resulting allocation of representatives to states matches the results of the PES.

\subsection{Our Contributions} \label{sec:our contributions}
The new contributions suggested in this work are: 
\begin{enumerate}
    \item A new and general method for performing {\bf batch-level RLAs}, which can be applied for many election systems, is presented in~\Cref{sec:Batchcomp}. This method, which we call "Batchcomp", is usable for any social choice function that can be audited using the SHANGRLA framework~\cite{stark2020sets}. To the best of our knowledge, a generic method for converting ballot-polling RLAs into batch-level RLAs was suggested only once before~\cite{stark2022alpha}. Our method is based on that conversion, and outperforms it significantly on real data from the election system we tested (the Israeli Knesset elections).
    
    \item An RLA method for the Israeli Knesset (The Israeli parliament) elections, based on the SHANGRLA framework, is presented in~\Cref{sec:Knesset RLA}. This method can be applied as-is to conduct ballot-level RLAs, or be combined with Batchcomp to conduct a batch-level RLA. To test both the Knesset RLA method and Batchcomp, we simulate their combination on the real results of three recent election cycles.   
    
    While our Knesset RLA method can be viewed as a synthesis and adaptation of previous suggestions in the literature, it is the first time RLAs are applied to this setting. Current recounts in Israel elections are done without an evidence-based approach.

    \item A new type of RLA that applies to population censuses. This new type of audit is applicable in nations where political representatives are allocated to the nation's geographical regions based on their population, like the United States, Germany, Cyprus and more. It relies on data that is already collected in many countries, as part of an existing method for assessing the accuracy of population censuses called a ''post enumeration survey" (PES).  To the best of our knowledge, this is the first and only method which verifies the census' resulting allocation of representatives to federal-states with a clear statistical guarantee. The method is presented in~\Cref{sec:census RLA}.
\end{enumerate}

\subsection{Related Work}
\label{sec:related work}
The need for post-election audits rose as early as 1969, when experts discovered that Los Angeles' computerized punch-card vote tabulation system could be secretly altered to rig election results~\cite{saltman1978effective}. Following this discovery, a number of state appointed committees suggested new methods to prevent fraudulent electronic vote tabulation. The Los Angeles county election security committee suggested ''A statistical recount of a random sample of ballots (should) be conducted after each election using
manual, mechanical or electronic devices not used for the specific election". To the best of my knowledge, this is the first proposal for verifying the results of an electronic vote tabulation system using a partial manual recount of ballots.

Before the advent of RLAs, post-election audits mostly consisted of a manual recount of the ballots cast in a number of randomly selected polling places. Early legislation in the United State demanded that a certain fixed percentage of polling places would be extensively audited~\cite{norden2007post}. These audits, however, {\em did not provide any statistical guarantee}, as they did not depend in any way on the margin of the elections. 
Other suggested auditing methods focused mainly on the number of ballots or polling places that would need to be examined to detect a result-altering miscount, as a factor of the election result's margin~\cite{mccarthy2008percentage, rivest2006estimating, simon2006end}. While some of these methods did provide some statistical guarantee, often under certain assumptions, they only involved manually recounting a set number of paper ballots. If the reported winners of the elections could not be approved based on this initial sample, a full manual recount would be required to complete the audit. This is unlike RLAs, which operate sequentially and can therefore overcome an unlucky initial sample.

One exception to the observations made above is a post-election audit suggested by Johnson~\cite{johnson2004election}. This method claims to be risk-limiting, and includes the option to sample additional ballots if the initial sample does not provide sufficient evidence that the reported winner of the election is correct. However, as pointed out previously by Lindeman and Stark~\cite{lindeman2012bravo}, it does not truly fulfill the~\hyperlink{The RLA Guarantee}{RLA guarantee}- the analysis of the risk-limit of this audit contains a critical error. In reality, the risk-limit of the audit can exceed the pre-set parameter as the audit samples additional ballots.

These issues raised the need for a new type of post-election audit, which provides a clear statistical guarantee, while having the ability to avoid a full manual recount even if the first ballots to be audited do not represent the true  distribution of all votes. RLAs, which were first introduced by Philip B.\ Stark in  2008~\cite{stark2008conservative} and received their name shortly after~\cite{stark2009risk}, fulfill both of these conditions. Early works in the field, headed by Lindeman and Stark, focused mainly on plurality elections~\cite{lindeman2012gentle, lindeman2012bravo},
where the candidate who receives the most votes wins the elections. Later works expanded the domain of RLAs to additional social choice functions~\cite{blom2019raire, stark2014verifiable, stark2020sets}. While often mentioned in literature as a tool for confirming the results of an electronic vote tabulations system, RLAs can be used to confirm any type of initial vote count. This can be the tally according to a computerized voting system, the result of an optical scan of paper-backup ballots, or a manual vote count.

Most RLA methods belong to one of three categories, as defined by Lindeman and Stark~\cite{lindeman2012gentle}:
\begin{enumerate}
    \item \textbf{Ballot-comparison audits:} In ballot-comparison audits, the auditor knows which paper-ballot matches which electronic-ballot. This category of audits is the most efficient, since it contains the most information about the election results. However, since they require finding a matching paper-ballot for any randomly selected electronic-ballot, they place a heavier burden on the body running the election. For this reason, they are seldomly used in practice~\cite{lindeman2012gentle, lindeman2018next}.
    
    \item \textbf{Ballot-polling audits:} In ballot-polling audits, a single paper-ballot can be sampled and examined, but it does not need to be matched to its corresponding electronic-ballot. This category of audits appears to be the most popular in practice~\cite{StateLegislatureConf}.
    
    \item \textbf{Batch-level audit:} In batch-level audits, ballots are partitioned into batches. The reported tally of each batch is available, but there’s no guarantee that a paper-ballot in the batch can be connected to its electronic counterpart. As mentioned in previous sections, ballots are usually not randomly partitioned, and different batches are of different sizes. Batch-level audits are generally the least efficient of the three categories, as they require reading more ballots to get a representative sample of the overall vote distribution.
\end{enumerate}

As mentioned, one of the main goals in recent RLA literature is to develop RLAs for additional election systems. Towards this purpose, Stark suggested a general framework called SHANGRLA~\cite{stark2020sets} which aids in adapting existing RLA algorithms to new social choice functions. This method is based on an abstraction called ``sets of half-average nulls'' (SHAN), where given a collection of lists containing unknown non-negative numbers, we wish to test whether the average of all of those lists is greater than $\half$ by querying for the values at different indexes. In the paper introducing SHANGRLA, it is shown that testing whether the reported winners of an election are correct, for many social choice functions, is reducible to the problem of ``sets of half-average nulls". SHANGRLA has opened the way for auditing new social choice functions, and therefore adapting RLAs to new election systems. A more detailed description of this framework is presented in~\Cref{sec:SHANGRLA}. 

Following SHANGRLA, a number of papers attempted to utilize and improve this framework: Blom, Stuckey and Teague~\cite{blom2021assertion} suggested an even more general way of reducing the problem of approving election results into the problem of SHAN, including such a reduction for election systems which use the D'Hondt method. Waudoby-Smith, Stark and Ramdas~\cite{waudby2021rilacs} and Stark~\cite{stark2022alpha} provide new and generally more efficient ways for solving the SHAN problem, thereby improving the efficiency of any SHANGRLA based RLA. Spertus and Stark~\cite{spertus2022sweeter} expanded the SHANGRLA framework to one that supports stratified RLAs, a type of ``split audit" which allows a certain part of the audit to use ballot-comparisons while only relying on ballot-polling for the rest.

One particularly useful algorithm that is based on the SHANGRLA framework is the ALPHA martingale test~\cite{stark2022alpha}. This test provides one of the most efficient solutions for the SHAN problem, meaning that every risk-limit auditing problem which can be reduced to SHAN can be solved with this test. One unique benefit of ALPHA is its relatively simple expansions to the fields of both stratified RLAs and batch-level RLAs. To the best of my knowledge, it provides the first and only batch-level RLA for the SHANGRLA framework, though its efficiency was not previously analyzed, either analytically or by using simulations. The batch RLA algorithm provided in~\Cref{sec:batch RLA} of this work relies on their suggested method.

The batch RLA method provided by ALPHA can convert any existing SHANGRLA based ballot-polling RLA into a batch-level RLA. That method, however, does not utilize the reported tallies of each batch; it only uses the reported winners of the elections, the sizes of the batches and the overall reported tally of the elections. This could be useful if the reported tallies of specific batches are not available, but is likely to be sub-optimal otherwise. One could naively convert their batch-level RLA method into one that uses these reported batch tallies, but any such simple conversion I could think of turned out to be less efficient than their original ALPHA's batch method, at least when simulated on the Israeli Knesset elections. This observation inspired our work towards new batch-level RLAs.

Since the new batch-level RLA method suggested in this work is based on SHANGRLA and ALPHA, it requires some understanding of these two works. The goal of the following two subsections is to provide all necessary information regarding them. We begin by presenting the SHANGRLA framework, and follow by showing how the ALPHA martingale test can be used to audit any election system which has a SHANGRLA-style reduction.

\subsubsection{The SHANGRLA Framework}
\label{sec:SHANGRLA}
As mentioned previously, one popular way of designing an RLA is the SHANGRLA framework~\cite{stark2020sets}. This framework is based on a reduction of the problem of verifying the election result to another problem, called  ``sets of half average nulls". The exact reduction is dependant on the social choice function used in the elections. Once a reduction for some specific election system is found, a number of existing algorithms~\cite{stark2022alpha, stark2020sets, waudby2021rilacs} for the ``sets of half average nulls" problem can be used to perform an RLA on that system.  

\paragraph*{Sets of Half-Average Nulls}
In this setting we have $\ell$ lists, each containing $n$ unknown entries of non-negative numbers. We denote the values in these lists as:
\begin{gather*}
    [x^1_1, x^1_2, ..., x^1_n]\\
    \vdots \\
    [x^\ell_1,x^\ell_2,...,x^\ell_n],
\end{gather*}
where we are guaranteed that for all $1 \leq i \leq \ell$ and $1 \leq j \leq n$ we have $x^j_i \geq 0$. The goal of an algorithm for this problem to determine w.h.p.\ (up to a pre-set parameter) whether the average of {\em all} of these lists is above $\half$. Meaning, to determine whether for every $j\in[\ell]$ we have:
$$    
    \frac{1}{n} \sum_{i=1}^n{x_i^j} > \half.
$$

The algorithm has query access to the values in the lists, where each query returns the values at some specified index in all lists.  Meaning, if the algorithm queries for index $i$, it learns the values of $x^1_i, x^2_i,...,x^\ell_i$.

The efficiency of such an algorithm for a specific input and a parameter $0\leq \alpha \leq 1$  is measured by the expected number of queries it performs to determine whether an input is a yes-instance (e.g. 
all lists have an average greater than $\half$), w.p.\ of at least $1-\alpha$. Typically, we wish that an algorithm would be as efficient as possible on yes-instances, but we do not care about its efficiency on no-instances.

For the purposes of RLAs, we are interested in randomized and adaptive algorithms for this setting. Such an algorithm can query for indexes sequentially, and decide after each query whether to query again, or to stop and declare that all of the list-averages are at greater than $\half$. 

A typical algorithm for this problem keeps $\ell$ p-values, each corresponding to a different list. Each of these values represents the probability of obtaining the query results we previously received if the average of its corresponding list is at most $\half$. The algorithm then queries for random indexes iteratively, where after each query it updates the p-values based on the values it learns. If all p-values are below $\alpha$ simultaneously, it decides the average of all lists is greater than $\half$. Otherwise, it queries for another random index.

\paragraph*{Reduction From the Problem of Approving Election Results}
Let $\mathcal{C}$ be the set of all possible ballots a single voter may cast, including the option to cast an invalid ballot. For example, in standard parliamentary elections (such as the elections for the Knesset), $\mathcal{C}$ includes all running parties, plus an invalid ballot. In ranked choice voting, $\mathcal{C}$ is the set of all permutations over all subsets of candidates, plus an invalid ballot. Let $B\in \mathcal{C}^n$ be the list of all ballots in an election with $n$ voters. For simplicity, we assume that the entries of $B$ are given in random order, and denote these ballots as $b_1,b_2,...,b_n$.

Given the reported winners of the elections, we wish to reduce the problem of finding whether these reported winners are the true winners, to the aforementioned problem of sets of half-average nulls. In the SHANGRLA framework, this reduction is done
by finding $\ell$ functions called {\em assorters}:
\begin{definition}
    Let $\mathcal{C}$ be the set of ballots a single voter may cast in some election system. A set of functions:
    $a_1,a_2,...,a_\ell\colon  C \rightarrow [0,\infty )$
    are \textbf{assorters} for the election system if they fulfill the condition:
    The reported winners are the true winners iff for every $ k\in[\ell]$ we have:
    $$
        \frac{1}{n}\sum_{i=1}^n{a_k(b_i)} > \half.
    $$
    These $\ell$ inequalities are referred to as the \textbf{assertions} of the audit.
\end{definition}

If we find such functions $a_1,...,a_\ell$, then an RLA could be performed by solving the SHAN problem on the following $\ell$ lists:
\begin{gather*}
    [a_1(b_1),a_1(b_2),...,a_1(b_n)]\\
    \vdots \\
    [a_\ell(b_1),a_\ell(b_2),...,a_\ell(b_n)],
\end{gather*}
where we query for an index by sampling the corresponding ballot and calculating $a_1,a_2,...,a_\ell$ over that ballot. Recall that when approving the reported election winners,  we wish to minimize the query complexity on inputs where the reported winners are correct. If the reported winners are not correct, a full recount would be in order anyway, so we would not mind it if the audit counts many (or potentially all) ballots to discover so. 

\paragraph*{Example - Plurality Elections}
Say we wish to audit a plurality election between two candidates, Alice and Bob, where the candidate who receives more votes wins. When we only have two running candidates and no invalid ballots, this reduces to a simple majority election.
If Alice reportedly won the elections, we could audit them using the SHANGRLA framework using a single assorter:
$$
    a(b) = 
    \begin{cases}
        1 & \text{if $b$ is for Alice} \\
        0 & \text{if $b$ is for Bob} \\
        \half & \text{if $b$ is invalid}
    \end{cases}
$$
and $\frac{1}{n} \sum_{b\in B}{a(b)} > \half$ iff Alice got more votes than Bob.

If we have more than 2 candidates, we could add one more similar assorter for every reportedly losing candidate $c$. This assorter has a mean of $\half$ or more iff the reportedly losing candidate $c$ receives less votes than Alice (the reported winner):
$$
    a_c(b) = 
    \begin{cases}
        1 & \text{if $b$ is for Alice} \\
        0 & \text{if $b$ is for $c$} \\
        \half & \text{if $b$ is invalid}
    \end{cases}
$$
Verifying that all such assorters have a mean greater than $\half$ using an algorithm for the SHAN problem is equivalent to verifying that Alice received more votes than all reportedly losing candidates, making Alice the true winner of the elections.

\subsubsection{Finding the Correct Assertions} \label{sec: finding assorters}
In the example above, finding the correct set of assertions and assorters is relatively simple. For other election systems, which use more complicated social choice functions, verifying the correctness of the election winners can sometimes be reduced to verifying a set of linear inequalities, but it is not  immediately clear how to reduce them to assertions of the form $\frac{1}{|B|}\sum_{b\in B}{a(b)}>\half$. For such cases, Blom et al.~\cite{blom2021assertion} suggests a generic solution, by reducing the problem of verifying that a set of linear inequalities that depend on the various vote tallies are all true to the problem of verifying that a set of assorters all have a mean greater than $\half$. This section explains this reduction.

Say we have $\ell$ inequalities that we wish to confirm, each of the form:
\begin{align} \label{eq: lin ineq form}
    \sum_{c\in \mathcal{C}}{\beta_c v^{true}(c)} > d,
\end{align}
where $v^{true}(c)$ is the number of cast ballots of of type $c$ according to the true results, and $d$ and $\beta_c$ (for each $c\in\mathcal{C}$) are constants. We wish to convert each inequality in the form of~\eqref{eq: lin ineq form} to an assertion in SHANGRLA form:
\begin{align} \label{eq: shangrla form}
    \frac{1}{|B|} \sum_{b\in B}{a(b)} > \half,
\end{align}
where $B$ is the list of all paper-backup ballots in the election and $a$ is a non-negative function. Meaning, given~\eqref{eq: lin ineq form}, we wish to find a function $a\colon C\rightarrow [0,\infty)$ such that~\eqref{eq: shangrla form} is equivalent to~\eqref{eq: lin ineq form}. As Blom et al.\ suggest, this is achieved by defining: 
\begin{align}
    a(b) := -\frac{\beta_b - z}{2\left(z-\frac{d}{|B|}\right)},
\end{align}
where $z:= min_{c\in\mathcal{C}}\left\{\beta_c\right\}$. The value $\beta_b$ here is the coefficient of the type of ballot $b$ is in~\eqref{eq: lin ineq form}. Using this definition for $a$, we have it that inequality~\eqref{eq: lin ineq form} is true iff~\eqref{eq: shangrla form} is true, and $a$ is a non-negative function, as required.

Note that this assorter is valid as long as $z-d/|B|<0$. Otherwise, it may return negative values. As explained by Blom et al., having $z-d/|B|\geq 0$ would indicate that~\eqref{eq: lin ineq form} is either always false or always true, for any distribution of votes. Thus, the assorters this method generates are non-negative in all non-trivial cases.

Given a set of inequalities as in~\eqref{eq: lin ineq form}, we can use this definition to create one SHANGRLA assertion (as in~\eqref{eq: shangrla form}) per inequality. The set of these SHANGRLA assertions are all true iff the set of the original inequalities are all true.

\subsubsection{The ALPHA Martingale Test}
\label{sec:alpha test}
This section explains the ALPHA martingale test~\cite{stark2022alpha} solution for the set of half-average nulls problem. Using the reduction described in~\Cref{sec:SHANGRLA}, this algorithm can be used to perform RLAs. For brevity, This description details how this algorithm operates on the problem of approving election results directly. This means that instead of writing $x^i_k$, as defined in the SHAN problem, we use $a_k(b_i)$, which is the the $i$th value in list number $k$ in the reduction described in~\Cref{sec:SHANGRLA}. The version described here relies on sampling ballots without replacement.

Before presenting the full algorithm, we provide a high level description of its operation- the ALPHA martingale test operates by keeping $\ell$ variables $T_1,...,T_\ell$, each representing the multiplicative inverse of a p-value for the hypothesis that a certain list has an average greater than $\half$. The test then queries sequentially for random paper-backup ballots, without replacement, where after each ballot it updates these $k$ variables. If at any point a statistic $T_k$ surpasses the threshold $\frac{1}{\alpha}$, it means that we have sufficient evidence that the mean of its corresponding assorter $a_k$ over all ballots is greater than $\half$. If after a certain query, all of $T_1,...,T_\ell$ have surpassed $\frac{1}{\alpha}$ at some point during the audit, then the reported winners of the elections are approved.

After each queried paper-backup ballots $b_i$, the algorithm updates the statistic $T_k$ for every $k\in[\ell]$. This update is performed by comparing $a_k(b_i)$ to the following values, which are set before $b_i$ is revealed:
\begin{enumerate}
    \item $\mu_k$: The mean value of $a_k$ over all  ballots that have yet to be audited, given that the mean of $a_k$ over all ballots is $\half$. Recall that if the mean of $a_k$ over all ballots is at most $\half$, then the reported winners of the elections are wrong, which is the case the algorithm wishes to detect. This means that if at some point during the audit, we sample a ballot $b$ with $a_k(b)\leq \mu_k$, it provides evidence that the reported winners of the elections are less likely to be correct, and vice-versa.
    
    \item $\eta_k$: A guess for what we would expect $a_k(b_i)$ to be based on the reported results and the ballots previously queried. This guess can be made in several ways while maintaining the algorithm's correctness. One reasonable way to do so is to set $\eta_k$ to be the mean of $a_k$ over ballots that have yet to be audited, assuming that the reported tally is completely accurate.   The audit becomes more efficient, meaning less ballots need to be examined, the more accurate this guess is.
    
    \item $u_k$: In the paper presenting ALPHA, $u_k$ was defined as the maximal value $a_k$ may return. In reality, the ALPHA martingale test is risk-limiting even for other choices of $u_k$, as long as the inequality $\mu_k<\eta_k<u_k$ is always maintained. For our purposes, $u_k$ can be thought of as a guess for whether the next sampled ballot would indicate that assertion $k$ is more or less likely to be true. If the next ballot to be sampled increases our confidence that the assertion is true, the audit is more efficient when $u_k$ is large, and vice-versa. 
\end{enumerate}
After each query, the test updates $T_k$ according to $a_k(b_i), \mu_k, \eta_k$ and $u_k$. If $T_k>\frac{1}{\alpha}$, it concludes that $\frac{1}{n}\sum_{b\in B}{a_k(b)}> \half$. Otherwise, it updates $\mu_k, \eta_k$ and $u_k$ in preparation for the next query.

If after querying for some ballot $b_i$ we have $a_k(b_i)\leq \mu_k$, then $T_k$ would shrink - indicating that it's now less likely that $\frac{1}{n}\sum_{b\in B}{a_k(b)}> \half$. Otherwise, if $a_k(b_i) > \mu_k$, then $T_k$ will increase. The magnitude with which $T_k$ increases depends on $a_k(b_i)$ and $\eta_k$. $T_k$ grows more significantly when $a_k(b_i)$ is large and when $a_k(b_i)$ is close to $\eta_k$. For this reason, we set $\eta_k$ to be the best guess we can make for the value $a_k$ would return on the next ballot we sample. 

The variable $u_k$, controls how stable $T_k$ is. Meaning, how substantially $T_k$ changes per ballot. Choosing a larger $u_k$ causes $T_k$ to be more stable, meaning that the magnitude of its change based on a single ballot is smaller. Choosing a smaller $u_k$ increases that magnitude and therefore raises the variance of the audit - it could cause it to finish earlier, since it allows $T_k$ to grow more substantially per ballot, but might slow it down or potentially cause it to read all ballots, if the order in which we sample ballots is ``unlucky".  

The algorithm presented here is a slightly altered version of the one presented in the original paper. The exact differences are discussed at the end of this section. 

\paragraph*{ALPHA Martingale Test Algorithm}
Let the inputs to the algorithm be ballots $B=(b_1,b_2,...,b_n)$, which are given to us in random order, and assorters $a_1,...,a_l:C\rightarrow [0,\infty)$ where $C$ is the set of all possible ballots a voter can cast. Recall that we assume that the reported winners of the elections are correct iff for all $k\in[\ell]$:
$$    
    \frac{1}{n}\sum_{i=1}^n{a_k(b_i)} > \half.
$$

The description below initializes $\eta_k$ to be the mean of $a_k$ over all ballots according to the reported results, and initializes $u_k$ to be the maximal value $a_k$ can return. Other initialization and update rules for $\eta_k$ and $u_k$ are also valid (the algorithm would still fulfill the~\hyperlink{The RLA Guarantee}{RLA guarantee}), as long as we always have $u_k>\eta_k>\mu_k$. The algorithm operates as follows: 

\begin{enumerate}
    \item \textbf{Initialization}
    \begin{enumerate}[label*=\arabic*.]
        \item Initialize $\mathcal{K} = [\ell]$. During the test, an index is removed from $\mathcal{K}$ whenever we have sufficient evidence that its corresponding assertion is correct.
        \item For each $k \in \mathcal{K}$ initialize: 
        \begin{itemize}
            \item $T_k:=1$.
            \item $\mu_k := \half$.
            \item $u_k:=\max_{b\in \mathcal{C}}\{a_k(b)\}$.
            \item $\eta_k:= a^{rep}_k(B)$, where $a^{rep}_k(B)$ is $\frac{1}{n}\sum^{n}_{i=1}a_k(b_i)$ given that the reported results are completely accurate.
        \end{itemize} 
    \end{enumerate}
    \item \textbf{Auditing Stage:} For each $i\in[n]$:
    \begin{enumerate}[label*=\arabic*.]
        \item \label{alpha step3} Sample the next paper backup-ballot $b_i$ and read it.
        \item For each $k\in \mathcal{K}$, update $T_k$:
        $$
            T_k \leftarrow T_k \left( \frac{a_k(b_i)}{\mu_k} \frac{\eta_k -\mu_k}{u_k - \mu_k}  + \frac{u_k - \eta_k}{u_k - \mu_k}\right)
        $$ \label{alpha t update}
        \item For each $k\in \mathcal{K}$, if $T_k > \frac{1}{\alpha}$, remove $k$ from $\mathcal{K}$. This means we have sufficient evidence that the assertion $\frac{1}{|B|}\sum_{b\in B}a_k(b) > \half$ is true.
        \item For each $k\in \mathcal{K}$ update $\mu_k, \eta_k$ and $u_k$:
        \begin{itemize}
            \item $\mu_k  \leftarrow \frac{\half n - \sum_{j=1}^{i}{a_k(b_j)}}{n - i}$
            \item $\eta_k \leftarrow max\left\{\mu_k + \epsilon, \frac{a_k^{rep}(B) - \sum_{j=1}^{i}{a_k(b_j)}}{n-i}\right\}$
            \item $u_k \leftarrow max\{u_k, \eta_k + \epsilon\}$ 
        \end{itemize}
        Where $a^{rep}_k(B)$ is $\frac{1}{n}\sum^{n}_{i=1}a_k(b_i)$ given that the reported results are completely accurate, and $\epsilon>0$ is some very small positive meant to ensure that $\mu_k < \eta_k < u_k$.
        \item \label{alpha step7} if $\mu_k<0$, the $k$th assertion is necessarily true, so remove $k$ from $\mathcal{K}$.
        \item If $\mathcal{K}=\emptyset$, approve the reported winners and finish the audit.
    \end{enumerate}
    \item \textbf{Output:} If the audit hasn't approved the reported winners yet, we recounted all ballots and the true winners are known.
\end{enumerate}
Note that after each iteration, we define $\mu_k$ to be the mean of $a_k$ over the remaining ballots, if the mean of $a_k$ over all ballots was $\half$. Conversely, $\eta_k$ is the mean of $a_k$ over the remaining ballots, if the reported tally is correct. 

\begin{theorem} \label{thm: alpha martingale test}
    For any election system that can be audited using the SHANGRLA framework and for any $0 \leq \alpha \leq 1$, if the reported winners of the elections are wrong, then the ALPHA martingale test will approves the results with probability of at most $\alpha$.
\end{theorem}
\begin{proof}
Fix $\alpha \geq 0$, a list of backup paper-ballots $B$ whose tally is the true results of the election and some wrongful reported tally regarding them which leads to the wrong winners. Let the set of assorters used to audit the ballots be $a_1, \ldots, a_\ell$ . Since we assume that the set of reported winners of the elections is wrong, there must be some $a_k$ with $\frac{1}{n}\sum_{i=1}^n{a_k(b_i)}<\half$. Assume w.l.o.g.\ that it is $a_1$.

\begin{assumption} \label{ass: assertion 1 is wrong}
    $\frac{1}{n}\sum_{i=1}^n{a_1(b_i)}<\half$
\end{assumption}

To prove that the algorithm fulfills the~\hyperlink{The RLA Guarantee}{RLA guarantee}, we need to show that the test approves the reported winners w.p.\ of at most $\alpha$. It suffices to show that the algorithm approves that the mean of $a_1$ over all ballots is greater than $\half$ w.p.\ of at most $\alpha$. Meaning, that the probability of 1 getting removed from the set $\mathcal{K}$ is at most $\alpha$.

The index 1 cannot be removed from $\mathcal{K}$ in step~\ref{alpha step7}, since that would mean that at some point during the audit we had:
$$
    \mu_k < 0 \Longrightarrow \half n - \sum_{j=1}^{i}{a_1(b_j)} < 0 \Longrightarrow \half < \frac{1}{n} \sum_{j=1}^{i}{a_1(b_j)} \leq  \frac{1}{n} \sum_{i=1}^{n}{a_1(b_j)},
$$
contradicting our assumption. This means that the algorithm approves the results only if at some point, $T_1>\frac{1}{\alpha}$. 

Let $b_1,...,b_n$ be random variables which represent the ballots that are sampled by the audit, in the order in which they are sampled. Each of these values is a random variable which depends on $B$ and on the randomness of the audit. Denote by $T^0_1, T^1_1, \ldots, T^n_1$ the values of $T_1$ after each sampled ballot, where $T^0_1$ is its initial value. Similarly, let $\mu^1_1, \mu^2_1,...,\mu^n_1$, and $\eta^1_1, \eta^2_1,...,\eta^n_1$ and $u^1_1, u^1_1,...,u^n_1$ be the values that $\mu_1$, $\eta_1$ and $u_1$ have, respectively, when sampling each ballot. By their definition, each $T_1^i$ is a random variable whose value is determined by $b_1, ..., b_i$, and each of $u^i_1, \mu^i_1$ and $\eta^i_1$ are determined only by  $b_1, ..., b_{i-1}$, and not by $b_i$.

With these definitions in mind, we use Ville's inequality~\cite{durrett2019probability} (also referred to as Doob's inequality), to show that:
$$
    \Pr(\exists i\in[n],\ T^i_1> \frac{1}{\alpha}) \leq \alpha,
$$
thereby proving that the algorithm fulfills the~\hyperlink{The RLA Guarantee}{RLA guarantee}.
\begin{center}

    \fbox{\begin{minipage}{40em}
      \begin{center} \label{Ville's inequality}
        \textbf{Ville's Inequality~\cite{durrett2019probability}}
      \end{center}
      If $X_1,X_2,...,X_n$ is a non-negative supermartingale, meaning that for any $i\in[n]$ we have $\Pr[X_i\geq 0]=1$ and
          $\bbE[X_i|X_{1},X_{2},...,X_{i-1}] \leq X_{i-1}$, then for any $\alpha>0$:
        $$
            \Pr\left[\max_{i\in[n]}\{X_i\}>\frac{1}{\alpha} \right] \leq \alpha \cdot \bbE[X_1].
        $$
    \end{minipage}}
\end{center}

To use this inequality, we need to show that $T^1_1,...,T^n_1$ is a non-negative supermartingale, which we do in the following two claims:
\begin{claim} \label{claim: t non negative}
    $T^i_1$ is non-negative for every $i\in [n]$.  
\end{claim}
\begin{proof}
    Fix $i\in[n]$ and observe the update rule of $T_1$ in step~\ref{alpha t update}. Since we have $0 \leq \mu_1^i < \eta_1^i < u_1^i$, and since $a_1$ is a non-negative function, we have:
    $$
        T^i_1 = T^{i-1}_1 \left( \underbrace{\frac{a_1(b_i)}{\mu^i_1}}_{\geq 0} \underbrace{\frac{\eta^i_1 -\mu^i_1}{u^i_1 - \mu^i_1}}_{>0}  + \underbrace{\frac{u^i_1 - \eta^i_1}{u^i_1 - \mu^i_1}}_{>0}\right).
    $$
    And since $T^0_1=1$, by induction, $T^i_1$ is non-negative, concluding the proof of this claim.
\end{proof}

\begin{claim} \label{claim: t non increasing}
    For any $i\in[n]$, we have $\bbE[T_1^i\,|\,T_1^1,...,T_1^{i-1}] \leq T_1^{i-1}$.
\end{claim}    
\begin{proof}
    First, note that since we are under the assumption that $\frac{1}{n}\sum_{j=1}^n{a_1(b_j)} < \half$, we have:
    $$
        \bbE[a_1(b_i)] \leq \mu^i_1,
    $$
    where the expectation is over the choice of the $i$th ballot that is audited, $b_i$.

    Now, observe that fixing the first $i-1$ ballots that were audited $b_1,...,b_{i-1}$ fixes $T_1^1,...,T^i_1$, and vice-versa, since $T_1^1,...,T^i_1$ is deterministically determined by the ballots that are audited. Thus, we have it that:
    \begin{align} \label{eq: t and b}
        \bbE[T_1^i\,|\,T_1^1,...,T_1^{i-1}] &=\bbE[T_1^i\,|\,b_1,...,b_{i-1}]. \\
        \intertext{Once $b_1,...,b_{i-1}$ are fixed, the values $\mu^i_1,\eta^i_1$ and $u^i_1$ are also fixed, and can be calculated by simulating the audit on the first $i-1$ ballots. Continuing from~\eqref{eq: t and b}, by this the update rule of $T_1$ in step~\ref{alpha t update}, we have:}
         &=T^{i-1}_1 \left( \frac{\bbE[a_1(b_i)]}{\mu^i_1} \frac{\eta^i_1 -\mu^i_1}{u^i_1 - \mu^i_1}  + \frac{u^i_1 - \eta^i_1}{u^i_1 - \mu^i_1}\right) , \nonumber \\
         \intertext{and since $\bbE[a_1(b_i)] \leq \mu^i_1$:}
        &\leq T^{i-1}_1 \left( \frac{\eta^i_1 -\mu^i_1}{u^i_1 - \mu^i_1}  + \frac{u^i_1 - \eta^i_1}{u^i_1 - \mu^i_1}\right) \nonumber \\[1.5ex]
        &= T^{i-1}_1 \left( \frac{u^i_1 -\mu^i_1}{u^i_1 - \mu^i_1}\right) \nonumber \\[1.5ex]
        &=T^{i-1}_1, \nonumber
    \end{align}
    proving the claim.
    \end{proof}
    By these two claims, $T^1_1,..., T^n_1$ is a non-negative supermartingale. This concludes the proof, since Ville's inequality states that for any $\alpha \in[0,1]$ we have:
    \begin{equation*}
        \Pr\left[ \max_{i\in[n]}\{T^i_1\} > \frac{1}{\alpha} \right] \leq \alpha \cdot \bbE[T^0_1] = \alpha
    \end{equation*}
    and therefore the probability that the algorithm wrongfully approves the reported election winners is at most $\alpha$.
\end{proof}
    
\paragraph*{Changes from Original Algorithm and Proof}
In the original ALPHA martingale test, $u_k$ is defined as the maximal value that $a_k$ may return. In the definition above, $u_k$ is seen as a variable which controls the magnitude with which $T_k$ changes. The proof presented by Stark~\cite{stark2022alpha} doesn't actually require that $u_k = max_{b\in \mathcal{C}}{a_k(b)}$, but only requires to have $u_k>\eta_k>\mu_k$ in every iteration. 

Additionally, in the original statement of the SHAN problem, we assume to know an upper bound on the values in each list. In the statement in~\Cref{sec:SHANGRLA} of this work, this assumption is omitted. 

For most audits, this distinction does not matter, since the assorters frequently return the maximal value in their image. However, for batch-level RLAs, this distinction allows for more efficient audits, as will be explained in later sections. If it was not for this alternate definition of $u_k$, then the efficiency of the Batchcomp algorithm presented in~\Cref{sec:batch RLA} would be significantly reduced.

\subsection{Road Map}
\label{sec:Road Map}
\Cref{sec:batch RLA} formally defines the batch-level RLA model, and presents a new, general method for performing batch-level RLAs. This method is usable for every social choice function which could be reduced to SHANGRLA assertions as described in~\Cref{sec:SHANGRLA}.

While this method is generic and could be used to convert many existing RLAs to batch-level RLAs,~\Cref{sec:Knesset RLA} focuses on its application for the Israeli Knesset elections. In the Israeli Knesset elections, 120 seats are allocated to different parties using party-list proportional representation, according D'Hondt method (also known as the Jefferson method) with a few added caveats. The seats of each party are then allocated to specific party-members according to a ranked list which is submitted by the parties ahead of the elections. Simulated results of the Batchcomp method compared to the ALPHA-batch method are shown in~\Cref{sec:batch RLA simulations}.

Finally,~\Cref{sec:census RLA} presents a new type of RLA which verifies that a country's population census results in correct allocation of political power to different regions within that country, as is done e.g.\ in the US.~\Cref{sec: census RLA sim} shows simulated results for the application of this method on the census and house of representatives of Cyprus.

\subsection{Acknowledgments}
\label{sec:Acknowledgments}
I would like to express my deepest gratitude to my supervisor, Moni Naor, for introducing me to the field and providing guidance, feedback and support throughout this process.

I would additionally like to thank my parents and friends, primarily Daniel Shwartz and Nicole Kezlik, for listening to my rumblings and providing some essential feedback, professional or not. 

\newpage
\section{Batch Risk-Limiting Audits}
\label{sec:batch RLA}

\subsection{Preliminaries and Notation}
\label{sec:batch RLA prelims}

\subsubsection{The Batch-level RLA Model}
\label{sec:Batch RLA Model}

In the batch-level RLA model, ballots are partitioned into batches, denoted as $B_1,B_2,...,B_d$. The set of all ballots in the elections, which is the union of these batches, is denoted as $B$. As mentioned previously, we make no assumption on how the ballots were partitioned into batches, but we do assume that their size is known. Our goal, just as before, is to perform an RLA for the reported election winners. However, instead of sampling single ballots, we can now only sample a complete batch, and tally all ballots in it. We assume that we know the reported tally of each batch. Meaning, instead of only knowing the reported tally of all of the votes, we also know the individual reported tally of each batch and can use it during the audit. 

Such a model could be useful for election systems where each polling place tallies its own votes, and the reported winners of the elections (pre-audit) are determined according to the sum of the tallies. If a governing body wishes to audit these results and verify that the reported winners, as calculated by the tally each polling place performed, are accurate, it can use a batch-level RLA to do so. Note that this audit doesn't ensure that all batches were counted accurately. It only verifies that it is unlikely that there is a counting mistake which {\em changes  the winners of the elections}.

\subsection{The Batchcomp RLA}
\label{sec:Batchcomp}
This section describes a generic and efficient way of performing batch-level RLAs, when the results of the elections can be verified using SHANGRLA assertions, as described in~\Cref{sec:SHANGRLA}. This algorithm is original to this work. If an election system has a ballot-level RLA which uses the SHANGRLA framework, this method can be used to audit in the batch-level RLA model. The inspiration for the Batchcomp method comes from another SHANGRLA-based batch RLA suggested in Section 4 of the paper introducing the ALPHA martingale test~\cite{stark2022alpha}, which we refer to as {\em ALPHA-batch}.

\begin{comment}
This method is inspired by the batch RLA method suggested in ALPHA~\cite{stark2022alpha}, section 4, which we refer to as {\em ALPHA-batch}. Simulations presented in~\Cref{sec:batch RLA simulations} show that the method suggested in this work is more efficient for auditing parliamentary elections. As we have seen, auditing parliamentary elections, as well as most other election systems, essentially reduces to affirming that a number of linear inequalities are all true, where each inequality involves the vote tallies of a few parties. For this reason, it could be reasonable to assume that our method is more efficient than the ALPHA-batch method for additional election systems. 
\end{comment}

The ALPHA-batch method is performed by examining the mean of every assorter over each sampled batch according to its paper-backup ballots. It does not use the reported vote tally of the batches beyond the total number of ballots they contain. The Batchcomp method attempts to improve on ALPHA's efficiency by auditing something slightly different - instead of auditing the mean value of an assorter $a$ over the paper-backup ballots (true results) in a sampled batch, it audits the discrepancy between the mean value $a$ has over a batch according to its reported tally, and the mean value it has over the same batch according to its paper-backup ballots. The values returned by the ALPHA-batch assorters can change drastically from batch to batch, depending on their vote distribution according to the true results. The values the Batchcomp assorters return depend primarily on the accuracy of the reported tally; if two batches with different vote distributions were both counted accurately in the reported results, a Batchcomp assorter will return the same value when applied on each of them. This fact is shown in \Cref{batchcomp assorters advantage}, when discussing the advantages of the Batchcomp assorters.

Recall that before sampling and reading a paper-backup ballots, the ALPHA martingale test guesses the value that each assorter would return on this ballot (this guess is $\eta_k$, for each assorter $a_k$). As explained by Stark~\cite{stark2022alpha}, the audit is more efficient when these guesses are accurate. If each assorter returns a similar value for all batches, as is the case with Batchcomp, then the audit can make guesses which are more accurate. This is the root cause for Batchcomp outperforming ALPHA-batch in the simulations shown in~\Cref{sec:batch RLA simulations}.
\begin{comment}
\moni{I think that you should be careful in claiming that Batchcomp outperforms Alpha-batch; you should say it, but in a qualified manner.}  \end{comment}
    
\subsubsection{Preliminaries and Notation}
\label{sec:Batchcomp prelims}
    As mentioned previously, the Batchcomp method is applicable for any election system that can be reduced to assertions according to the SHANGRLA framework, as described in~\Cref{sec:SHANGRLA}. Through the rest of~\Cref{sec:batch RLA}, we describe how the Batchcomp method applies to some generic elections where the problem of verifying that the reported winners are correct can be done using SHANGRLA. Fix some elections system, a set of ballots $B$ and a partition of these ballots into batches $B_1,...,B_d$. By assuming that we can verify that the reported winners of the elections are correct using the SHANGRLA framework, we are making the following assumption:
    \begin{assumption}
        Assume we have $\ell$ assorters $a_1,..,a_\ell$ such that the reported winners are true iff for all $k \in [\ell]$:
    $$
         \frac{1}{n}\sum_{b\in B}{a_k(b)} > \half .
    $$
    \end{assumption}
    
   Throughout the following sections, we sometimes abuse notation and apply assorters over entire batches. When doing so, $a_k(B_i)$ is defined as the mean of $a_k$ over all ballots in batch $B_i$:
   \begin{align} \label{eq: assorter on batch}
       a_k(B_i)=\frac{1}{|B_i|} \sum_{b\in B_i}{a_k(b)}.
   \end{align}
   
   Before proceeding, note that for each batch has a reported tally, which we know before the audit begins, and a true tally, which we may only learn during the audit. Therefore, each assorter has a reported and true mean value over each batch, which can be calculated from the reported and true tally, respectively. We denote the reported mean of an assorter $a_k$ over a batch $B_i$ as $a_k^{rep}(B_i)$, and its true mean over that batch as $a_k^{true}(B_i)$.
   
   Using this notation, we now define a new assorter $A_k$ for each original assorter $a_k$. Unlike the original assorter $a_k$, this new assorter can only be applied on batches, and not over single ballots. To differentiate it from regular assorters, we refer to these new assorters as {\em batch-assorters}. During a Batchcomp RLA, the audit uses these new batch-assorters instead of the original ones. 

   \begin{definition} \label{def: Batchcomp assorter}
       Let there be some election system with assorters $a_1,...,a_{\ell}$. For each assorter $a_k$, we define the {\em Batchcomp-assorter} $A_k:C^*\rightarrow [0,\infty)$ as:
   $$
       A_k(B_i) := \half + \frac{M_k + a^{true}_k(B_i) - a^{rep}_k(B_i)}{2(w_k - M_k)}.
   $$
   Where $M_k$ is the reported margin of assorter $a_k$ across all batches:
   $$
       M_k := a^{rep}_k(B) - \half,
   $$
   and $w_k$ is the maximal reported value of $a_k$ across all batches:
   $$
       w_k := \max_{j}\{a^{rep}_k(B_j)\}.
   $$
   \end{definition}
    The denominator in the definition of $A_k$ was chosen such that the minimal value $A_k$ may return is 0, as proven in the following claim.
   \begin{claim}
       For any $k\in[\ell]$, $A_k$ is non-negative.
   \end{claim}
   \begin{proof}
       Fix an assorter $a_k$ and its Batchcomp counterpart $A_k$. To check minimal value $A_k$ may return, we examine the minimum of $a^{true}$ and the maximum of $a^{rep}$. Since assorters are non-negative, for any batch $B_i$ we have $a^{true}(B_i)\geq 0$, and since we know the reported tallies of the batches before the audit begins, we can calculate the maximum of $a^{rep}_1$ across all batches, $w_k$. Thus, for any batch $B_i$:
       $$
           A_k(B_i)=\half + \frac{M_k + \overbrace{a^{true}_k(B_i)}^{\geq 0} - \overbrace{a^{rep}_k(B_i)}^{\leq w_k}}{2(w_k - M_k)} \geq \half + \frac{M_k- w_k}{2(w_k - M_k)}=0.
       $$
       Concluding this proof
   \end{proof}

    \paragraph*{Advantages of the Batchcomp Assorters} \label{batchcomp assorters advantage}
   
   These batch-assorters have a couple of useful properties for batch-level RLAs. First, they can be used in-place of the original assorters during a batch-level audit. This is because each of them has a normalized mean greater than $\half$ on all batches (normalized according to the size of the batches) iff the mean of its corresponding regular assorter over all ballots is at least $\half$. A proof of this fact is shown in~\Cref{claim: Batchcomp 1}. 
   
   Additionally, these batch-assorters are useful since the values they return only depend on the accuracy of the reported tallies of the batches, and not on the vote distribution within each batch.
   If the reported tallies were calculated properly, we expect only small discrepancies between the reported and true tallies of all batches. Therefore, auditing these discrepancies would reduce the audit's dependence on the order in which batches are sampled. To understand why these batch-assorters only depend on the tally's discrepancy,  
    we can examine the case in which the reported tally of each batch is equal to its true tally. In this case, any batch $B_i$ necessarily has  $a_k^{true}(B_i)=a_k^{rep}(B_i)$, meaning that for every $i\in[d]$:
    $$
        A_k(B_i) = \half + \frac{M_k}{2(w_k - M_k)}.
    $$
    
    Therefore, if the reported tallies of all batches are accurate, then this batch-assorter returns the same value over all batches. As mentioned in~\Cref{sec:Batchcomp}, this makes the audit agnostic to the order in which we sample the batches and also improves its efficiency. 

    \begin{conclusion}
        For any assorter $a_k$ and its Batchcomp conversion to a batch-assorter $A_k$, $A_k$ operates over the batch-level discrepancy between the reported and true results, and has $A_k(B)>\half$ iff $a_k(B)>\half$.
    \end{conclusion}

\paragraph*{Auditing The Batch-Assertions}

To run a batch-level RLA, we use a very similar method to the ALPHA martingale test described in~\Cref{sec:alpha test}, except we sample batches instead of single ballots. In our method, we use our defined batch-assorters $A_1,...,A_\ell$ instead of the original assorters $a_1,...,a_\ell$. Additionally, we need to consider the fact that different batches have different sizes. For this reason, whenever the audit samples a new batch, each batch is chosen with probability that is proportional to its size. 

Running the ALPHA martingale test also requires setting $u_k$ - a variable which ``guesses" whether the next sampled batch $B_i$ will have $a_k(B_i)< \mu_k$ or $a_k(B_i) > \mu_k$. We must always have $u_k>\eta_k$, and any such choice for $u_k$ yields a valid RLA.
Here, since the algorithm uses the batch-assorters $A_1,...,A_\ell$, we denote these variables as $U_1,...,U_\ell$. If we expect to have $a_k(B_i) > \mu_k$ w.h.p.\, then the algorithm is more efficient when $U_k$ is as small. Here, since accurate reported tallies lead to always having $A_k(B_i)>\mu_k$, we are encouraged to set $u_k$ to be as small as possible. This leads us to set these variables to be, for every $k\in[\ell]$:
$$
 U_k = \half + \frac{M_k + \delta}{2(w_k-M_k)},
$$
where $\delta$ is a very small positive. We must have $\delta > 0$, but other than that any choice of $\delta$ is valid. A smaller value for $\delta$ leads to a more efficient audit when all batches have accurate reported tallies, but decreases efficiency when there are large / malicious errors in the reported tally.~\Cref{sec:choosing delta} further examines how $\delta$ should be chosen.

The algorithm operates by sequentially sampling batches of paper-backup ballots, reading them to get their true tally, and calculating the value of each batch-assorter $A_1,...,A_\ell$ over the sampled batches. For each sampled batch $B_i$ and for each batch-assorters $A_k$, $A_k(B_i)$ is compared to 3 variables which are set before the backup ballots in $B_i$ are read - $\mu_k, \eta_k$ and $U_k$. $\mu_k$ is the algorithm's guess for $A_k(B_i)$ given that the reported winners of the elections are wrong. $\eta_k$ is its guess given that the reported winners are correct. $U_k$ essentially guesses whether $A_k(B_i)\leq mu_k$ or not. By comparing $A_k(B_i)$ to these values, the algorithm updates a p-value which represents the risk-limit with which the $k$th assertion can be approved. The algorithm either approves the $k$th assertion if this p-value is below the risk-limit $\alpha$, or decides to sample another batch. For a more extensive explanation regarding the variable's in this algorithm, see $\Cref{sec:alpha test}$.

The full Batchcomp algorithm operates as follows:

\subsubsection{The Batchcomp Algorithm Description}
\label{sec:Batchcomp algorithm description}
    \begin{enumerate}
    \item \textbf{Initialization:}
    \begin{enumerate}[label*=\arabic*.]
        \item Initialize $\mathcal{K} = [\ell]$, which holds the indexes of assertions we have yet to approve.
        \item Initialize $\mathcal{B}^1=(B_1,B_2,...,B_d)$ and $\mathcal{B}^0=\emptyset$. As the algorithm progresses, $\mathcal{B}^0$ holds the batches which were already audited and $\mathcal{B}^1$ the batches that have yet to be audited. 
        \item For each $k\in \mathcal{K}$ initialize:
        \begin{itemize}
        \item $T_k := 1$.
        \item $\mu_k := \half$.
        \item $\eta_k := \half + \frac{M_k}{2(w_k-M_k)}$.
        \item  $U_k := \half + \frac{M_k + \delta}{2(w_k-M_k)}$.~\Cref{sec:choosing delta} examines how to choose $\delta$, but technically any $\delta > 0$ works.
        \end{itemize} 
    \end{enumerate}
    \item \textbf{Auditing Stage:} As long as $\mathcal{B}^1\neq \emptyset$, perform: 
        \begin{enumerate}[label*=\arabic*.]
            \item\label{step4} Sample a batch from $\mathcal{B}^1$ and denote it as $B_i$. Each batch $B_j$ in $\mathcal{B}^1$ is sampled with probability proportional to its size: $\frac{|B_j|}{\sum_{B_t\in \mathcal{B}^1}{|B_t|}}$.
            \item Remove $B_i$ from $\mathcal{B}^1$ and add it to $\mathcal{B}^0$.
            \item \label{Batchcomp t update} For each $k\in K$, update $T_k$:
            $$
                T_k \leftarrow T_k \left(\frac{A_k(B_i)}{\mu_k} \frac{\eta_k - \mu_k}{U_k - \mu_k} + \frac{U_k - \eta_k}{U_k - \mu_k}\right)
            $$
            \item\label{batchcomp approve step} For each $k\in \mathcal{K}$, if $T_k > \frac{1}{\alpha}$, we have sufficient evidence that the $k$th assertion is true, so set $\mathcal{K} = \mathcal{K}\setminus\{k\}$.
            \item \label{batchcomp variable update step}For each $k\in \mathcal{K}$ update $\mu_k, , \eta_k$ and $u_k$:
            \begin{itemize}
                \item $\mu_k  \leftarrow \frac{\half n - \sum_{B_j\in \mathcal{B}^0}{|B_j|A_k(B_j)}}{n - \sum_{B_j\in \mathcal{B}^0}{|B_j|}}$
                \item $\eta_k \leftarrow max\left\{\half + \frac{M_k}{2(w_k-M_k)}, \mu_k + \epsilon\right\}$
                \item $U_k  \leftarrow max\{U_k, \eta_k + \epsilon\}$
            \end{itemize}
            Where $\epsilon$ is some very small positive meant to ensure that $\mu_k < \eta_k < U_k$. We assume these variables are updated according to the order of their listing above.
            \item \label{batchcomp step 2.6} If $\mu_k < 0$, then the $k$th assertion is true, so remove $k$ from $\mathcal{K}$.  
            \item If $\mathcal{K}=\emptyset$, approve the reported winners.
        \end{enumerate}
     \item \textbf{Output:} If the audit hasn't approved yet, it recounted all batches, and the true winners of the elections are known. 
\end{enumerate}

Note that we changed the update rule of $\eta_k$, the variable that represents our guess for what $A_k$ would return over the next batch we sample, compared to the one presented in the ALPHA martingale set. In ALPHA, $\eta_k$ was set to be the mean value of $a_k$ over all ballots that were not audited yet, given that its mean over all ballots is $a^{rep}(B)$. Here, we set it to be the value that $A_k(B_i)$ would return over a batch with an accurate reported vote tally. Using the old update rule here would still be valid, but could lead to an undesired situation where a counting error in one batch leads us to guess that there are counting errors in the opposite direction in other batches. Our new update rule is based on the principle that we do not expect the next audited batch to have counting errors that are skewed in a specific direction.

\label{sec: Batchcomp proof}
\begin{theorem} \label{thm: Batchcomp theorem}
    For any election system that can be audited using the SHANGRLA framework, for any $\alpha \geq 0$ and for any partition of the ballots into batches, if the reported winners of the elections are wrong, then the Batchcomp RLA approves them with probability of at most $\alpha$.
\end{theorem}
\begin{proof} 
By assuming that we can we can audit an elections system using the SHANGRLA framework, we are essentially making the following assumption:
\begin{assumption}
    We have $\ell$ assorters $a_1,..,a_\ell$ such that the reported winners of the elections are true iff for all $k\in[\ell]$:
    $$
       \frac{1}{n}\sum_{b\in B}{a_k(b)} > \half.
    $$
\end{assumption}
The proof of this theorem relies on two claims. First we show that if there exists $k\in [\ell]$ s.t.\ $a_k(B)<\half$, then we also have $A_k(B)<\half$ (\Cref{claim: Batchcomp 1}). Afterwards, we show that if there exists some $k\in [\ell]$ s.t.\ $A_k(B)<\half$, then the algorithm approves the reported winners w.p.\ of at most $\alpha$ (\Cref{claim: Batchcomp 3}). The combination of these claims means that if the reported winners of the elections are not correct, the algorithm approves them w.p.\ of at most $\alpha$, which is the~\hyperlink{The RLA Guarantee}{RLA guarantee}. 

\begin{claim} \label{claim: Batchcomp 1}
    If there exists $k\in [\ell]$ s.t.\ $a_k(B)<\half$, then we also have $A_k(B)<\half$.
\end{claim}
\begin{proof}
    Assume such a $k$ exists. By the definition of $A_k$ and $M_k$ (\Cref{def: Batchcomp assorter}), we therefore have:
    \begin{align*}
       A_k(B) &= \half + \frac{M_k + a^{true}_k(B) - a^{rep}_k(B)}{2(w_k - M_k)} \\[1.5ex]
       &= \half + \frac{a_k^{rep}(B) - \half + a^{true}_k(B) - a^{rep}_k(B)}{2(w_k - M_k)} \\[1.5ex]
       &= \half + \frac{a^{true}_k(B) - \half }{2(w_k - M_k)} \\
       \intertext{and since $a^{true}_k(B) < \half$:}
       &<
       \half + \frac{\half - \half }{2(w_k - M_k)} \\[1.5ex]
       &= \half,
    \end{align*}
    proving the claim.
\end{proof}

The next claim formally proves the following statement: at any stage during the audit, the following two values are equal, for any Batchcomp assorter $A_k$:
\begin{itemize}
    \item The result, in expectation, of sampling a new batch $B_i$ and calculating $A_k(B_i)$.
    \item The result of applying $A_k$ over a batch which includes all previously unaudited ballots.
\end{itemize}
Calculating the second value here might at first appear problematic, as $A_k$ cannot be applied on any arbitrary collection of ballots; the domain of $A_k$ comprises of batches of ballots which have reported tallies. However, note that all unaudited ballots are exactly the union of all previously unaudited batches. Meaning, we can calculate the reported tally of all of the unaudited ballots by summing the reported tallies of all of the unaudited batches. More generally, this means we can apply a $A_k$ over any union of batches, by summing the reported tallies of these batches to get the reported tally of the union. Given a set of batches $Z\subseteq \{B_1,B_2,...,B_d\}$ where $n_z:=\sum_{B_j\in Z}{|B_j|}$ denotes the total number of ballots in $Z$, the value of $A_k(\cup_{B_i\in Z}{B_i})$ can be calculated as follows:
\begin{align} 
    A(\cup_{B_i\in Z}{B_i}) = \half + \frac{M_k + a^{true}_k(\cup_{B_i\in Z}{B_i}) - a^{rep}_k(\cup_{B_i\in Z}{B_i})}{2(w_k - M_k)} \label{eq: assorter on union}
\end{align}
where:
\begin{align} 
    a^{true}_k(\cup_{B_i\in Z}{B_i}) & = \frac{1}{n_z}\sum_{B_i\in Z}{\sum_{b\in B_i}{a^{true}(b)}}, \label{eq: assorter true union}\\[1.5ex]
    a^{rep}_k(\cup_{B_i\in Z}{B_i}) & = \frac{1}{n_z}\sum_{B_i\in Z}{\sum_{b\in B_i}{a^{rep}(b)}}. \label{eq: assorter rep union}
\end{align}
Note that $a^{rep}(b)$ over a single ballot $b$ is not well defined, as we do not have reported results over single ballots; we only have them over entire batches. However, here we only use sums of $a^{rep}(b)$ over all ballots in a batch, and we do know the values of these sums.

We can now proceed to the formally stating and proving this claim:

\begin{claim} \label{claim: Batchcomp 2}
        For any assorter $a$ and its conversion to a batch-comp assorter $A$, and for any $Z\subseteq\{B_1,...,B_d\}$, we have $\bbE_{B_i\sim Z}[A(B_i)] = A(\cup_{B_i\in Z}{B_i})$, where $B_i\sim Z$ means that $B_i$ is sampled from $Z$ w.p.\ proportional to $|B_i|$.
\end{claim}
\begin{proof}
As before, denote the total number of ballots in $Z$ as $n_z$. The equality now follows from the definition of the batch-assorter $A$:
        \begin{align*}
            \bbE_{B_i\sim Z}[A(B_i)] &= \sum_{B_i\in Z}{\frac{|B_i|}{n_z}A(B_i)}, \\
            \intertext{inserting the definition of $A$:}
            &= \frac{1}{n_z}\sum_{B_i\in Z}{|B_i| \left( \half + \frac{M+a^{true}(B_i)-a^{rep}(B_i)}{2(w-M)}\right)} \\[1.5ex]
            &= \half + \frac{M}{2(w-M)} + \frac{1}{2(w-M)n_z}\sum_{B_i\in Z}{|B_i| \left(a^{true}(B_i)-a^{rep}(B_i)\right)}, \\
            \intertext{using the definition of $a$ when applied on batches (see~\eqref{eq: assorter on batch}):}
            &= \half + \frac{M}{2(w-M)} + \frac{1}{2(w-M)n_z}\sum_{B_i\in Z}{|B_i| \left(\frac{1}{|B_i|}\sum_{b\in B_i}{a^{true}(b)} - \frac{1}{|B_i|}\sum_{b\in B_i}{a^{rep}(b)}\right)}, \intertext{slightly re-arranging the sums yields:}
            &= \half + \frac{M}{2(w-M)} + \frac{1}{2(w-M)}\frac{1}{n_z}\sum_{B_i\in Z}{\sum_{b\in B_i}{(a^{true}(b) - a^{rep}(b))}}, \\
            \intertext{and using~\eqref{eq: assorter true union} and~\eqref{eq: assorter rep union}:}
            &= \half + \frac{M}{2(w-M)} + \frac{1}{2(w-M)}(a^{true}(\cup_{B_i\in Z}{B_i}) - a^{rep}(\cup_{B_i\in Z}{B_i})) \\[1.5ex]
            &= \half + \frac{M + a^{true}(\cup_{B_i\in Z}{B_i}) - a^{rep}(\cup_{B_i\in Z}{B_i})}{2(w-M)}, \\
            \intertext{finally, by~\eqref{eq: assorter on union}:}
            &= A(\cup_{B_i\in Z}{B_i}),
        \end{align*}
        concluding the proof of this claim.
    \end{proof}
    
    \begin{claim} \label{claim: Batchcomp 3}
        If there exists $k\in [\ell]$ s.t.\ $A_k(B)<\half$, then the algorithm will approve the reported winners w.p.\ of at most $\alpha$.
    \end{claim}
    \begin{proof}
        Assume such a $k$ exists and w.l.o.g.\ let it be $k=1$. The index $1$ cannot be removed from $\mathcal{K}$ in step~\ref{batchcomp step 2.6}, because then:
        $$
        \mu_k < 0 \Longrightarrow \half n - \sum_{B_j\in \mathcal{B}^0}{A_k(B_j)|B_j|} < 0 \Longrightarrow \half < \sum_{B_j\in \mathcal{B}^0}{A_k(B_j)|B_j|} \leq  A_1(B),
        $$  
        contradicting our assumption. This means assertion 1 ($A_1(B)>\half$) can only be approved in step~\ref{batchcomp approve step}. We now show that this happens w.p.\ of at most $\alpha$.
        
        As we do in the proof of~\Cref{thm: alpha martingale test}, denote by $T^0_1, T^1_1, ..., T^d_1$ the values of $T_1$ after each sampled batch. Similarly, let $\mu^1_1, \mu^2_1,...,\mu^d_1$, $\eta^1_1, \eta^2_1,...,\eta^d_1$ and $U^1_1, U^2_1,...,U^d_1$ be the values that $\mu_1$, $\eta_1$ and $U_1$ have when sampling each batch. To prove this claim, it suffices to prove $T^0_1,T^1_1,...,T^d_1$ is a non-negative supermartingale. If we can do so, then the algorithm approves the results w.p.\ of at most $\alpha$ by Ville's inequality, as explained while proving~\Cref{thm: alpha martingale test}.
        
        First, we show $A_1(B_i)$ is non-negative for any $i\in [t]$. Towards this goal, recall the update rule of $T^i_1$ in step~\ref{Batchcomp t update}:
        $$
            T_1^i = T^{i-1}_1 \left( \frac{A_1(b_i)}{\mu^i_1} \frac{\eta^i_1 -\mu^i_1}{u^i_1 - \mu^i_1}  + \frac{U^i_1 - \eta^i_1}{U^i_1 - \mu^i_1}\right).
        $$
        To show that $T^i_1$ is non-negative for any $i$, it suffices to prove that $A_1(B_i)\geq 0$, since we always have $0 \leq \mu_1^i \leq \eta_1^i \leq U_1^i$. We now do so:
        \begin{align*}
            A_1(B_i) = \half + \frac{M_1 + \overbrace{a^{true}_1(B_i)}^{\geq 0} - \overbrace{a^{rep}_1(B_i)}^{\leq w_1}}{2(w_1 - M_1)}
                \geq \half + \frac{M_1 + 0 - w_1}{2(w_1 - M_1)}
                = \half - \half
                = 0.
        \end{align*}
        \begin{conclusion}
            $T^i_1$ is non-negative.
        \end{conclusion}  
        
        What remains is to show that for any $i\in [d]$ we have $\bbE[T^i_1\,|\,T^{0}_1,...,T^{i-1}_1]\leq T^{i-1}_1$. Towards this purpose, fix some $i\in [d]$. For simplicity, we use $\mathcal{B}^0$ and $\mathcal{B}^1$ to denote the sets of batches $\mathcal{B}^0$ and $\mathcal{B}^1$ were when sampling the $i$th batch to audit. 
        
        In every iteration of the algorithm, in step~\ref{batchcomp variable update step}, $\mu_1$ is defined to be the value that $A_1(\cup_{B_j\in \mathcal{B}^1}{B_j})$ would have, given that $A_1(B)=\half$. This means that since $A_1(B)\leq \half$, we have $A_1(\cup_{B_j\in \mathcal{B}^1}{B_j})\leq \mu^i_k$. By~\Cref{claim: Batchcomp 2}, we therefore have $\bbE_{B_j \sim \mathcal{B}^1}[A_1(B_j)]\leq \mu^i_1$, where $B_j\sim \mathcal{B}^1$ means that we sample a batch from $\mathcal{B}^1$ w.p.\ that is proportional to that batch's size. 
        By the same reasoning as~\eqref{eq: t and b} in the proof of~\Cref{thm: alpha martingale test}:
        \begin{align*}
            \bbE[T^i_1\,|\,T^{i-1}_1,...,T^{0}_1]
            &= \bbE[T^i_1\,|\,B_1,...,B_{i-1}], \\
            \intertext{now, by the update rule of $T_1$ from step~\ref{Batchcomp t update}:}
            &= T^{i-1}_1 \left( \frac{\bbE_{B_j \sim \mathcal{B}^1}[A_1(B_j)]}{\mu^i_1} \frac{\eta^i_1 -\mu^i_1}{U^i_1 - \mu^i_1}  + \frac{U^i_1 - \eta^i_1}{U^i_1 - \mu^i_1}\right), \\
            \intertext{and since $\bbE_{B_j \sim \mathcal{B}^1}[A_1(B_j)]\leq \mu^i_1$:}
            &\leq T^{i-1}_1 \left( \frac{\mu^i_1}{\mu^i_1} \frac{\eta^i_1 -\mu^i_1}{U^i_1 - \mu^i_1}  + \frac{U^i_1 - \eta^i_1}{U^i_1 - \mu^i_1}\right) \\[1.5ex]
            &= T^{i-1}_1 \left( \frac{\eta^i_1 -\mu^i_1}{U^i_1 - \mu^i_1}  + \frac{U^i_1 - \eta^i_1}{U^i_1 - \mu^i_1}\right) \\[1.5ex]
            &= T^{i-1}_1 \left( \frac{U^i_1-\mu^i_1}{U^i_1 - \mu^i_1}\right) \\
            &= T^{i-1}_1.
        \end{align*}
        This concludes the proof of this claim, since we have shown that $T^{1}_1,..., T^{d}_1$ is a non-negative supermartingale, meaning that the audit approves the assertion ($\sum_{b\in B}{a_1(b)}>\half$) w.p.\ of at most $\alpha$.
    \end{proof}
    As explained previously, the combination of these 3 claims concludes the proof of this theorem.
\end{proof}

\subsubsection{Choosing \texorpdfstring{$\delta$}{Lg}}
\label{sec:choosing delta}
As explained in~\Cref{sec:Batchcomp prelims}, for every assorter $a_k$ and its Batchcomp counterpart $A_k$ we initialize:
   $$
     U_k = \half + \frac{M_k + \delta}{2(w_k-M_k)},
   $$

where $M_k=a^{rep}(B)-\half$ is the reported margin of $a_k$ and $\delta>0$. Different choices for $\delta$ all produce valid RLAs (any choice maintains the~\hyperlink{The RLA Guarantee}{RLA guarantee}), but under certain conditions, certain values of $\delta$ yield more efficient audits. This section attempts to give intuition regarding the ideal choice of $\delta$. Generally, the more we expect the reported vote tallies of the different batches to be accurate, the smaller $\delta$ should be. We show this by comparing $\mu_k$ to the expected value of a Batchcomp assorter on the next batch to be sampled.

\begin{claim}
    During a Batchcomp RLA, if the next sampled batch $B_j$ satisfies $A_k(B_j) \geq \mu_k$ for some batch-assorter $A_k$, then choosing a smaller $U_k$ increases the audit's efficiency, and vice-versa; if $A_k(B_j) < \mu_k$, then setting a larger $U_k$ increases the audit's efficiency.  
\end{claim}
\begin{proof}
    Examine some Batchcomp assorter $A_k$. First, note that auditing the assorter is more efficient, meaning it requires examining fewer ballots, the more significantly $T_k$ grows per batch. This is because the audit approves assertion $k$ when $T_k > \frac{1}{\alpha}$. Therefore, it suffices to show that if $A_k(B_j) \geq \mu_k$, then $T_k$ grows more significantly when $U_k$ is small, and vice-versa.
    
    Towards this purpose, denote the next audited batch as $B_i$. To prove this claim, we take the derivative by $U_k$ of the update rule of $T_k$ in step~\ref{Batchcomp t update} of the Batchcomp algorithm:
    $$
        T_k \leftarrow T_k\left(\frac{A_k(B_i)}{\mu_k}\frac{\eta_k-\mu_k}{U_k-\mu_k}+\frac{U_k - \eta_k}{U_k - \mu_k}\right).
    $$
    Taking its derivative by $U_k$ results in:
    \begin{align*}
        & T_k\left(-\frac{A_k(B_i)}{\mu_k}\frac{\eta_k-\mu_k}{(U_k-\mu_k)^2} + \frac{1}{U_k-\mu_k} - \frac{U_k - \eta_k}{(U_k-\mu_k)^2}\right) \\[1.5ex]
        =& \frac{T_k}{(U_k - \mu_k)^2} \left(-\frac{A_k(B_i)}{\mu_k}(\eta_k-\mu_k) + U_k - \mu_k - U_k + \eta_k\right) \\[1.5ex]
        =& \frac{T_k}{(U_k - \mu_k)^2} \left(-\frac{A_k(B_i)}{\mu_k}(\eta_k-\mu_k) - \mu_k + \eta_k\right) \\[1.5ex]
        =&\underbrace{T_k\frac{\eta_k - \mu_k}{(U_k-\mu_k)^2}}_{>0}\left(1 - \frac{A_k(B_i)}{\mu_k}\right).
    \end{align*}
    
    Where the term on the left (above the underbrace) is positive since $T_k$ is positive, and since we always have $U_k>\eta_k>\mu_k > 0$. We can observe that if $A_k(B_j)> \mu_k$, this derivative is negative, meaning that choosing a smaller value for $U_k$ causes $T_k$ to increase more significantly. If $A_k(B_j) < \mu_k$, then the opposite is true. This concludes the proof of this claim.
\end{proof}

By this claim, if we expect to have $A_k(B_j)> \mu_k$ for all batch-assorters and batches, we should choose a smaller $\delta$, and vice versa. When using the Batchcomp assorter, we have:
$$
    A_k(B_i) = \half + \frac{M_k + a^{true}_k(B_i) - a^{rep}_k(B_i)}{2(w_k - M_k)}
$$
And $w_k>M_k > 0$ by the definition of $M_k$. Therefore, as long as the discrepancies between the reported and true vote counts are small, we expect to consistently have $A_k(B_i)\geq \mu_k$, meaning we should choose a smaller $\delta$. To get $A_k(B_i)<\mu_k$, we would need to have $a^{true}_k(B_i) - a^{rep}_k(B_i) > M_k$, meaning that the discrepancy in vote counts, as it relates to the assorter $a$, is greater than its reported margin. If the margin isn't extremely small, and the errors in the vote count are uncorrelated and rare, this is very unlikely to happen. This suggests that if the counting mistakes aren't malicious, we should choose a very small $\delta$.

If we choose a very small $\delta$ and the counting mistakes are malicious, the audit might become inefficient, but it will still fulfill the~\hyperlink{The RLA Guarantee}{RLA guarantee}. This could encourage choosing a very small value for $\delta$, since it would only make the audit inefficient if it's likely that the vote counting was malicious. If there were correlated mistakes in the vote counting (such as errors that are skewed against a certain party), we would not mind a more exhaustive audit. Meanwhile, if the counting mistakes are ``honest", meaning the chance of any ballot to be misinterpreted is equal, we would like the audit to examine as few batches as possible.

\begin{conclusion}
    A Batchcomp RLA is more efficient when $\delta>0$ is very small, as long as the vote tallying is not done maliciously. 
\end{conclusion}

On the extreme case, we could technically choose $\delta \rightarrow 0^+$. This essentially means that $\eta_k \approx U_k$ for every $k\in[\ell]$ and the update rule of $T_k$ becomes:
$$
    T_k \leftarrow T_k\left(\frac{A_k(B_i)}{\mu_k}\frac{\eta_k-\mu_k}{U_k-\mu_k}+\frac{U_k - \eta_k}{U_k - \mu_k}\right) \approx T_k \frac{A_k(B_i)}{\mu_k}.
$$
This inspires a simplified batch-level RLA, which is very efficient as long as the counting-errors are random: we could ditch $\eta_k$ and $U_k$ entirely, and just update each $T_k$ by $T_k\leftarrow T_k\frac{A(B_i)}{\mu_k}$. If the counting mistakes are random, this method appears to perform just as well as the full batch-comp RLA, according to simulations (for brevity, the full plots are not included here). This is true even when auditing election results with assertions which have very tight margins, meaning that even small random mistakes might lead to some cases where $A_k(B_i)< \mu_k$ for some $k$. However, note that this method could completely fail if we ever get $A_k(B_i)=0$. Having a batch-assorter and batch with $A_k(B_i)=0$ would mean that:
\begin{align*}
    A(B_i)=0 
    & \Longrightarrow \half + \frac{M_k + a^{true}_k(B_i) - a^{rep}_k(B_i)}{2(w_k - M_k)} = 0 \\[1.5ex]
    & \Longrightarrow \half = \frac{a^{rep}_k(B_i) - a^{true}_k(B_i) - M_k}{2(w_k - M_k)} \\[1.5ex]
    & \Longrightarrow w_k-M_k = a^{rep}_k(B_i) - a^{true}_k(B_i) - M_k \\[1.5ex]
    & \Longrightarrow w_k = a^{rep}_k(B_i) - a^{true}_k(B_i),
\end{align*}
which indicates that $a_k^{rep}(B_i)=w_k=\max_{j}\{a^{rep}_k(B_j)\}$ and $a^{true}_k(B_i)=0$. This means that reportedly, this batch is the ``the best one" for assorter $a$, while in reality its as bad as it can get. Such a scenario is very unlikely to occur unless the vote counting is malicious, and in that case, we wouldn't mind a full manual recount anyway.

\begin{conclusion}
    As long as the counting errors are not malicious, the update rule of $T_k$ in the Batchcomp RLA could be changed to simply $T_k \leftarrow T_k \frac{A_k(B_j)}{\mu_k}$. This simplifies the audit but adds a small probability of an unnecessary full manual recount of the paper-backup ballots.
\end{conclusion}

\newpage
\section{Israeli Knesset Elections RLA}
\label{sec:Knesset RLA}
This section describes how to perform an RLA to verify results of the Israeli Knesset elections. This method can be used in Israel currently to verify the initial hand-count of the votes, which is not performed centrally - each polling place independently tallies its own ballots. This method can also become useful if, in the future, the vote tallying will be done by some electronic means, such as an optical reader. In such cases, this method could confirm that the winners outputted by the electronic vote tabulation system are highly likely to be correct. 

In the context of the Knesset elections, we wish to verify that the correct Knesset members were elected. Since the Knesset elections use a closed party-list, meaning that each party submits a ranked list of its candidates ahead of the elections, we only need to verify that each party receives the correct number of seats. Again, we do not mind if the vote tallies are not completely accurate, as long as the correct number of seats was given to each party.

Before moving to explain the how the Knesset elections work, we define some notation. Let $P$ be the set of all parties running in the elections, and let $S$ be the number of available seats. For every party $p\in P$, let $v^{true}(p)$ denote the true number of votes $p$ received, according to the paper-backups ballots. Similarly, denote the true number of invalid votes as $v^{true}(invalid)$ and the true number of valid votes as $v^{true}(valid)$.

\subsection{Knesset Elections Method}
\label{sec: knesset elections method}
The Knesset is the Israeli parliament and its sole legislative authority. It consists of $S:=120$ members who are elected according to closed party-list proportional representation. Before each election cycle, each party submits a ranked list of its candidate. On polling day, each voter votes for a single party, and parties receive seats in proportion to the share of the votes they received. The seats each party wins are given to the top-ranked candidates in the party's list. 
    
Allocating Knesset seats to the various parties is done as follows~\cite{knessetElectionMethod}:
    
\subsubsection*{Electoral Threshold}
    
In the Knesset elections, only parties which receive at least 3.25\% of the valid votes are eligible to win seats. We denote this threshold as $t:=0.0325$.
    
\subsubsection*{Seat allocation}
\label{sec: knesset seat allocation}
    
The allocation of seats is done according to the D'Hondt method, a highest average method, and can be formulated in multiple ways. The description here, which was suggested previously by Gallagher~\cite{gallagher1991proportionality}, lends itself more naturally to the SHANGRLA framework. Meaning, it simplifies the process of developing SHANGRLA assertions which are all true iff the reported winners of the elections are correct. 
      
To find how many seats each party deserves, we do the following:
    \begin{enumerate}
        \item Imagine a table with a row for each party which is above the threshold (meaning $v^{true}(p) \geq tv^{true}(valid)$), and $S$ columns. In the cell of party $p$ and column $s$, we write ${v^{true}(p)}/{s}$. All cells are initially uncolored.
        \item Color (or mark) the $S$ cells with the largest values in the table.
        \item The number of colored cells a party has in its row is the number of seats it should receive.
    \end{enumerate} 
    Note that the values in each row are monotonically decreasing, so each row would be fully colored up to a certain column, and not colored for the rest of it.
    
    For example, in an election with 3 parties and 10 seats, the results would be as follows:
    \begin{center}
        \includegraphics[scale=0.65]{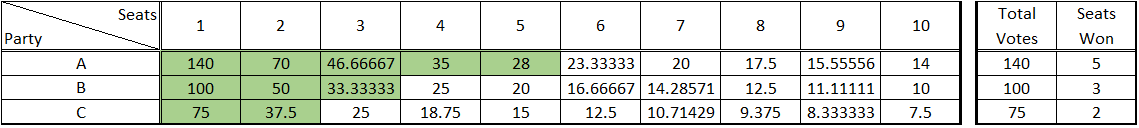}    
    \end{center}
    
\subsubsection*{Apparentment (Also Known as Electoral Alliances or Heskemei Odafim)}
 Two parties may sign prior to the the electoral day an apparentment agreement, which may allow one of them to gain an extra seat. If two parties sign an apparentment agreement, and only if both are above the threshold, they essentially unite to a single allied party during the seat allocation stage. Then, the number of seats their alliance received is split between them according to the same seat allocation method. Meaning, for each apparentment between two parties, we write another table with only these two parties, and allocate the number of seats they won together between them, using the same method described above. 
    
If one of the parties in the apparentment is below the electoral threshold while using only its own votes, the apparentment is ignored. Each party may only sign a single apparentment agreement.

\subsection{Designing Assorters}
\label{sec:knesset assorters}
This section presents assorters that can be used to perform an RLA for the Knesset elections, using the SHANGRLA framework. We begin by presenting 3 conditions which all hold true if and only if the reported winners of the election are correct. We then proceed to develop assorters for each of these conditions, such that the assorters all have a mean greater than $\half$ if and only if these conditions all hold true. 
    
    \begin{theorem} \label{thm: Knesset}
        Let $s^{rep}(p)$ and $s^{true}(p)$ be the reported and true number of seats that a party $p$ won in a Knesset elections, respectively. We have it that $s^{rep}(p)=s^{true}(p)$ for every party $p\in P$, if and only if these 3 conditions all hold true:
        \begin{enumerate}
            \item Every party who is reportedly above the electoral threshold, is truly above the electoral  threshold.
            \item Every party who is reportedly below the electoral threshold, is truly below the electoral  threshold.
            \item For every two parties $p_1, p_2$ who are reportedly above the electoral threshold, the condition $\left(s^{rep}(p_1) \geq s^{true}(p_1)\right) \vee \left(s^{rep}(p_2) \leq s^{true}(p_2)\right)$ is true.
        \end{enumerate}
    \end{theorem}
    \begin{proof}
        Fix some reported and true tallies for the elections, and calculate the number of seats each party reportedly and truly won according to these tallies. If the reported and true number of seats each party won are equal, then the 3 conditions above hold true trivially. 
        
        Otherwise, assume there is a discrepancy between the reported and true seat allocation. Under this assumption, there is at least one party who won more seats according to the reported result compared to the true results, which we denote as $p_{r}$, and at least one party who won less seats according to the reported results compared to the true results, which we denote as $p_{t}$. We now show that at least one of the three conditions above are violated.

        If $p_{r}$ is not truly above the electoral threshold, then Condition 1 is violated, as it receives seats according to the reported tally, which indicates that it is reportedly above the threshold. Similarly, if $p_{t}$ is below the threshold according to the reported tally, then Condition 2 is violated. Otherwise, both parties are reportedly and truly above the threshold.

        If both parties are reportedly above the electoral threshold, then $p_{t}$ reportedly won less seats than it truly deserves, meaning that $s^{rep}(p_t) < s^{true}(p_t)$. Similarly, we have $s^{rep}(p_r)>s^{true}(p_r)$. This violates Condition 3 and concludes our proof.
\end{proof}

    Next, we present SHANGRLA assertions which confirm that each of these 3 conditions are true. These assertions and their corresponding assorters can be used to perform an RLA for Knesset elections using the Alpha martingale test, or to perform a batch-level RLA using the Batchcomp method.

    Throughout this section, denote by $v^{true}(p)$ and $v^{rep}(p)$ the true and reported number of votes a party $p$ received, respectively. Similarly, denote the reported and true number of seats a party $p$ won by $s^{rep}(p)$ and $s^{true}(p)$.
    
    \subsubsection{Above Threshold Assertion}
    \label{sec: above assorter}
    The role of this assertion is to check that Condition 1 holds, i.e.\  that a party who reportedly received more votes than the electoral threshold, is indeed above the threshold. Verifying this is equivalent to ensuring that every party who is reportedly above the threshold has in fact received at least a $t$-share of the valid votes.  Stark~\cite{stark2020sets} has previously suggested a SHANGRLA assertion for this condition exactly - verifying whether a candidate or party won a certain share of the valid votes (super-majority). For the Knesset elections, this assertion is used once per each party who is reportedly above the threshold.

    For every party $p$ who reportedly is above  the electoral threshold, we add a single assertion to the set of assertions we audit: $\frac{1}{n}\sum_{b\in B}{a^{above}(b)} > \half$, using the following assorter:

    \begin{definition}
        An above threshold assorter, which verifies that a party $p$ is truly above the electoral threshold, is defined as:
        \begin{equation*}
        a_p^{above}(b) := \begin{cases}
            \frac{1}{2t} & \text{if $b$ is for party $p$} \\
            \half & \text{if $b$ is invalid}\\
            0 & \text{otherwise}
        \end{cases}
    \end{equation*}    
    \end{definition}
    \label{def: threshold assorter}
    The mean of this assorter over all ballots is:
    $$
        \frac{1}{n}\sum_{b\in B}{a^{above}(b)}= \frac{1}{n}\left(\half v^{true}(\text{invalid}) + \frac{1}{2t}v^{true}(p)\right) = \frac{1}{2n}\left(v^{true}(\text{invalid}) + \frac{1}{t}v^{true}(p)\right).
    $$
    And this is equal or greater than $\half$ iff:
    $$
        v^{true}(\text{invalid}) + \frac{1}{t}v^{true}(p) \geq n 
        \Longleftrightarrow
        v^{true}(p)\geq t(n-v^{true}(invalid))   
    $$
    as needed.
    \subsubsection{Below Threshold Assertion}
    \label{sec: below assorter}

    In order to check Condition 2, that every party who is reportedly below the threshold is truly below it, we need to confirm that every such party received less than $t$ of the valid votes. This is equivalent to verifying that all other parties received at least $1-t$ of the valid votes. Therefore, we can use a similar assorter to the one from~\Cref{def: Batchcomp assorter}. For every party $p$ who is reportedly below the electoral threshold, we add the assertion $\frac{1}{n}\sum_{b\in B}{a^{below}(b)}>\half$ to the set of assorters we audit, where the assorter $a^{below}$ is defined as:
    \begin{definition}
        A below threshold assorter, which verifies that a party $p$ is truly below the electoral threshold, is defined: 
        \begin{equation*}
        a^{below}_p(b) := \begin{cases}
            0 & \text{if $b$ is for party $p$} \\
            \half & \text{if $b$ is invalid} \\
            \frac{1}{2(1-t)} & \text{otherwise}
        \end{cases}
    \end{equation*}
    \end{definition}
    
    \subsubsection{Move-Seat Assertion}
    \label{sec: move-seat assorter}
    The role of this assertion is to verify that Condition 3 holds, i.e.\  that for every two parties $p_1, p_2$ who are reportedly above the electoral  threshold, $\left(s^{rep}(p_1) \geq s^{true}(p_1)\right) \vee \left(s^{rep}(p_2) \leq s^{true}(p_2)\right)$ is true. An assertion for this condition was previously suggested by Blom et al., section 5.2~\cite{blom2021assertion}, but is developed here independently. We begin by reducing the problem of verifying Condition 3 to the problem of confirming that some linear inequality is true (\Cref{claim: move-seat assorter}). From this inequality, we develop an assorter which verifies that Condition 3 is true using the method described in~\Cref{sec: finding assorters}.

    \begin{claim} \label{claim: move-seat assorter}
        For any two different parties $p_1, p_2$ who are reportedly above the threshold, we have: 
        $$
            \left(\frac{v^{true}(p_1)}{s^{rep}(p_1)+1} < \frac{v^{true}(p_2)}{s^{rep}(p_2)} \right) \Longrightarrow \left(\left(s^{rep}(p_1) \geq s^{true}(p_1)\right) \vee \left(s^{rep}(p_2) \leq s^{true}(p_2)\right)\right).
        $$
        Meaning that confirming that the inequality on the left is true also confirms that Condition 3 is true. 
    \end{claim}
    \begin{proof}
       
Fix some reported and true election results, and two different parties who are reportedly above the threshold $p_1,p_2$. Say that we allocate seats according to the true tally, using the colored table method  described in~\Cref{sec: knesset elections method}. Note that in this described seat-allocation table, each party $p$ has exactly its first $s^{true}(p)$ cells colored. 
        
        Examine the case where the condition on the right side of the claim is false, meaning that the condition $\left(s^{rep}(p_1) < s^{true}(p_1)\right) \wedge \left(s^{rep}(p_2) > s^{true}(p_2)\right)$ is true. If this condition is true, then the cell at index $[p_1, s^{rep}(p_1)+1]$ of the table is colored, as $p_1$ wins more than $s^{rep}(p_1)$ according to the true results, while the cell at $[p_2, s^{rep}(p_2)]$ is not, as $p_2$ wins less than $s^{rep}(p_2)$ according to the true results. Therefore, to show that $\left(s^{rep}(p_1) < s^{true}(p_1)\right) \wedge \left(s^{rep}(p_2) > s^{true}(p_2)\right)$ is false, it suffices to show that if the cell at index $[p_2, s^{rep}(p_2)]$ is not colored, then the cell at $[p_1, s^{rep}(p_1)+1]$ is not colored either.

        Recall that the colored cells in the table are the ones which hold the $S$ largest values. Thus, to show that if the cell at index $[p_2, s^{rep}(p_2)]$ is not colored, then the cell at $[p_1, s^{rep}(p_1)+1]$ is not colored, it suffices to show that the value at $[p_1, s^{rep}(p_1)+1]$ is smaller than the value at $[p_2, s^{rep}(p_2)]$. Meaning, to show that Condition 3 is true regarding $p_1, p_2$, it suffices confirm that:
         \begin{align} \label{eq:condition 3 eq 1} 
            \frac{v^{true}(p_1)}{s^{rep}(p_1)+1} < \frac{v^{true}(p_2)}{s^{rep}(p_2)}.
        \end{align}
        The smaller term here is the value at index  $[p_1, s^{rep}(p_1)+1]$, while the larger term is the value at index $[p_2, s^{rep}(p_2)]$. This completes the proof of this claim.        
    \end{proof} 
    By this claim, to verify that Condition 3 holds for two parties $p_1, p_2$, it is sufficient to verify that~\eqref{eq:condition 3 eq 1} holds. At first glance, it may appear that Condition 3 can be true while~\eqref{eq:condition 3 eq 1} is not. This can be problematic, as there may be elections where Condition 3 is true regarding some two parties, without~\eqref{eq:condition 3 eq 1} being true. Meaning, if we develop an assertion for verifying~\eqref{eq:condition 3 eq 1}, we may encounter election results where Condition 3 is true but this assertion is false. An audit which uses this assertion might unnecessarily require a full manual recount, despite the reported winners being correct. 
    
    A more careful examination, however, shows that such a scenario is not possible. If the reported winners of the elections are correct, then for any two parties who are above the electoral threshold $p_1, p_2$, party $p_2$ truly wins $s^{rep}(p_2)$ seats while party $p_1$ truly wins less than $s^{rep}(p_2)+1$ seats. Thus, the value at cell $[p_2, s^{rep}(p_2)]$ in the imaginary table from~\Cref{claim: move-seat assorter} is colored, while the cell at $[p_1, s^{rep}(p_1)+1]$ is not, meaning that~\eqref{eq:condition 3 eq 1} is true.

    \begin{conclusion}
        If the reported winners of the election match the true results, then~\eqref{eq:condition 3 eq 1} is true for any two parties $p_1,p_2$ who are reportedly above the electoral threshold. 
    \end{conclusion}
    Thus, to verify Condition 3, we can develop assertions that verify~\eqref{eq:condition 3 eq 1}. Using such assorters, we are guaranteed that if the reported winners of the elections are correct, these assertions will all be true. Meaning, the assertions we audit are both sufficient and necessary conditions for the winners of the elections to be correct. 
    
    We now move to developing the assorters of these assertions. Fix two parties $p_1,p_2\in P$ who are reportedly above the threshold. We wish to find a non-negative function $a_{p_1,p_2}$ such that~\eqref{eq:condition 3 eq 1} is equivalent to a SHANGRLA assertion of the form:
    \begin{align} \label{eq:condition 3 goal}
        \frac{1}{|B|}\sum_{b \in B}{a_{p_1,p_2}(b)} > \half.
    \end{align}
    This is achieved by using the method described in~\Cref{sec: finding assorters}, which converts linear inequalities to SHANGRLA assertions. To use this method, we re-arrange~\eqref{eq:condition 3 eq 1} as a linear inequality over the tallies of the various parties:
    \begin{align}
        \label{eq:condition 3 eq 2} 
            \frac{1}{s^{rep}(p_2)} v^{true}(p_2) - \frac{1}{s^{rep}(p_1)+1} v^{true}(p_1) > 0 .
    \end{align}
    We can now apply the method from~\Cref{sec: finding assorters} to find a non-negative function $a_{p_1,p_2}$ such that~\eqref{eq:condition 3 goal} and~\eqref{eq:condition 3 eq 2} are equivalent. This results in the following definition for $a_{p_1,p_2}$:
    \begin{definition} \label{def: move seat assorter}
        An assertion which verifies that Condition 3 is true for two parties $p_1,p_2$ who are reportedly above the threshold, is $\frac{1}{|B|}\sum_{b \in B}{a^{move}_{p_1,p_2}(b)} > \half$ where:
    
    \begin{equation*}
        a^{move}_{p_1,p_2}(b) := 
        \begin{cases}
        \half + \frac{s^{rep}(p_1)+1}{2s^{rep}(p_2)} & \text{if $b$ is for $p_2$} \\
        0 & \text{if $b$ is for $p_1$} \\
        \half & \text{otherwise}
        \end{cases}
    \end{equation*}
    \end{definition}
    And we need add two instances of this assertion to the audit for every two parties who are reportedly above the threshold, one using $a^{move}_{p_1,p_2}$ and one using $a^{move}_{p_2,p_1}$.

\subsubsection{Handling Apparentments}

The assertions above ignore the existence of apparentments. To handle them, we can simply treat each two allied parties who are reportedly above the electoral threshold as a united party when adding move-seat assertions. Afterwards, we also need to verify that the seat allocation between every two allied parties is correct. To do so, we can add two move-seat assertions (one in each direction) for every two allied parties who are reportedly above the electoral threshold.

\subsection{Simulations Based on Recent Elections}
\label{sec:batch RLA simulations}

We describe the results of simulating the execution of a  batch-comparison RLA over three different elections for the Knesset, all conducted between 2019 and 2021. The election results used in this section are the true election results, as reported by the Israeli Central Elections Committee\footnote{See \url{https://votes22.bechirot.gov.il/}, \url{https://votes23.bechirot.gov.il/}, \url{https://votes24.bechirot.gov.il/}.}. The partition of ballots to batches is also done according to the real election results, and each batch contains ballots from a single polling place. A typical batch contains between 250 and 550 ballots, with the average across the three election cycles being 386. The audit uses assertions as described in~\Cref{sec:knesset assorters}, and converts their assorters to Batchcomp assorters as described in~\Cref{sec:Batchcomp prelims} Finally, the Batchcomp method described in~\Cref{sec:Batchcomp algorithm description} is used to perform the RLA over these Batchcomp assorters.

We begin by showing the simulated performance of the Batchcomp algorithm on each of the 3 election cycles, {\em assuming all vote tallies are accurate}. For each cycle, we compare Batchcomp with the ALPHA-batch algorithm described in section 4.2 of~\cite{stark2022alpha} using the same SHANGRLA assertions from~\Cref{sec:knesset assorters}.
Afterwards, we move to examine the Batchcomp algorithm's efficiency when there is a small, independent probability (0.01) that each ballot is misread during the reported count. 

\begin{comment}
Additionally, we examine the accuracy of the prediction for the number of batches each assorter requires, as generated in~\Cref{sec:estimating number of batches}.  
\end{comment}

\subsubsection{Technical Details}

The following results detail the execution of the this suggested batch-level RLA with a risk-limit of $\alpha = 0.05$ and with $\delta$ set to $10^{-10}$. The latter was determined after some experimentation - lower choices for $\delta$ do not improve efficiency when the reported results are accurate, while higher values reduce the audit's efficiency.

For all plots and tables, the number of audited ballots by each method is averaged across 10 simulations. An examination of these simulations shows that the number of ballots required to approve each assertion has very low standard deviation. The mean standard deviation, across all assertions in all elections, is 1,898, while the maximal standard deviation across all assertions is 6,041.

The code used for these simulations was written in Python, and is available in \url{https://github.com/TGKar/Batch-and-Census-RLA}.

\subsubsection{Results with Accurate Vote Tabulations}
    \label{sec:sim accurate tallies}

The upcoming plots present the number of ballots required to approve each assertion during the audit, both by the ALPHA-batch method and by our Batchcomp method. Each point in these plots represents a single assertion, where its value on the x axis is its margin in log-scale (minimal number of ballots that would need to be altered for the assertion to become false), and its value on the y axis is the number of ballots that the audit examined before approving the assertion. Each point in the plot is colored by the type of assertion it represents - either an above threshold assertion, a below threshold assertion, or a move-seat assertion.

For each election cycle, the top plot shows these results when using the ALPHA-batch RLA, the middle plot when using Batchcomp, and the bottom plot shows the difference in ballots required per assertion between ALPHA-batch and Batchcomp.

\paragraph*{Using the 22nd Knesset Election Results (2019)}

\begin{center}
    \hspace*{-0.2in}
    \includegraphics[scale=0.42]{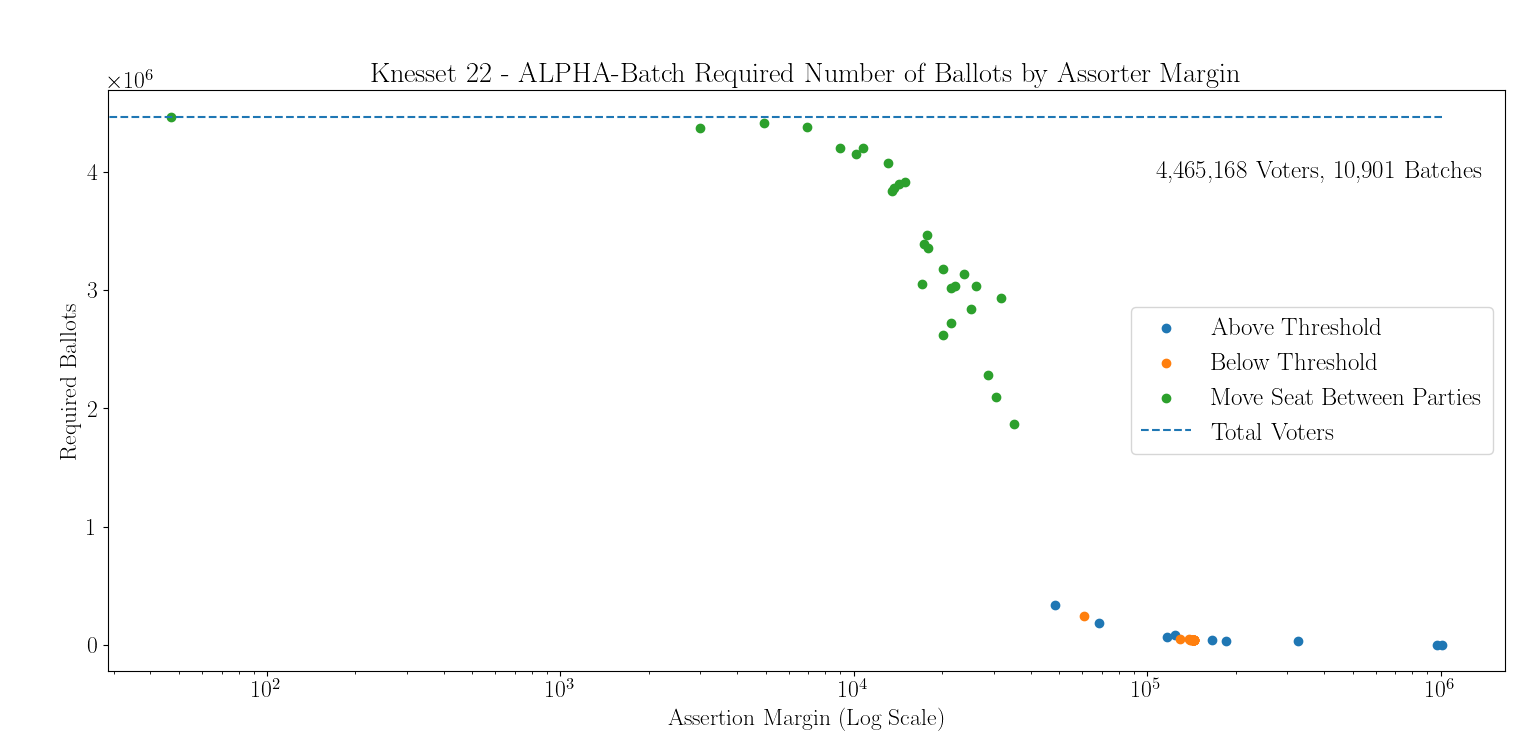} \\[2ex]
    \hspace*{-0.2in}
    \includegraphics[scale=0.42]{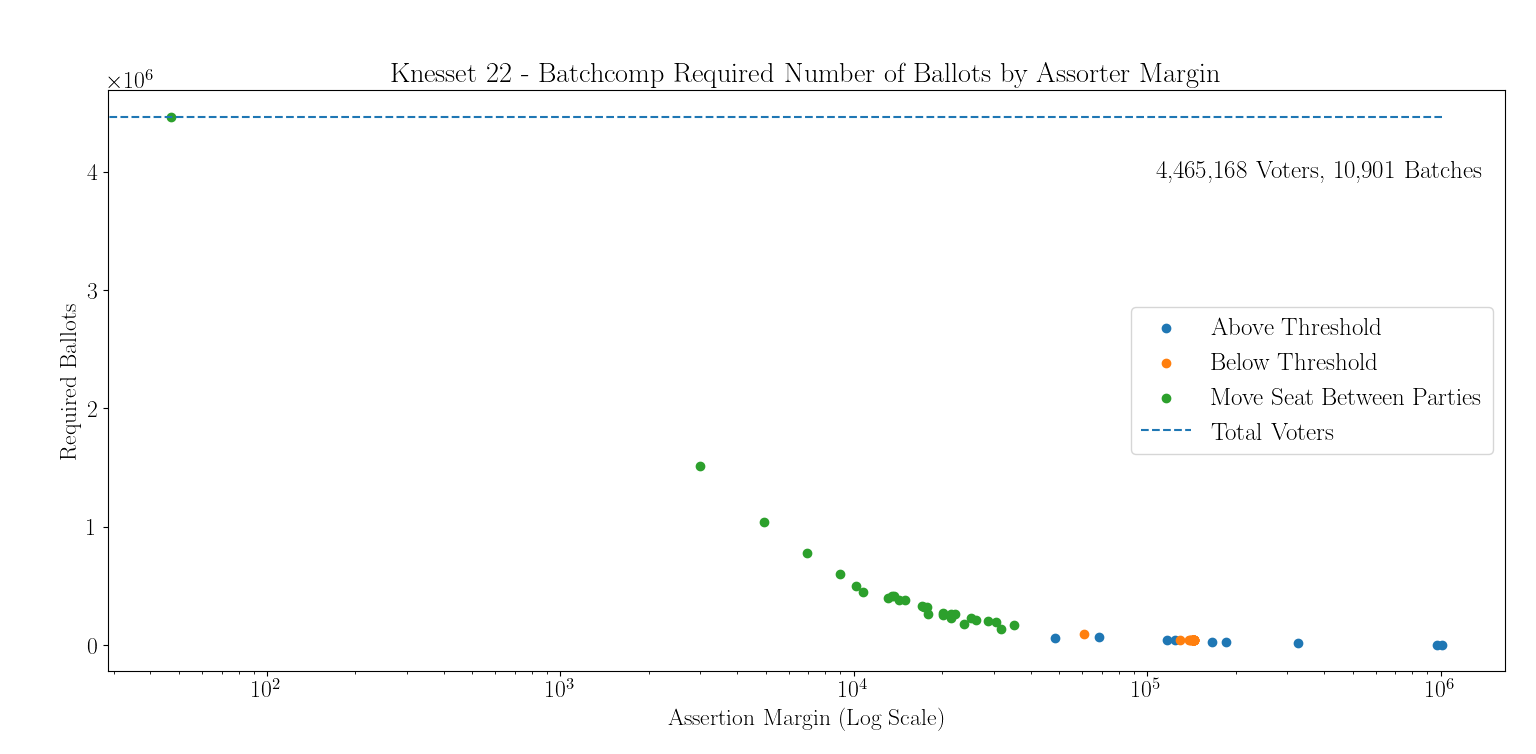} \\[2ex]
    \hspace*{-0.2in}
    \includegraphics[scale=0.42]{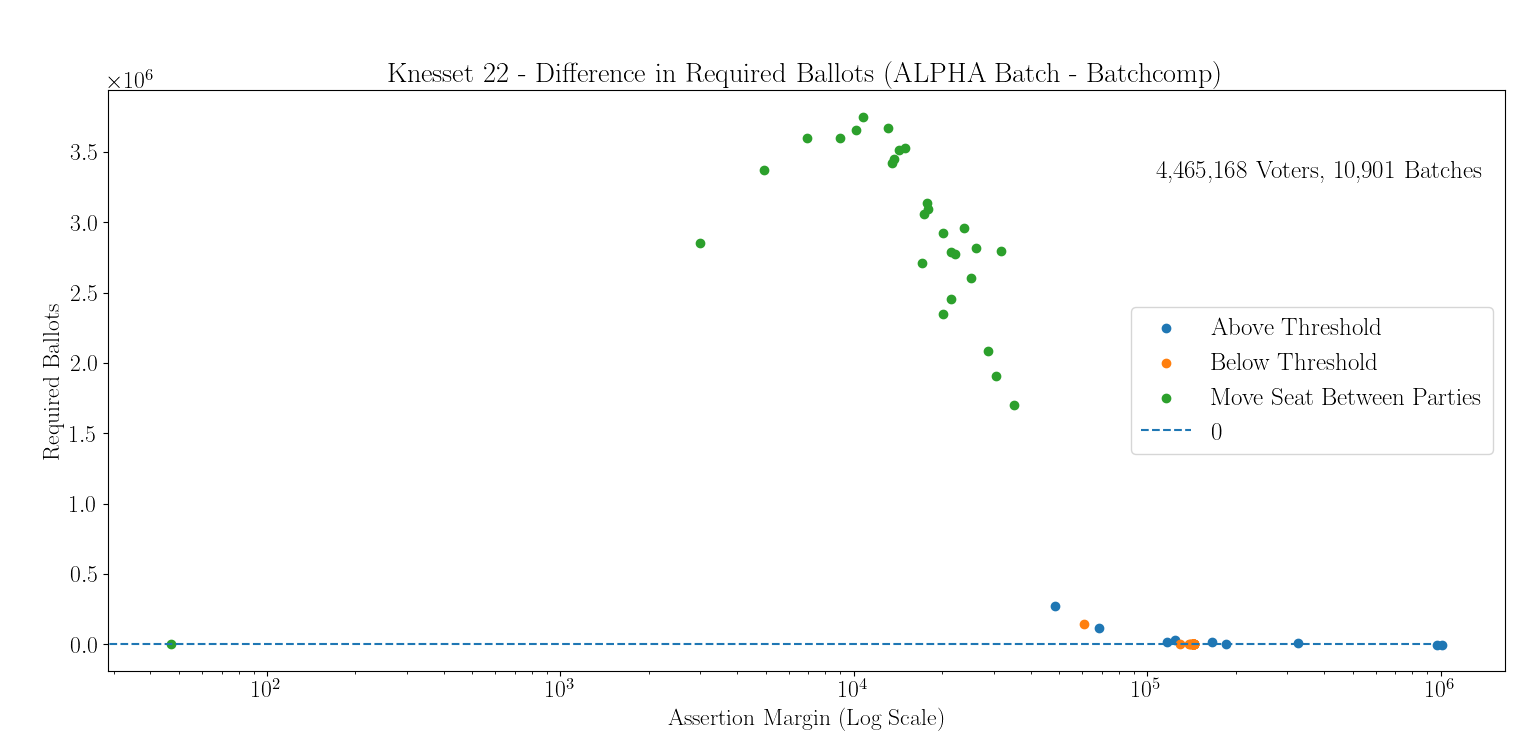}
\end{center}

Approving the reported winners for this election cycle required auditing  virtually all ballots by both methods, due to a single assertion which had a very small margin (47).   Without this assertion, the Batchcomp audit would be done after auditing 34\% of the ballots, while ALPHA-batch would still require 98\%.

The three assertions which required the most ballots to be approved by the Batchcomp algorithm are:
    \begin{center}
        \begin{tabular}{|c || c | c | c |} 
         \hline
          & Margin & Batchcomp & ALPHA \\
         Assertion & (\% of votes) & (\% of votes) & (\% of votes)
         \\[0.5ex] 
         \hline\hline
         Don't move a seat from & 47 & 4,465,090 & 4,465,139 \\
         Likud \& Yamina to UTJ \& Shas & $(0.001\%)$ & $(\approx 100\%)$ & $(\approx 100\%)$\\
         \hline
         Don't move a seat from & 2,996 & 1,513,454 & 4,367,793 \\
         Blue and White to Yisrael Beiteinu & $(0.07\%)$ & $(34\%)$ & $(98\%)$\\
         \hline
         Don't move a seat from & 4,919 & 1,036,336 & 4,409,273 \\
         Blue and White \& Yisrael Beiteinu & $(0.11\%)$ & $(23\%)$ & $(99\%)$ \\
         to UTJ \& Shas & & & \\
         \hline
    \end{tabular}
    \end{center}
    
\paragraph*{Using the 23rd Knesset Election  Results (2020)}
    \begin{center}
    \hspace*{-0.2in}
    \includegraphics[scale=0.42]{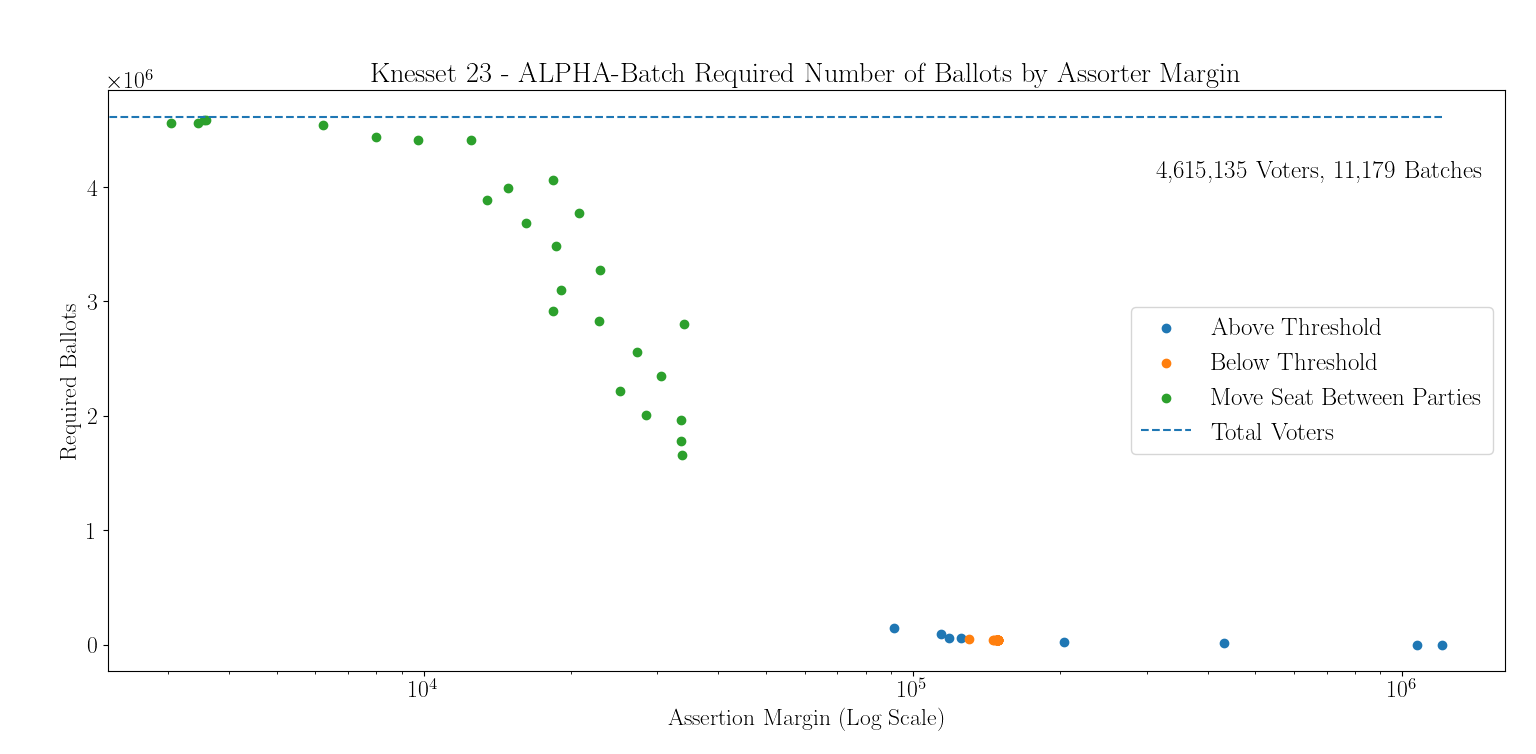} \\[2ex]
    \hspace*{-0.2in}
    \includegraphics[scale=0.42]{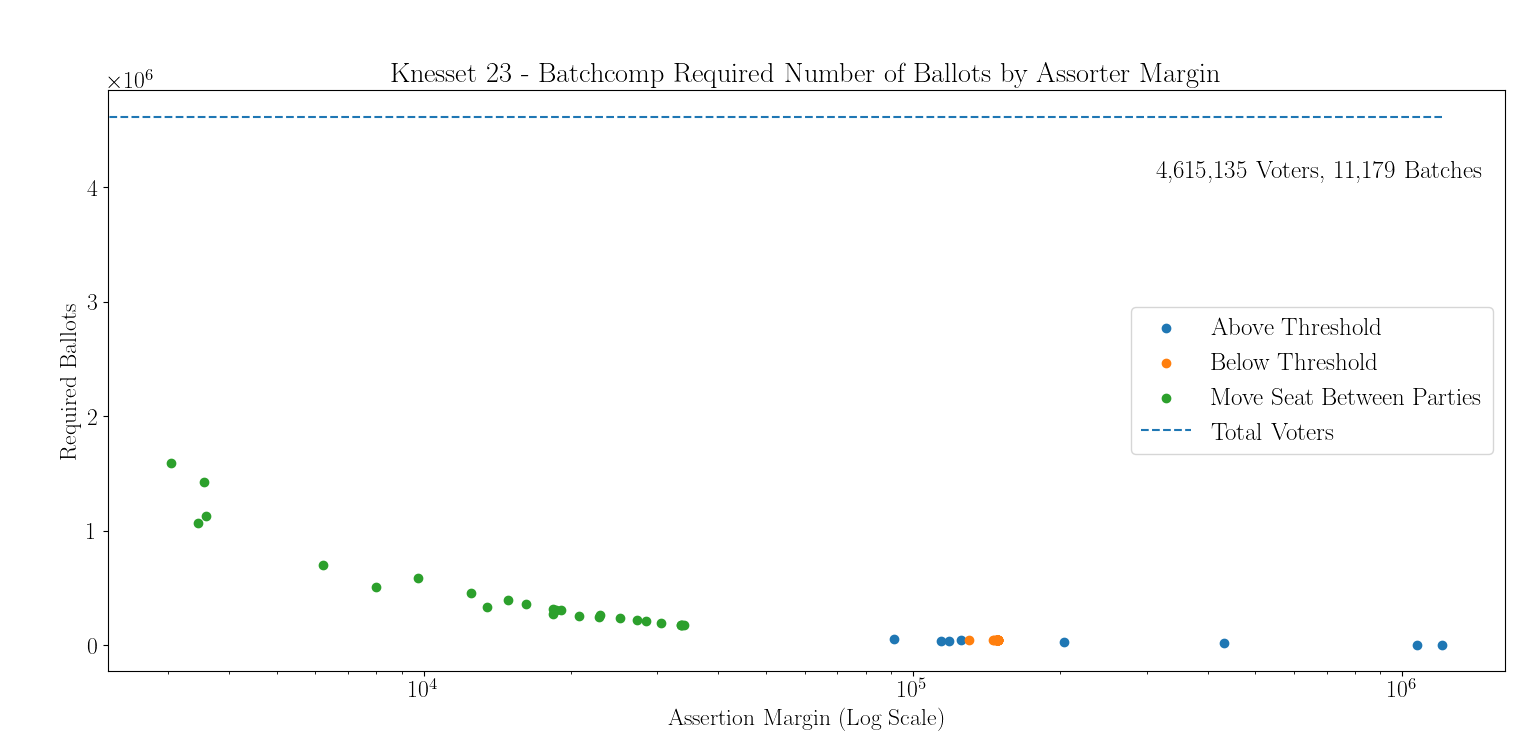} \\[2ex]
    \hspace*{-0.2in}
    \includegraphics[scale=0.42]{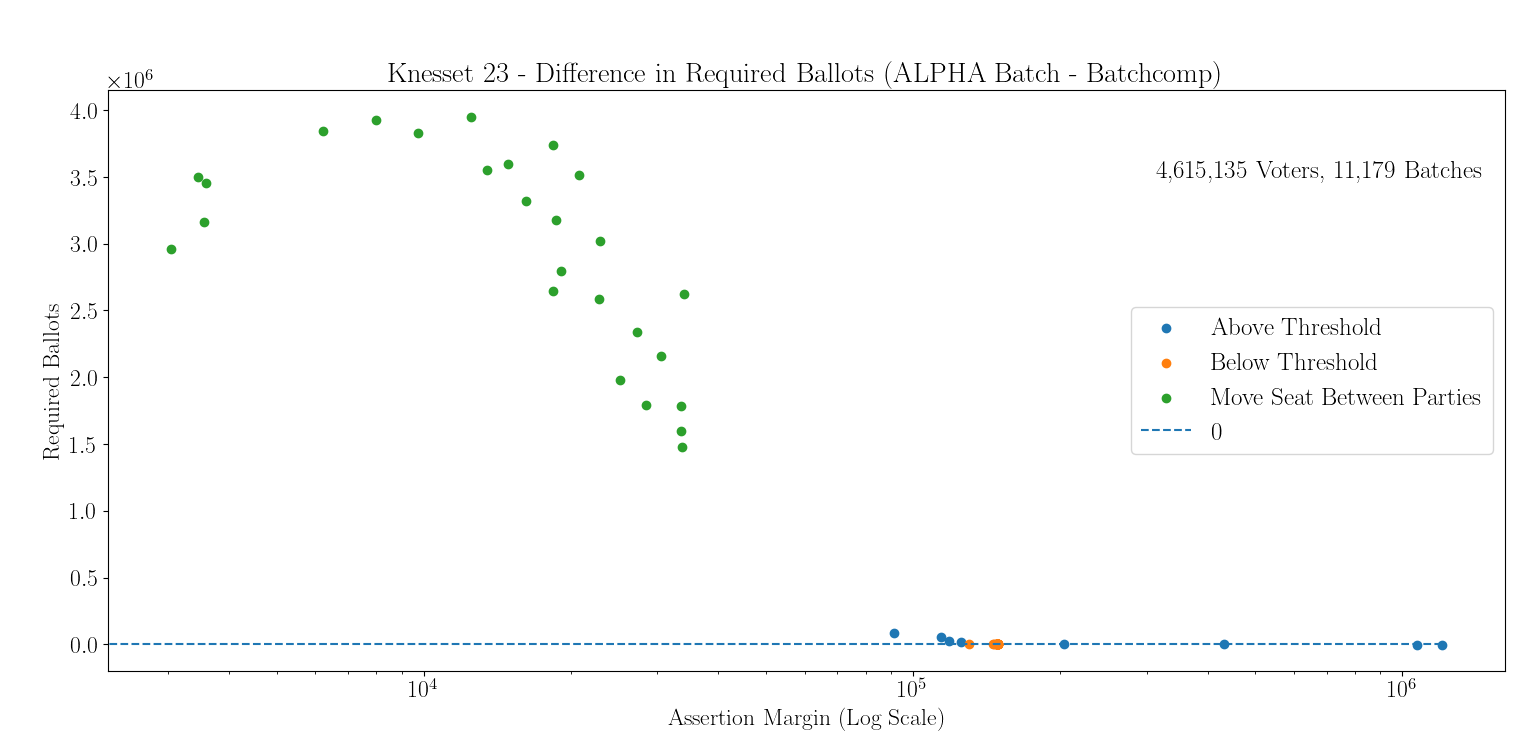}
\end{center}
Approving the reported winners for this election cycle required auditing 35\% of ballots by Batchcomp, while requiring 99\% by ALPHA-batch. 

The three assertions which required the most ballots to be approved by the Batchcomp algorithm are:

\begin{center}
    \begin{tabular}{|c || c | c | c |} 
     \hline
      & Margin & Batchcomp & ALPHA \\
     Assertion & (\% of votes) & (\% of votes) & (\% of votes)
     \\[0.5ex] 
     \hline\hline
     Don't move seat from & 3,042 & 1,593,006 & 4,556,963 \\
      Emet \& Blue and White to Likud \& Yamina  & $(0.07\%)$ & $(35\%)$ & $(99\%)$\\
     \hline
     Don't move a seat from & 3,545 & 1,421,340 & 4,583,377 \\
     Emet \& Blue and White to UTJ \& Shas & $(0.08\%)$ & $(31\%)$ & (99\%) \\
     \hline
     Don't move a seat from  & 3,591 & 1,126,277 & 4,580,758 \\
     Yisrael Beiteinu to UTJ \& Shas & $(0.08\%)$ & $(24\%)$ & (99\%) \\
     \hline
    \end{tabular}
\end{center}

\paragraph*{Using the 24th Knesset Election Results (2021)}
    \begin{center}
    \hspace*{-0.2in}
    \includegraphics[scale=0.42]{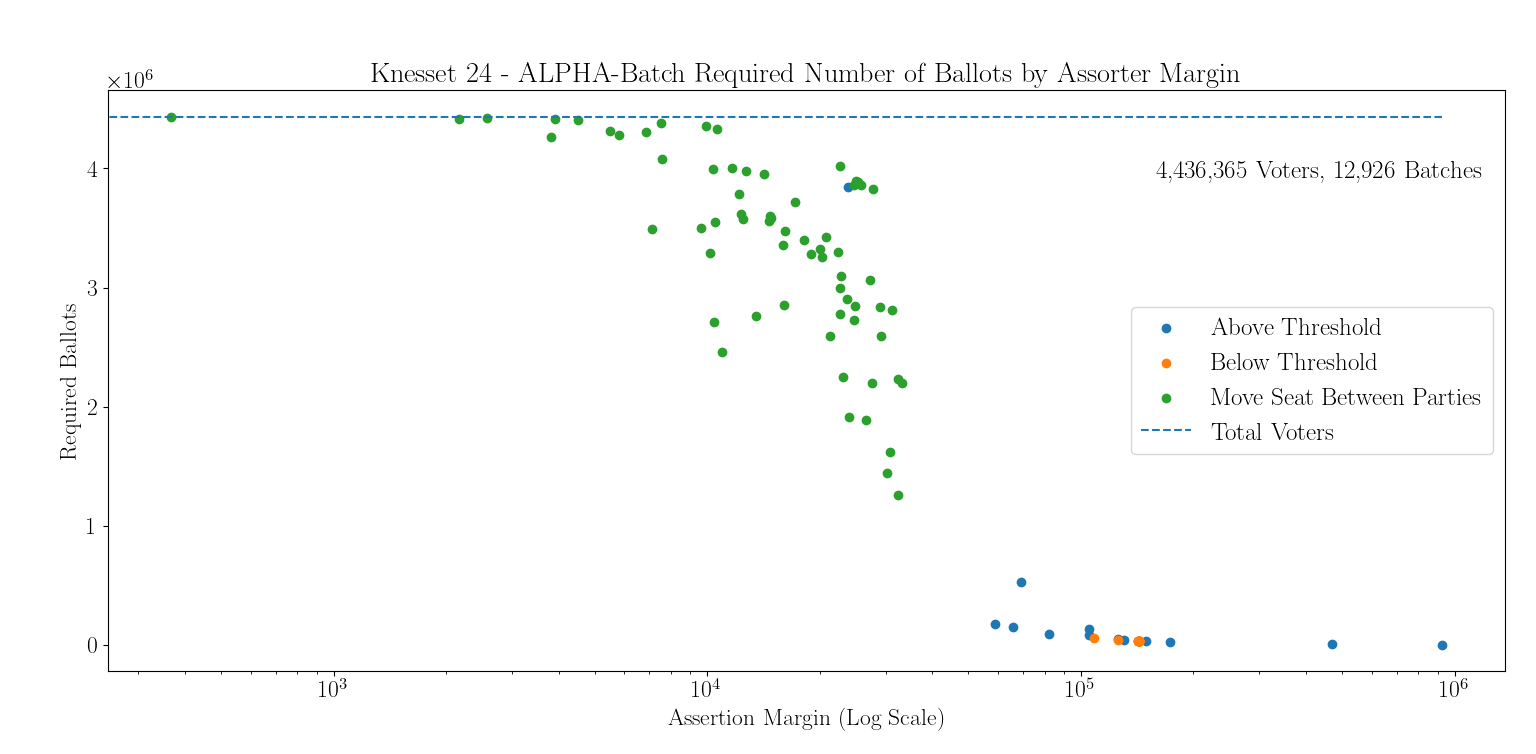} \\[2ex]
    \hspace*{-0.2in}
    \includegraphics[scale=0.42]{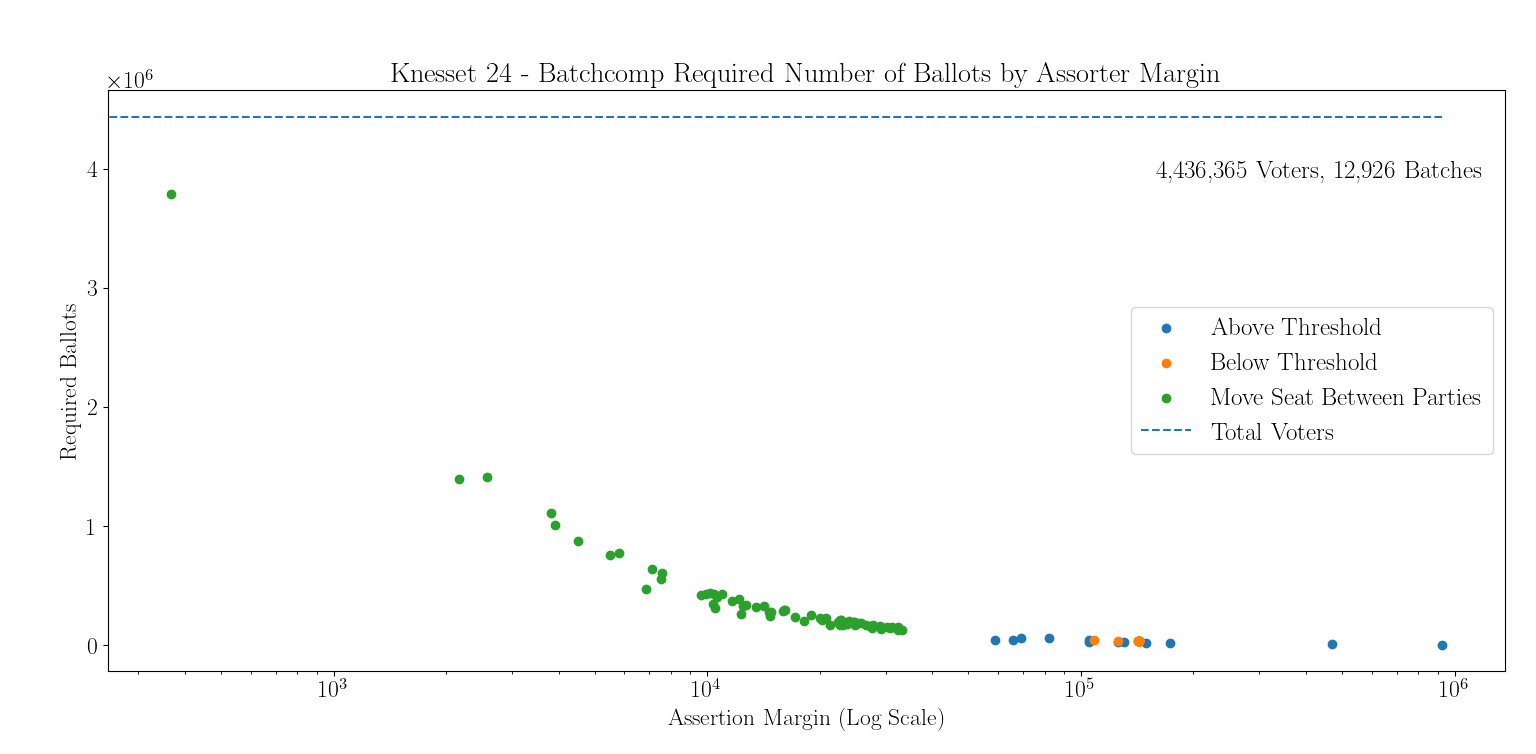} \\[2ex]
    \hspace*{-0.2in}
    \includegraphics[scale=0.42]{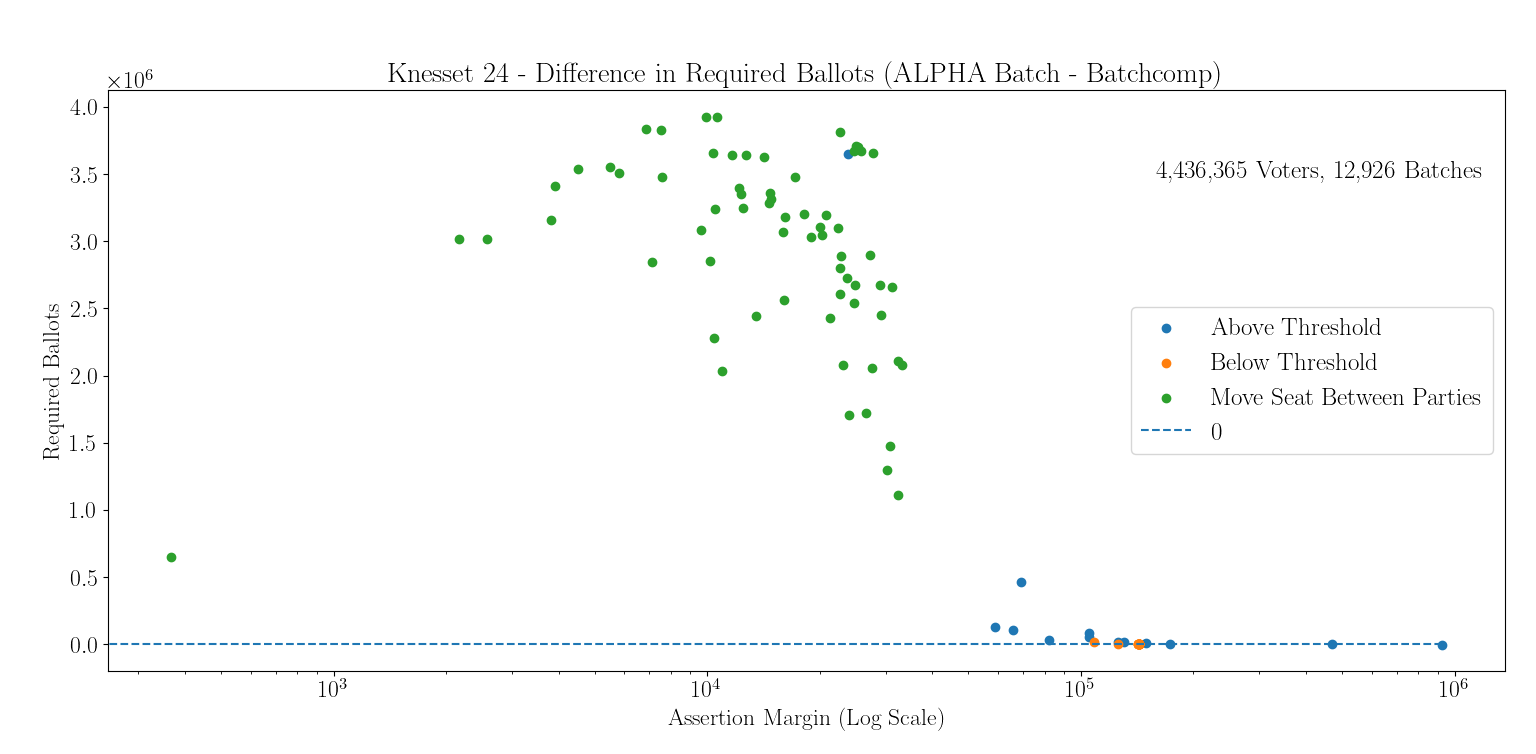}
\end{center}
Approving the reported winners for this election cycle required auditing 85\% of ballots by Batchcomp, while requiring virtually all ballots by ALPHA-batch. If it wasn't for a single assertion which had a very small margin (367 ballots), the Batchcomp audit would be done after auditing \~32\% of the ballots, while ALPHA-batch would still require reading nearly all  ballots.

The three assertions which required the most ballots to be approved by the Batchcomp algorithm are:
\begin{center}
    \begin{tabular}{|c || c | c | c |} 
     \hline
      & Margin & Batchcomp & ALPHA \\
     Assertion & (\% of total votes) & (\% of votes) & (\% of votes)
     \\[0.5ex] 
     \hline\hline
     Don't move a seat from & 367 & 3,782,124  & 4,435,111 \\ 
     Meretz to Labor & $(0.008\%)$ & $(85\%)$ & $(\approx 100\%)$\\
     \hline
     Don't move a seat from & 2,567 & 1,410,184 & 4,424,102 \\
     The Joint List to Likud \& Religious Zionist  & $(0.06\%)$ & $(32\%)$ & $(\approx 100\%)$\\
    \hline
     Don't move a seat from & 2,162 & 1,392,993 & 4,412,059\\
     New Hope to Yamina & $(0.05\%)$ & $(31\%)$ & $(99\%)$\\ \hline
    \end{tabular}
\end{center}

\subsubsection{Results with Small Tabulation Inaccuracies}
    \label{sec:sim inaccurate tallies}
    
In addition to checking the Batchcomp method's efficiency under ``perfect" conditions, we examine its tolerance to small counting errors. For this purpose, we compare the number of ballots it requires to approve each assertion under two conditions:
\begin{enumerate}
    \item When each ballot has a probability of $0.01$ to be misread in the reported tally. If a ballot is misread, it either becomes invalid (w.p.\ 0.1) or is counted towards a party drawn uniformly at random.
    \item When the reported vote tallies of all batches are completely accurate, as examined previously in~\Cref{sec:sim accurate tallies}.
\end{enumerate}
We present the plot described in~\Cref{sec:sim accurate tallies} for each condition, as well as one additional plot which shows the difference in ballots required between the two conditions. Note that the assertion margins presented in these plots (the x axis) are calculated according to the reported results and not the true ones, since the true margin changes in each repetition of the simulation.

The choice of $0.01$ probability for miscounting each ballot is inspired by historical data. Unless critical errors occur, both manual and electronic vote tabulations miscount less than 1\% of ballots~\cite{blom2020random, ansolabehere2018learning}. 

The 22nd Knesset elections require very small counting errors to change their results. For this reason, it's extremely unlikely for the seat-allocation to remain identical if approximately 1\% of the votes are miscounted. Since we are only interested in Batchcomp's performance when the reported and true winners match, this section only examines the 23rd and 24th Knesset elections.

\paragraph*{Using the 23rd Knesset Results (2020)}
\begin{center}
    \hspace*{-0.2in}
    \includegraphics[scale=0.42]{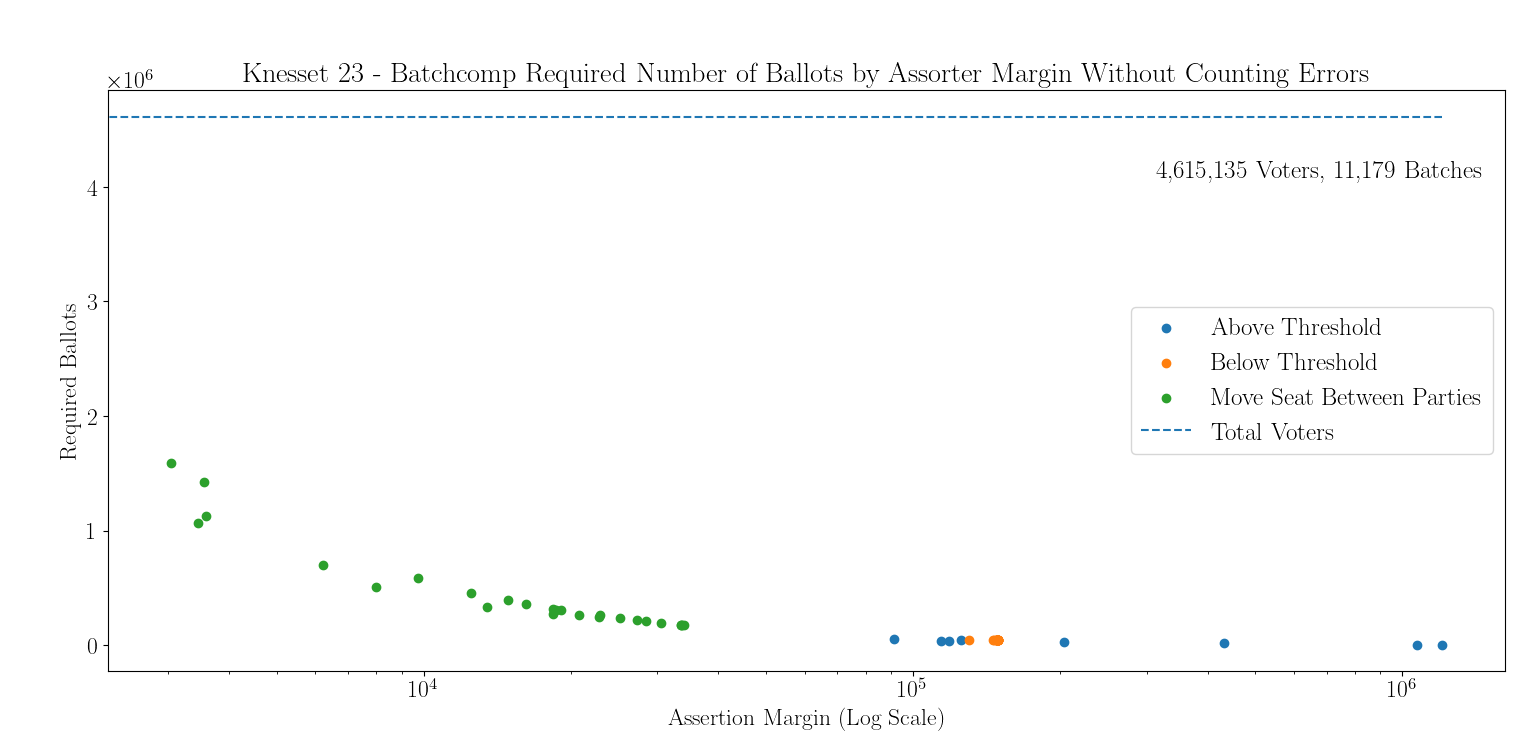} \\[0.2In]
    \hspace*{-0.2in}
    \includegraphics[scale=0.42]{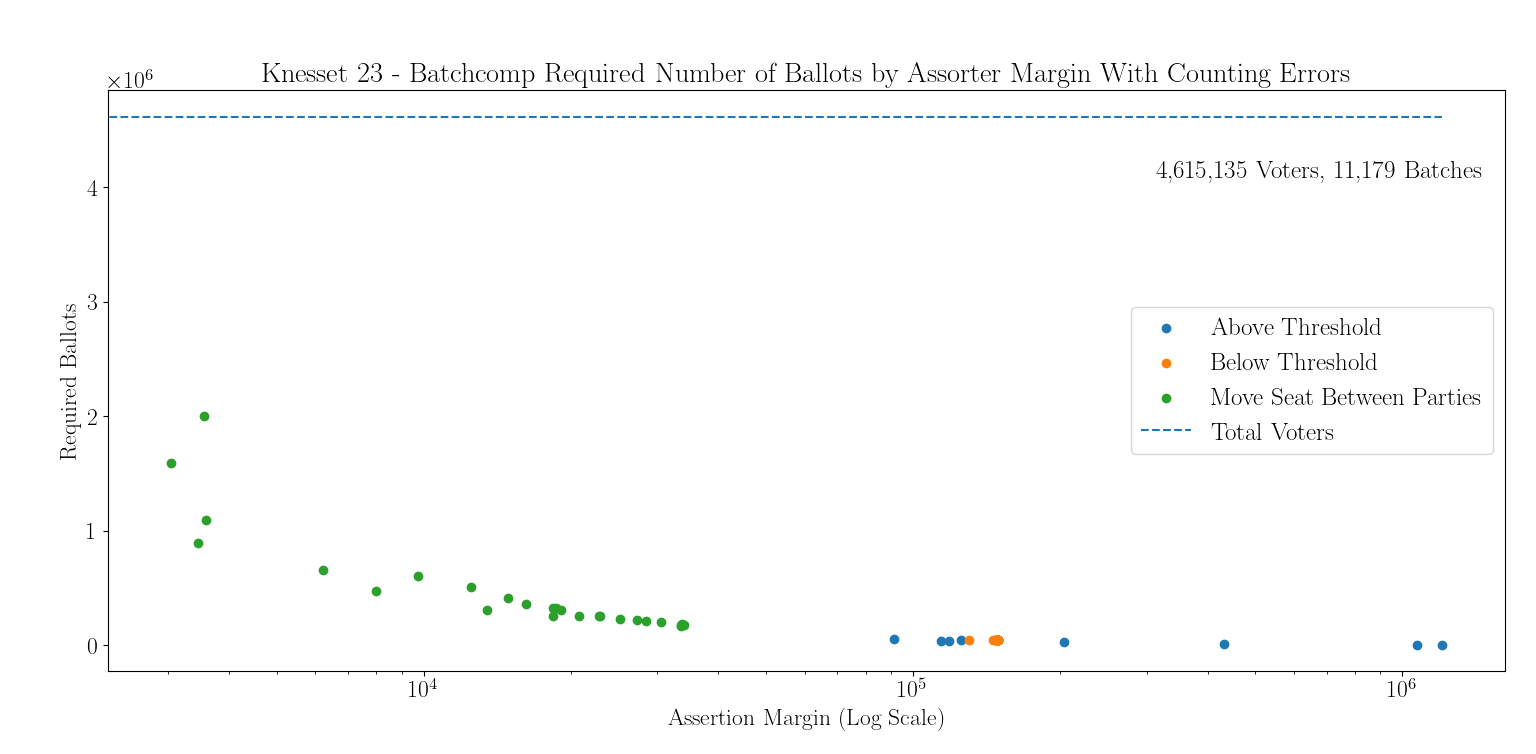} \\[0.2In]
    \hspace*{-0.2in}
    \includegraphics[scale=0.42]{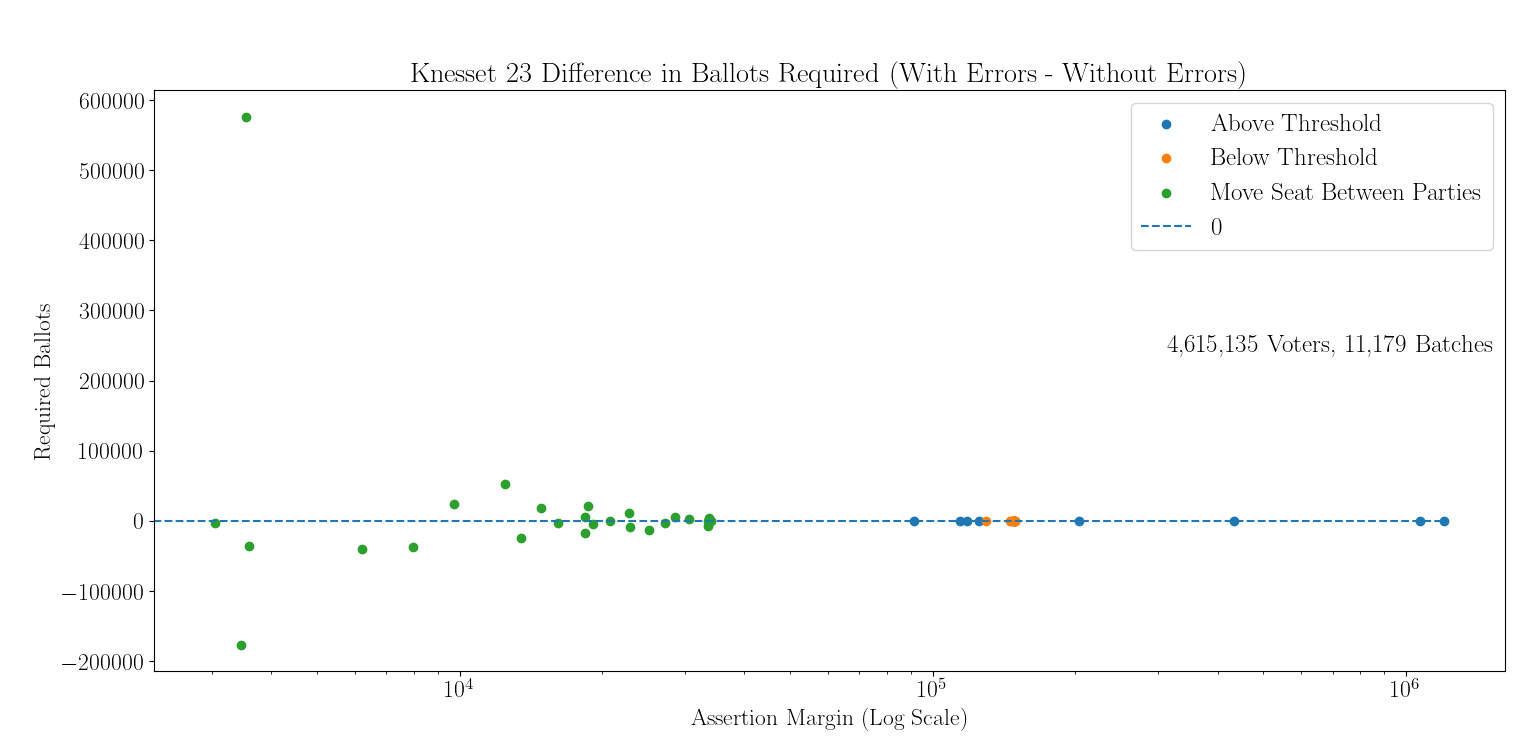}
\end{center}
\paragraph*{Using the 24th Knesset Results (2021)}
\begin{center}
    \includegraphics[scale=0.42]{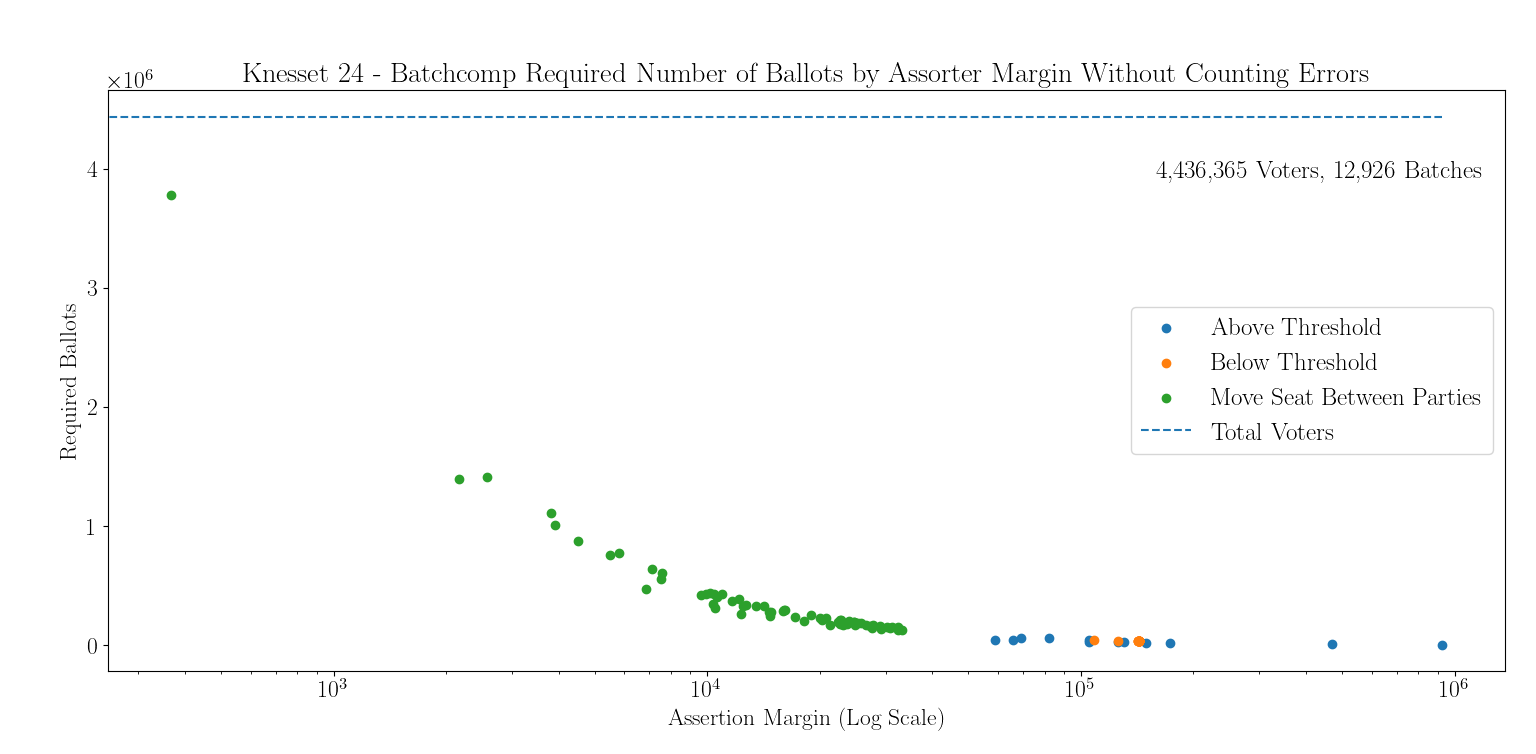} \\[0.2In]
    \hspace*{-0.2in}
    \includegraphics[scale=0.42]{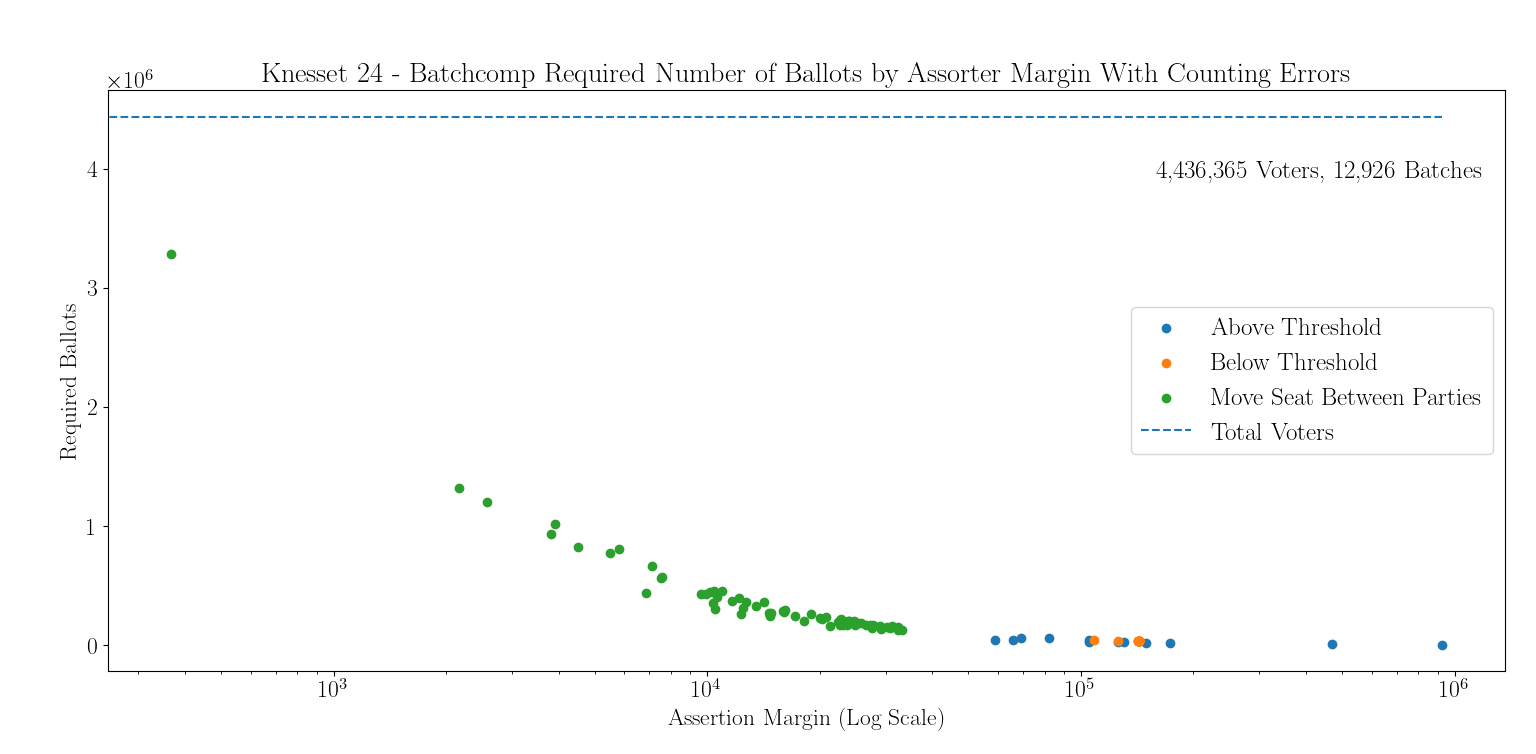} \\[0.2In]
    \hspace*{-0.2in}
    \includegraphics[scale=0.42]{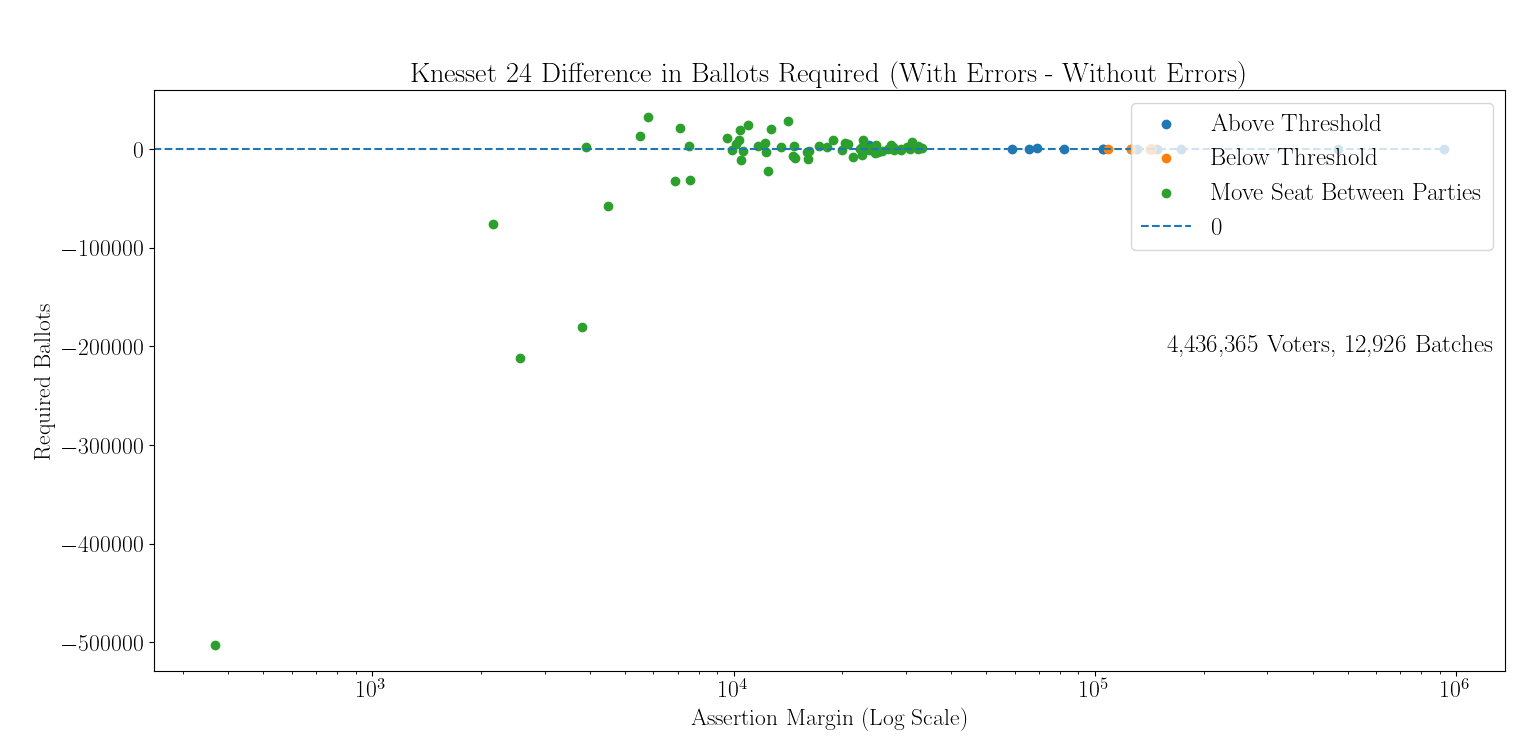} 
    \end{center}
Examining both election cycles shows that most assertions are not significantly effected by the existence of small counting errors. Meaning, the  number of ballots that are required to approve them remains similar. However, we can observe that random counting errors do effect some small margined assertions. Typically, assertions of the same type and of similar margins behave similarly. Here, somewhat surprisingly, some of these tight move-seat assertions become much easier to audit, while others become much more difficult. 

\paragraph*{Why Similar Assertions Exhibit Different Error Tolerance}

A full exploration of this phenomenon is beyond the scope of this work, but we attempt to provide a brief and mostly intuitive explanation for it. First, note that randomly miscounted votes are distributed evenly to all parties, while being disproportionately taken from parties who are above the electoral threshold, as they receive the vast majority of votes. Therefore, a party who is above the threshold will typically lose more votes than it gains due to random miscounts.    
    
Next, examine a move-seat assertion which confirms that compared to the reported seat-allocation, some party A doesn't deserve extra seats at the expense of party B. The counting errors which could cause this assertion to be false are either overcounting the votes for party A, or undercounting the votes for party B. Let $m_{red}$ be the minimal number of votes we would need to reduce from party A (compared to its reported tally) to make the assertion false. Similarly, let $m_{add}$ be the number of votes we would need to add to party B to make the assertion false.

We could now partition all move-seat assertions into two categories: (I) assertion for which $m_{red}<m_{add}$ and (II) assertions for which $m_{red}\geq m_{add}$. An examination of move-seat assertions from these simulations shows that assertions from category (II) are helped by random counting errors, meaning the errors reduce the number of ballots required to confirm these assertions. Meanwhile, assertions from category (I) are harder to confirm when random counting errors exist.

Since we've established that random errors typically reduce the tallies of parties who passed the threshold, assertions for which $m_{red}<m_{add}$ are harder to confirm when the parties of the assertion both lose votes, and vice-versa. This is because the margins of assertions from category (I), according to the true results, decreases by the existence of random errors, while the margin of assertions from category (II) increases.

\subsubsection{Simulation Conclusions} \label{sec: rla sim conclusions}

\paragraph*{Knesset Elections are Difficult to Audit}

Most Knesset elections have very tight margins, which make them difficult to audit in a risk-limiting manner. If the Election winners win with a margin of below 0.001\% of the total ballots, it's unlikely that any RLA method could approve them without close to a full manual recount.

To remedy this, if RLAs are implemented for such elections, the auditing body could decide in advance that some tight-margined assertions are not worthy of auditing. In the 24th Knesset elections, for example, the tightest assertion (don't move a seat from Meretz to Avoda) relates to two parties who are in an electoral alliance, indicating that they are ideologically aligned. Auditing this assertion nearly triples the length of the audit, despite it being one of the least critical assertions for this election.

In addition to this, it is possible to adjust any move-seat assertion such that it verifies that no more than a single seat should be moved between two parties, compared to the reported results. For any two parties, $p_1, p_2$ and their move-seat assertion $a_{p_1,p_2}^{move}$, this is achieved by defining $a_{p_1,p_2}^{move}$ (see~\Cref{def: move seat assorter}) as if $p_1$ won one extra seat at the expense of $p_2$. Meaning, we define:
\begin{equation*}
    a^{move}_{p_1,p_2}(b) := 
    \begin{cases}
    \half + \frac{s^{rep}(p_1)+2}{2(s^{rep}(p_2)-1)} & \text{if $b$ is for $p_2$} \\
    0 & \text{if $b$ is for $p_1$} \\
    \half & \text{otherwise}
    \end{cases}
\end{equation*}
If a move-seat assertion has a very small margin, we can switch its assorter to this one, thereby shortening the audit at the expense of a weaker guarantee.

\paragraph*{Batchcomp consistently beats ALPHA-batch in These Settings}

While auditing the entire Knesset elections proves to be rather difficult, examining the number of ballots required to approve the various assertions shows that Batchcomp significantly outperforms ALPHA-batch. Generally, assertions that had very small or fairly large margins required a similar number of ballots by both algorithms, while assertions with margins of between 0.01\% and 2\% were significantly easier to audit using Batchcomp. Some assertions which ALPHA-batch could not approve without a nearly full manual recount were approved by Batchcomp while examining less than 20\% of the paper-backup ballots. 

\paragraph*{Resilience to Random Errors}

The Batchcomp RLA appears to be resilient to small random errors in the reported tallies of the various batches. While not shown within this work, this observation is also true regarding ALPHA-batch, as it makes no use of the reported tallies of specific batches.

\paragraph*{What Actually Effects The Audit's Efficiency}

The efficiency of the audit (number of ballots that are required to approve correct winners) seems to be effected only by the assertion which has the tightest margin. The last assertion to be approved remains consistent when repeating the audit's simulation multiple times. The number of assertions and the margin of any other assertion, other than the one with the minimal margin, appear to have no effect on the efficiency of the RLA. 

This is due to the fact that when the reported tallies of the batches are accurate, every Batchcomp assorter has the same value across all batches. Therefore, the progression of the p-value of each assertion ($\frac{1}{T_k}$, as defined in~\Cref{sec:Batchcomp algorithm description}) during the audit remains similar regardless of the order in which we audit the batches. 

This observation persists when there are small random counting errors, but may change the specific assertion which is the most difficult to audit. This is since, as explained in~\Cref{sec:sim inaccurate tallies}, two assertions with similar margins may be effected differently by random counting errors. 

In some cases, the problem of approving a reported election result can be reduced to SHANGRLA assertions in multiple ways, and the auditing body has to choose the specific set of assertions to use during the audit. If the audit is conducted using the Batchcomp method, this observation teaches that we should choose the set of assertions where the tightest margined assertion has the maximal margin. Ideally, if our reduction yields assertions that are both sufficient and necessary for the reported winners of the elections to be correct, then the margin of the tightest assertion is exactly the margin of the entire elections. If this is the case, then any sufficient and necessary set of assertions would perform similarly. Trying to minimize the number of assertions does not, by itself, effect the efficiency of the audit, though it might reduce the computation time required per ballot or batch.

\begin{conclusion}
    Any reduction into SHANGRLA assertions which are both sufficient and necessary yields a similarly efficient audit, regardless of the number of assertions.
\end{conclusion}

\subsection{Existing Recounting Methods}

To the best of my knowledge, recounts in the Israeli Knesset elections are currently performed without an evidence based approach or a clear statistical guarantee. Currently, votes are tabulated manually in each polling place, and the vote tally of each location is reported to the Central Election Committee. This committee reviews the tallies of each polling place, and may re-count the paper-backup ballots if necessary. There appears to be no systematic method for recounting ballots. From published cases where recounts were conducted, it appears that recounts are performed either after complaints of election fraud at a specific polling place, if inconsistent election records are discovered, or sometimes at randomly selected locations~\cite{decisionsGuidelines}.

One particularly interesting partial-recount of ballots happened in 2019, following the elections for the 21st Knesset. In these elections, the New Right party received 3.22\% of the valid votes, falling 1,454 ballots (0.03\% of the valid votes) short of the electoral threshold and thus receiving no seats in parliament. Following these results, the New Right party asked for a vote recount, which eventually resulted in recounting the ballots at 66 polling places where the party claimed for irregularities, amounting to approximately 26,000 votes. The recount resulted in the New Right party losing 3 votes, after which their leaders accepted the original election results~\cite{yeminHadash}.

RLAs offer an evidence based solution for this problem- running the audit with a single assertion, which attempts to confirm that the New Right party is truly below the electoral threshold (see~\Cref{sec: below assorter}). It comes with the statistical guarantee that if the examined party did pass the threshold, the probability of the audit approving the election results is bounded by a pre-set parameter. 

However, note that in this instance, the number of ballots that would need to be examined in such an audit is expected to be rather large. Using Batchcomp with a risk limit of 0.05 would require reading approximately 2,500,000 ballots to approve the results according to simulations, given that the initial vote tabulation was accurate.
\newpage
\section{Census Risk Limiting Audits}
\label{sec:census RLA}

This section presents a risk-limiting audit method for a population census. It applies to nations which allocate political power to their constituencies or federal-states in proportion to their population according to a certain class of methods (highest averages), and who conduct a post-enumeration survey as recommended by UN guidelines~\cite{un2010pes}. According to these guidelines, a PES is performed by randomly sampling a small number of households, re-running the census over this chosen sample, and then comparing the results to the original census. For consistency, throughout this section, we assume that this allocated political power is manifested as the number of representatives a region receives in parliament, and refer to these regions as the nation's federal-states. The goal of our audit is to provide a clear statistical guarantee regarding the correctness of this census' resulting allocation of representatives.

\begin{comment}
If one wishes to apply RLAs to verify a population census results, then  we need to define what true results are considered accurate enough for the audit to approve them. For elections, an RLA approves the results if it is highly likely that the reported winners of the elections are correct. Auditing census counts requires defining an equivalent condition. For this purpose, we examine the previously mentioned use of a census: allocating political power to different regions of the country in proportion to their population.
\end{comment}

To achieve such a guarantee, we first need to define what allocation is considered correct. When auditing parliamentary elections, we define the true number of seats a party should win as the number it deserves according to the tally of the paper-backup ballots. When auditing a census, one might wish to similarly define the results of the PES as the true results. Such a definition could be problematic, however, as the PES only runs over a small sample of households.

For this reason, we view the true results of the census as the results the PES would have if it was to run over all households. This means that technically, a census RLA assumes that the PES surveyed all households. During the actual census audit, however, it only asks for the information the PES collected on a small, randomly chosen sample of households, which is exactly the data that the PES actually has. 

The census RLA is performed by sequentially sampling households and processing the census and PES information regarding them.
Since the PES only runs over a small sample of households, the audit is limited in its length. For this reason, it could be unreasonable to set a risk-limit (probability of approving wrong results) before the audit begins, as we do in election RLAs. If we were to do so, then
the audit may fail to approve a correct representative allocation even when using the entire PES sample, resulting in an inconclusive outcome. This issue does not exist when auditing elections, as the audit can keep sampling and reading ballots until it either approves the reported winners, or until reading all ballots and learning the true results. 

The observation above leads us to slightly change the statistical guarantee that a census RLA provides: instead of setting the risk-limit and then running the audit, the census RLA runs over the entire PES and then returns the risk-limit with which it can approve the census representative allocation. This results in the following statistical guarantee:

\begin{center}
\fbox{\begin{minipage}{40em}
  \begin{center}
  \textbf{\hypertarget{census rla guarantee}{The Census RLA guarantee}:}
  
  For any $0 < \alpha \leq 1$, if running the PES over all households would lead to a different allocation of representatives than the census, then the probability that a census RLA returns a value $\alpha'$ such that $\alpha' \leq \alpha$ is at most $\alpha$.
    \end{center}
\end{minipage}}
\end{center}

Using this guarantee, a governing body could examine the risk-limit returned by a census RLA in order to decide whether the allocation of representatives to states is reliable enough. If it is not, they may decide to conduct a second round of re-surveying, and to continue the audit on these newly re-surveyed households. Alternatively, if the audit's outputted risk-limit is too high, it may also decide to re-run the census altogether. Our suggested census RLA method can also examine which state specifically is more or less likely to have a correct number of representatives.  This will be discussed further at the end of~\Cref{sec: crla conclusions}.

As a summary of this section, a census-RLA is an altered version of an election-RLA, where we wish to verify that the census results in a correct allocation of representatives to the federal states. The critical differences between an election-RLA and a census-RLA are summarized in the following table:
\begin{center}
        {\renewcommand{\arraystretch}{1.1}
        \begin{tabular}{| @{\hskip 0.1in}c@{\hskip 0.1in} || @{\hskip 0.1in}c@{\hskip 0.1in} | @{\hskip 0.1in}c@{\hskip 0.1in} |} 
         \hline
          Category & Election RLA & Census RLA \\ 
         \hline \hline 
         Goal & \begin{tabular}{c}
            Approve \\ election winner's
         \end{tabular} & \begin{tabular}{c}
            Approve census' allocation \\ of representatives
         \end{tabular}
         \\ \hline 
         Applicable & Ballots have & Allocation is proportional  \\
         When &  physical backups & to population and a PES exists
         \\ \hline
         Audited Unit &  Ballot &  Household 
         \\ \hline
         Reported Result &  Electronic vote count &  Census resident count 
         \\ \hline
         True Result &  Paper-backup ballots manual count &  PES resident count
         \\ \hline
         Risk-Limit & Pre-set & Outputted by audit 
         \\ \hline
        Audit's Length & 
        Factor of risk-limit &
         Factor of PES size
         \\ \hline
        \end{tabular}}
    \end{center}

\subsection{Preliminaries and Notation}

\subsubsection{Post Enumeration Survey}
\label{sec: pes}
A post enumeration survey is a process which measures the accuracy of a population census by conducting an  independent population survey over a small portion of randomly chosen households. Our census-RLA method assumes that the PES is done according to the guidelines published by the Department of Economic and Social Affairs of the United Nations~\cite{un2010pes}. According to these guidelines, the PES begins by choosing a partial sample of the households in a nation, such that each household has an equal probability of being included in this sample. In some instances, the pool of households from which this sample is taken includes all households that were surveyed in the original census. In other cases, such as in the US census, this pool of potential households is constructed independently of the original census. Our suggested census RLA method applies for both options.

After sampling the households to be included in the PES, a new and independent survey contacts each household and asks them the exact same questions as the original census. Since the PES is conducted a few months or years after the census, the household members report their answers as they were on the date of the original census. For our purposes, the only information of interest is the number of residents living at each  household\footnote{Some countries may allocate representatives to federal-states according to the number of a specific sector of the population that they hold (e.g.\ eligible voters or citizens). In this work, we assume it's simply the number of residents, but our methods apply in the same manner otherwise.}. If a sampled household did not respond during the PES, it reports that it holds no residents.

We denote the information given by the census as follows:
\begin{itemize}
    \item $H^{cen}$: A list of households that were surveyed.
    \item $g^{cen}(h)$: The number of residents a household $h\in H^{cen}$ has according to the census.
\end{itemize}
And denote the information given by the PES as:
\begin{itemize}
    \item $H^{PES}$: A list of households from which the sample used by the PES was chosen. This may or may not be identical to $H^{cen}$. %We assume to know which households from $H^{PES}$ are shared with $H^{cen}$ and which are unique to only one of these sets. 
    \item $\tilde{H}^{PES}$: The households which were surveyed by the PES. Must be a subset of $H^{PES}$.
    \item $g^{PES}(h)$: The number of residents a household $h\in {\tilde{H}^{PES}}$ holds according to the PES.
\end{itemize}

\subsubsection{The Census RLA Model}
\label{sec: census RLA model}
In our model a nation allocates $R$ representatives to its federal-states, whose set we denote as $\mathcal{S}$, in proportion to their population as measured by a country-wide census. We assume that after the census is finished, a PES is conducted as described in~\Cref{sec: pes}. Following the census and the PES, we learn $H^{cen}$ and $g^{cen}$ from the census, and $H^{PES}, \tilde{H}^{PES}$ and $g^{PES}$ from the PES.

We assume that the nation allocates its $R$ representatives to its federal-states using the census results, according to a highest averages method. Highest averages methods, such as the D'Hondt method described in~\Cref{sec: knesset seat allocation}, are a class of methods which allocate representatives to federal-states using an imaginary table. Each row in this table represents a state, and each column represents a potential number of representatives it may win. Each cell of this table holds a value which depends on its federal-state (its row), the number of representatives it represents (its column) and the number of residents in its federal-state according to the census. The $R$ cells with the highest values in the table are colored, and each federal-state receives a number of representatives equal to the number of colored cells it has in its row. The value at cell $[s, r]$, where $s\in \mathcal{S}$ and $r\in [R]$, is:

\begin{align} \label{eq: census repr allocation}
    \frac{g^{cen}(s) + c_s}{d(r)},
\end{align}
    
where $c_s$ is a constant which depends on the state $s$, and $d:\bbN \rightarrow \bbN$ is a monotonically increasing function. Since $d$ is monotonically increasing, the values in each row of the table are monotonically decreasing.

The choice of $c_s$ and $d$ determines the exact allocation method within the class of highest averages methods. For example, setting $d(j)=j$ and $c_s=0$ for any $s\in \mathcal{S}$ results in the D'Hondt method described in~\Cref{sec: knesset seat allocation}.

The additive factor in~\eqref{eq: census repr allocation}, $c_s$, allows our model some added flexibility, meaning it can include more political systems. If, for example, representatives are allocated to federal-states according to a weighted sum of their population and their land-area, as done in Denmark~\cite{denmarkElectorealSystem}, this model supports this type of seat allocation by defining $c_s$ to be the land-area of $s$ in appropriate units. This additive factor can also be used in cases where part of a state's population is not included in the PES. In the US, for example, we would want to exclude people living in group residence (e.g.\ homeless people, nursing home residents, people living in remote Alaska, etc') from the audit, since they are not covered by the PES. If we wish to exclude them from the RLA, we can assume their number according to the census is accurate and run the audit over the rest of the population. To do so within this model, we can define  $c_s$ to be the number of persons without permanent residence in state $s$ according to the census. 

More formally, a census-RLA is defined as follows:
\begin{definition}
    Let $f$ be a social choice function which allocates representatives to federal-states based on census data. Let $\mathcal{C}$ be a randomized algorithm which outputs a non-negative value, and takes the following inputs:
    \begin{itemize}
        \item A list of households according to the census $H^{cen}$ and according to the PES $H^{PES}$, where the state of each household is known.
        \item A subset of households which were surveyed during the PES: $\tilde{H}^{PES}\subseteq H^{PES}$.
        \item Census results $g^{cen}:H^{cen}\rightarrow \bbN_0$.
        \item PES results $g^{PES}:\tilde{H}^{PES}\rightarrow \bbN_0$.
    \end{itemize}
     
    $\mathcal{C}$ is a census risk-limiting audit for the social choice function $f$ if for any $H^{cen},H^{PES}$, $g^{cen}$ and PES results over all households $g^{PES}:H^{PES}\rightarrow \bbN_0$ , given a random subset of PES households $\tilde{H}^{PES} \subseteq H^{PES}$ of a certain pre-determined size, we have:
    \begin{gather*}
        f(g^{cen}(\cdot)) \neq f\left(g^{PES}(\cdot)\right) \\ \Downarrow \\
        \forall \alpha \in [0,1], \Pr \left[\mathcal{C}\left(H^{cen}, H^{PES}, \tilde{H}^{PES}, g^{cen}(\cdot), g^{PES}\left(\tilde{H}^{PES}\right)\right) \leq \alpha \right] < \alpha.
    \end{gather*}
    Where $g^{PES}\left(\tilde{H}^{PES}\right)$ denotes that $\mathcal{C}$ only receives access to the PES results over the surveyed households $\tilde{H}^{PES}$.
\end{definition}
Note that in this definition, $g^{PES}$ encodes the results of the PES if it was to run over all households. However, $\mathcal{C}$ only receives access to $g^{PES}$ over households that were actually surveyed during the PES. Thus, if $f(g^{cen}(\cdot)) \neq f\left(g^{PES}(\cdot)\right)$ is true, it means if the PES had surveyed all of $H^{PES}$, it would result in a different allocation of representatives than that of the census. If this is true, the~\hyperlink{census rla guarantee}{census RLA guarantee} demands that for any $\alpha\in[0,1]$ a census RLA will output a value smaller than $\alpha$ w.p.\ of at most $\alpha$, which is exactly the demand stated in the definition above.

\subsection{The Census RLA Algorithm}
This section suggests a new method to performs census RLAs, which relies on the SHANGRLA framework. In~\Cref{sec: census household-level assorter}, we design SHANGRLA assertions for auditing the census' resulting allocation of representatives to the federal-states. While these assertions can be used as-is to perform a census RLA, they are only an intermediate step in the development of more efficient assertions that are eventually presented in~\Cref{sec: census comparison assorter}. These assertions are used by an adapted version of the ALPHA martingale test to perform a census RLA, as described in~\Cref{sec: census RLA description}.

\subsubsection{Assumption and Notation}
\label{sec: census assumption}
\begin{comment}
Fix a certain country with federal-states $\mathcal{S}$, who allocates $R$ representatives to these states using a highest averages method. Denote the formula which is used to calculate the value at the row of state $s$ and at column $r$ as:
$$
    \frac{g^{cen}(s) + c_s}{d(r)},
$$
where $g^{cen}(s)$ is the number of residents at state $s$ according to the census, $c_s\in \bbR$ is a constant, and $d:\bbN \rightarrow \bbN$ is a monotonically increasing function. 

As mentioned in~\Cref{sec: pes}, we denote the list of households according to the census as $H^{cen}$, and the list of households from which the PES was sampled as $H^{PES}$. Denote the union of all households as $H:=H^{cen}\cup H^{PES}$. Denote the households that were actually surveyed during the PES as $\tilde{H}^{PES}$.
\end{comment}
Our census RLA  method relies on one simplifying assumption:
\begin{assumption}
    In both the census and in the PES, the number of residents in a single household is upper-bounded by a known value, denoted as $g^{max}$.
\end{assumption}

The value $g^{max}$ must be set before the PES is conducted. Both the census and the PES must report that all households have $g^{max}$ residents at most. 

This assumption is necessary due to a critical difference between elections and censuses; In elections, a single ballot has very limited power. In a census, if it was not for this assumption, a single household could hold an arbitrarily large number of residents and completely swing the allocation of representatives to the states.

Under this assumption, the number of residents at a household according to the census is given by the function $g^{cen}: H^{cen} \rightarrow [g^{max}]\cup\{0\}$, and the number of residents according to the PES is given by $g^{PES}:\tilde{H}^{PES} \rightarrow [g^{max}] \cup\{0\}$. Throughout the next sections, we sometimes abuse notation by applying $g^{cen}$ on households that are not from $H^{cen}$, or applying $g^{PES}$ on households that were not surveyed during the PES. In any such case, we assume that these functions return 0. Finally, for a state $s\in \mathcal{S}$ and a household $h\in H$, we denote the number of residents from state $s$ at household $h$ by $g^{cen}_s(h)$ (according to the census) and $g^{PES}_s(h)$ (according to the PES). If $h$ is not in state $s$, we simply have $g^{PES}_s(h) = g^{cen}_s(h) = 0$.

\subsubsection{Census Assorters}
\label{sec: census assorters}
We begin by adapting the definition of assertions and assorters to the language of census RLAs. When auditing elections, an assorter is defined as a non-negative function over the set of possible ballots a voter may cast. When auditing a census, we define an assorter as a non-negative function over the set of all households, meaning  $a\colon H\rightarrow [0, \infty)$. An assorter $a$ satisfies the assertion $\frac{1}{|H|}\sum_{h\in H}{a(h)} > \half$ if and only if some condition regarding the allocation of representatives to the federal states is true. 
\begin{definition}
    A set of functions: $a_1,...,a_\ell:H\rightarrow [0,\infty)$ are \textbf{census assorters} if the allocation of representatives according to the census and the PES match iff for all $k\in [\ell]$:
$$
    \frac{1}{|H|}\sum_{h\in H}{a_k(h)} > \half.
$$ 
These $\ell$ inequalities are referred to as the \textbf{census assertions}.
\end{definition}

\subsubsection{Designing Household-Level Assorters}
\label{sec: census household-level assorter}

In this section, we present assorters that can be used for a census RLA in our described model. In~\Cref{sec: census assorters}, we use these assorters to develop new, more efficient assorters which are used during the census RLA method described in~\Cref{sec: census RLA description}.

These assorters are developed by reducing the problem of confirming the census' allocation of representatives to the problem of verifying that a set of linear inequalities are all true (\Cref{thm: census assorters 1}). Once we have such inequalities, we use the method described in~\Cref{sec: finding assorters} to convert them to equivalent SHANGRLA assertions, giving us our census assorters.

\begin{theorem} \label{thm: census assorters 1}
    Assume the PES surveyed all households. The allocation of representatives according to the census and according to the PES match, if and only if for any two states $s_1,s_2\in\mathcal{S}$:
    \begin{align} \label{eq: census linear ineq}
        \frac{\sum_{h\in H} g_{s_1}^{PES}(h) + c_{s_1}}{d(r^{cen}(s_1))} > \frac{\sum_{h\in H} g_{s_2}^{PES}(h) + c_{s_2}}{d(r^{cen}(s_2)+1)}.
    \end{align}
    Where $r^{cen}(s)$ is the number of representatives that state $s$ is allocated according to the census. The rest of the notation is defined in~\Cref{sec: census RLA model}.
\end{theorem}
\begin{proof}
    First, assume that the two allocations of representatives match. Examine the imaginary table with which representatives are allocated to states according to the PES, as described in~\Cref{sec: census RLA model}. Recall that each state has exactly its first $r^{PES}(s)$ cells colored. Since we assume that for any $s\in\mathcal{S}$, $r^{PES}(s)=r^{cen}(s)$, we have it that for any $s_1,s_2\in \mathcal{S}$, the cell at index $[s_1,r^{cen}(s_1)]$ is colored, while the cell at $[s_2,r^{cen}(s_2)+1]$ is not. Since the colored cells are the ones which hold the largest values in the table, the cell at $[s_1,r^{cen}(s_1)]$ has a larger value than the cell at $[s_2,r^{cen}(s_2)+1]$. Writing these values out results exactly in \eqref{eq: census linear ineq}- the larger term is the value at $[s_2,r^{cen}(s_2)+1]$, and the smaller is the value at $[s_1,r^{cen}(s_1)]$.

    Towards proving the other direction of the equivalence, we show that if~\eqref{eq: census linear ineq} is true for any $s_1,s_2\in \mathcal{S}$, then a certain condition~\eqref{eq: census assorter condition} holds for any $s_1,s_2$. We then show that if this condition is true, then the allocation of representatives according to the census and according to the PES match.
    \begin{claim} \label{claim: census move-seat}
        Let $r^{PES}(s)$ be the number of representatives a state $s$ is allocated according to the full PES results. For any $s_1,s_2\in \mathcal{S}$, if \eqref{eq: census linear ineq} is true then:
        \begin{align} \label{eq: census assorter condition}
            \left(r^{PES}(s_1) \geq r^{cen}(s_1)\right) \vee \left(r^{PES}(s_2) \leq r^{cen}(s_2)\right)
        \end{align}
    \end{claim}
    \begin{proof}
        Assume towards contradiction that for some $s_1, s_2 \in \mathcal{S}$, the condition~\eqref{eq: census assorter condition} is false, meaning that its negation, $\left(r^{PES}(s_1) < r^{cen}(s_1)\right) \wedge \left(r^{PES}(s_2) > r^{cen}(s_2)\right)$, is true.

        Examine the table used to allocate representatives to states according to the PES results. According to this table, $s_2$ is awarded $r^{PES}(s_2)$ representatives. Since $r^{PES}(s_2) > r^{cen}(s_2)$, and since the row $s_2$ has exactly its first $r^{PES}(s_2)$ cells colored, the cell at $[s_2, r^{cen}(s_2)+1]$ is colored. Additionally, since $s_1$ was awarded exactly $r^{PES}(s_1)$ representatives and since $r^{PES}(s_1) < r^{cen}(s_1)$, the cell at $[s_1, r^{cen}(s_1)]$ is not colored.
    
        By the paragraph above, if $\left(r^{PES}(s_1) \geq r^{cen}(s_1)\right) \vee \left(r^{PES}(s_2) \leq r^{cen}(s_2)\right)$ is false, then the cell at $[s_2, r^{cen}(s_2)+1]$ is colored while the cell at $[s_1, r^{cen}(s_1)]$ is not. Since the colored cells are the ones which hold the largest values, it follows that the cell at $[s_2, r^{cen}(s_2)+1]$ has a larger value than the cell at $[s_1, r^{cen}(s_1)]$, meaning that:
        $$
            \frac{\sum_{h\in H} g_{s_1}^{PES}(h) + c_{s_1}}{d(r^{cen}(s_1))} \leq \frac{\sum_{h\in H} g_{s_2}^{PES}(h) + c_{s_2}}{d(r^{cen}(s_2)+1)}.
        $$
        The larger term in this inequality is the value at index $[s_2, r^{cen}(s_2)+1]$ and the smaller one is the value at index $[s_1, r^{cen}(s_1)]$. This contradicts~\eqref{eq: census linear ineq}, and thereby proves this claim.
    \end{proof}
    \begin{claim} \label{claim: allocations match}
        If \eqref{eq: census assorter condition} is true for any $s_1,s_2 \in \mathcal{S}$, then the allocation of representatives according to the census and according to the full PES are identical.
    \end{claim}
    \begin{proof}
        Assume towards contradiction that the two allocations are not identical. Therefore, there must be at least one state $s$ with $r^{PES}(s) \neq r^{cen}(s)$. If $r^{PES}(s) > r^{cen}(s)$, since the number of total representatives is constant, there must be another state $s'$ with $r^{PES}(s') < r^{cen}(s')$. Similarly, if $r^{PES}(s) < r^{cen}(s)$, there must be another state $s'$ with $r^{PES}(s') > r^{cen}(s')$. Either way, either $$\left(r^{PES}(s) \geq r^{cen}(s)\right) \vee \left(r^{PES}(s') \leq r^{cen}(s')\right)$$ 
        or 
        $$\left(r^{PES}(s') \geq r^{cen}(s')\right) \vee \left(r^{PES}(s) \leq r^{cen}(s)\right)$$ are false, meaning that~\eqref{eq: census assorter condition} is not true regarding all pairs of states. Thus, if~\eqref{eq: census assorter condition} is true for every pair of states, then the two allocations must be identical, completing the proof.
    \end{proof}
    Using these two claims, we can now complete the proof of this theorem. Assume~\eqref{eq: census linear ineq} is true for any pair of states. By~\Cref{claim: census move-seat},~\eqref{eq: census assorter condition} is also true for any pair of states, and by~\Cref{claim: allocations match}, this makes the allocation of representatives according to the census and according to the PES identical. This proves the other direction of the equivalence and concludes the proof of this theorem.
\end{proof}

\paragraph*{Finding the Assorters - Goal}

By~\Cref{thm: census assorters 1}, to prove that the census and PES lead to the same allocation of representatives, it suffices to verify that \eqref{eq: census linear ineq} is true for any pair of states. Next, for every pair of states $s_1,s_2\in\mathcal{S}$, we develop a SHANGRLA assertion which is equivalent to~\eqref{eq: census linear ineq}, giving us our census RLA assorters. Towards this goal, we slightly re-arrange~\eqref{eq: census linear ineq} into an equivalent form that is simpler to work with:
\begin{align} \label{eq: census linear ineq 2}
    \frac{\sum_{h\in H} g_{s_1}^{PES}(h)}{d(r^{cen}(s_1))} - \frac{\sum_{h\in H} g_{s_2}^{PES}(h)}{d(r^{cen}(s_2)+1)} > \frac{c_{s_2}}{d(r^{cen}(s_2) + 1)} - \frac{c_{s_1}}{d(r^{cen}(s_1))}.
\end{align} 

For every $s_1, s_2\in \mathcal{S}$, we wish to find a non-negative function $a^{PES}_{s_1,s_2}$ such that~\eqref{eq: census linear ineq 2} is equivalent to:
$$
    \frac{1}{|H|}\sum_{h\in H}{a^{PES}_{s_1,s_2}(h)} > \half.
$$
This is done using the method described in~\Cref{sec: finding assorters}, which converts linear inequalities regarding ballot tallies to SHANGRLA assertions. This method originally applies on elections, and not on censuses. To use it here, we need to view the census as a an election. 

\paragraph*{From Census to Elections}

To view the census as an election, we define an election system where each ballot corresponds to a household in the census, and the elections result in an allocation of representatives to states, just as is done in the census. In this election system, each ballot holds the state and number of residents of its corresponding household according to the PES.

More formally, the set possible ballots a voter may cast in these elections is $\mathcal{C}:= \mathcal{S} \times \{[g^{max}| \cup\{0\}\}$, and the ballots that were truly cast are $B=\{(s_h, g^{PES}(h))\}_{h\in H}$, where $s_h$ denotes the state of household $h$. Additionally, we use $v(s,k)$ to denote the number of ballots cast for $(s,k)$. Using this notation, we can rewrite~\eqref{eq: census linear ineq 2} as:
\begin{align} \label{eq: census linear ineq 3}
    \sum^{g^{max}}_{k=0}{\left(\frac{1}{d(r^{cen}(s_1))}k v(s_1,k)  - \frac{1}{d(r^{cen}(s_2)+1)} k v(s_2,k)\right)} > \frac{c_{s_2}}{d(r^{cen}(s_2) + 1)} - \frac{c_{s_1}}{d(r^{cen}(s_1))}.
\end{align}
This is equivalent to~\eqref{eq: census linear ineq 2} since both $\sum^{g^{max}}_{k=0}{k v(s,k)}$ and $\sum_{h\in H}{g^{PES}_{s}(h)}$ count the number of residents at state $s$.

\eqref{eq: census linear ineq 3} is a linear inequality regarding the vote tallies in some elections, which allows us to apply the method from~\Cref{sec: finding assorters} to convert it to an equivalent SHANGRLA assertion. The resulting assertions are $\frac{1}{|B|}\sum_{b\in B}{a_{s_1,s_2}(b)}>\half$ for each $s_1,s_2\in\mathcal{S}$, with:
\begin{align} \label{eq: census assorter bad}
        a_{s_1,s_2}(s, k) =
        \frac{k\bbOne_{s=s_1}}{c d(r^{cen}(s_1))}+\frac{g^{max} - k\bbOne_{s=s_2}}{cd(r^{cen}(s_2)+1)},
 \end{align}
where we denote:
\begin{align} \label{eq: c definition}
    c := 2\left(\frac{g^{max}}{d(r^{cen}(s_2)+1)}+\frac{c_{s_2}}{|H| d(r^{cen}(s_2) + 1)} - \frac{c_{s_1}}{|H|d(r^{cen}(s_1))}\right)
\end{align}

\paragraph*{Back From Elections to the Census}

The assorter above is intended for our imagined elections. We now wish to convert it to an assorter which operates on households instead of ballots. Towards this purpose, observe that for a household $h$ and its conversion to a ballot $(s_h, k)$, we have for any $s\in \mathcal{S}$: $k\bbOne_{s_h=s} = g^{PES}_{s}(h)$. Meaning that defining $a^{pes}_{s_1,s_2}$ to operate directly on the households results in the following:
\begin{definition} \label{def: census assorters prelim}
    The census assorter $a^{PES}_{s_1,s_2}$ is defined as:
    $$
        a^{PES}_{s_1,s_2}(h) :=
        \frac{g^{PES}_{s_1}(h)}{c d(r^{cen}(s_1))}+\frac{g^{max} - g^{PES}_{s_2}(h)}{cd(r^{cen}(s_2)+1)},
    $$
    where $r^{cen}(s)$ is the number of representatives state $s$ is awarded according to the census, $c$ is defined as in~\eqref{eq: c definition}.
\end{definition}

Since we used the method from~\Cref{sec: finding assorters}, we have, for any two states $s_1, s_2$:
\begin{align*}
    \left(\frac{1}{|H|}\sum_{h\in H}{a^{PES}_{s_1,s_2}(h)} > \half \right) \Longleftrightarrow
    \left(\frac{\sum_{h\in H} g_{s_1}^{PES}(h) + c_{s_1}}{d(r^{cen}(s_1))} > \frac{\sum_{h\in H} g_{s_2}^{PES}(h) + c_{s_2}}{d(r^{cen}(s_2)+1)}\right).
\end{align*}
And by~\Cref{thm: census assorters 1}, the allocation of representatives according to the census and PES match iff for all $s_1,s_2\in \mathcal{S}$:
\begin{align} \label{eq: census assorter valid 1}
     \frac{1}{|H|}\sum_{h\in H}{a_{s_1, s_2}(h)} > \half.
\end{align}

\subsubsection{Designing Household-Comparison Assorters}
\label{sec: census comparison assorter}
For each assorter $a^{PES}_{s_1, s_2}$ from~\Cref{def: census assorters prelim}, we now define a new assorter $A_{s_1,s_2}$ which can also be used to audit the same census. $A_{s_1,s_2}$ has a significant advantage over $a^{PES}_{s_1,s_2}$, which motivates us to use it instead. Each assorter $a^{PES}_{s_1,s_2}$ essentially audits the number of residents per household according to the PES, without using the per-household census data. Meanwhile, $A_{s_1,s_2}$ audits the discrepancy in the number of household members between the census and the PES. Since we typically expect this discrepancy to be small, this yields a more stable and efficient audit.

\paragraph*{Some Intuition}

This next part is only meant to explain how these final assorters are deduced, and not to prove that auditing them results in a valid census-RLA. A formal proof that these assorters satisfy the condition described above is shown in~\Cref{thm: census assorters}.

First, note that each assorter $a_{s_1,s_2}^{PES}$ can also be defined over the census population counts instead of the PES counts. We denote this as $a_{s_1,s_2}^{cen}$:
\begin{definition}
    \begin{align*}
        a_{s_1,s_2}^{cen}(h) &:= \frac{g^{cen}_{s_1}(h)}{c d(r^{cen}(s_1))}+\frac{g^{max} - g^{cen}_{s_2}(h)}{cd(r^{cen}(s_2)+1)}.
    \end{align*}    
\end{definition}

As mentioned previously, we would like $A_{s_1, s_2}$ to operate over the per-household discrepancy between the census and the PES as it relates to $a^{PES}_{s_1,s_2}$. Meaning, it should operate over:
$$
    a_{s_1, s_2}^{PES}(h) - a_{s_1, s_2}^{cen}(h),
$$

A simple way of doing so is to define our new assorter $A_{s_1,s_2}$ to have a similar form to the Batchcomp assorter from~\Cref{def: Batchcomp assorter}:
\begin{align} \label{eq: census A base form}
    A_{s_1,s_2}(h) = \half + \frac{m_{s_1,s_2} + a^{PES}_{s_1,s_2}(h) - a^{cen}_{s_1,s_2}(h)}{\cdots},
\end{align}
with some choice of denominator in place of $(\cdots)$. $m_{s_1,s_2}$ here is the margin of $a^{cen}_{s_1,s_2}$:
\begin{align} \label{eq: census rep margin}
    m_{s_1,s_2} := \frac{1}{|H|}\sum_{h'\in H}{a^{cen}_{s_1, s_2}(h')} - \half.
\end{align}
Observe that for any $s_1,s_2\in \mathcal{S}$, $m_{s_1,s_2}>0$. This is true since otherwise, $\frac{1}{|H|}\sum_{h'\in H}{a^{cen}_{s_1, s_2}(h')} \leq \half$, meaning that if the PES and census results completely match on all households, we have $\frac{1}{|H|}\sum_{h'\in H}{a^{PES}_{s_1, s_2}(h')} \leq \half$, contradicting~\eqref{eq: census assorter valid 1}. 

With this definition of $A_{s_1,s_2}$, for any positive denominator we use in place of the dots ($\cdots$), we would have:
$$
    \left(\frac{1}{|H|}\sum_{h\in H}{a_{s_1, s_2}(h)} > \half \right) \Longleftrightarrow \left(\frac{1}{|H|}\sum_{h\in H}{A_{s_1, s_2}(h)} > \half\right) ,
$$
as we show while proving~\Cref{claim: census assorter valid 2}. What remains is to choose the denominator.

\paragraph*{Choosing the Denominator}

Ideally, we would like $A_{s_1,s_2}$ to return large values when the census and the PES agree on the number of residents of a certain household, since this would cause the audit to approve a correct census (one that matches the PES) sooner. When the census and the PES agree on some household $h$, we have:
$$
    A_{s_1,s_2}(h) = \half + \frac{m_{s_1,s_2} + \overbrace{a^{PES}_{s_1,s_2}(h) - a^{cen}_{s_1,s_2}(h)}^{=0}}{\cdots} = \half + \frac{m_{s_1,s_2}}{\cdots}.
$$
And since, as explained right after after~\eqref{eq: census rep margin}, $m_{s_1,s_2}>0$, this value will be larger the smaller our chosen denominator is. However, if we choose a denominator which is too small, $A_{s_1,s_2}$ could potentially return negative values. By these two observations, we should choose the denominator to be the smallest positive such that $A_{s_1,s_2}$ is non-negative. To find which value satisfies this, we find the minimal value that the nominator may have. By the definition of $a^{PES}_{s_1,s_2}$ and $a^{cen}_{s_1,s_2}$, the value of the nominator is:
$$
   m_{s_1,s_2} + a^{PES}_{s_1,s_2}(h) - a^{cen}_{s_1,s_2}(h) \\[1.5ex]
    = m_{s_1,s_2} + \frac{g^{PES}_{s_1}(h) - g^{cen}_{s_1}(h)}{c d(r^{cen}(s_1))}+\frac{g^{cen}_{s_2}(h) - g^{PES}_{s_2}(h)}{cd(r^{cen}(s_2)+1)}.
$$
$h$ is either from $s_1$, from $s_2$ or from neither of them. If it's from neither, than this expression is equal $m_{s_1,s_2}$. If it's from $s_1$, then:
\begin{align*}
    m_{s_1,s_2} + \frac{\overbrace{g^{PES}_{s_1}(h)}^{\geq 0} - \overbrace{g^{cen}_{s_1}(h)}^{\leq g^{max}}}{c d(r^{cen}(s_1))}+\frac{\overbrace{g^{cen}_{s_2}(h) - g^{PES}_{s_2}(h)}^{=0}}{cd(r^{cen}(s_2)+1)}
    &\geq m_{s_1,s_2} - \frac{g^{max}}{c d(r^{cen}(s_1))}, \\
\intertext{where $g^{max}$ is the maximal number of residents a single household may have. If $h$ is from $s_2$, then:} \\
    m_{s_1,s_2} + \frac{\overbrace{g^{PES}_{s_1}(h) - {g^{cen}_{s_1}(h)}}^{=0}}{c d(r^{cen}(s_1))}+\frac{\overbrace{g^{cen}_{s_2}(h)}^{\geq 0} - \overbrace{g^{PES}_{s_2}(h)}^{\leq g^{max}}}{cd(r^{cen}(s_2)+1)}
    & \geq m_{s_1,s_2} - \frac{g^{max}}{cd(r^{cen}(s_2)+1)}.
\end{align*}
So for any $h\in H$:
\begin{align} \label{eq: A min nominator}
    &m_{s_1,s_2} + a^{PES}_{s_1,s_2}(h) - a^{cen}_{s_1,s_2}(h) \nonumber \\
    \geq & \min \left\{m_{s_1,s_2}, m_{s_1,s_2} - \frac{g^{max}}{cd(r^{cen}(s_2)+1)}, m_{s_1,s_2} - \frac{g^{max}}{c d(r^{cen}(s_1))} \right\}
\end{align}
We can now set the denominator to be the smallest value for which $A_{s_1,s_2}$ is non-negative:
$$
    A_{s_1,s_2}(h) := \half + \frac{m_{s_1,s_2} + a^{PES}_{s_1,s_2}(h) - a^{cen}_{s_1,s_2}(h)}{-2\min \left\{m_{s_1,s_2} - \frac{g^{max}}{cd(r^{cen}(s_2)+1)}, m_{s_1,s_2} - \frac{g^{max}}{c d(r^{cen}(s_1))}, m_{s_1,s_2} \right\}}.
$$
For brevity, we denote:
\begin{align} \label{eq: z def}
    z_{s_1,s_2} := \max \left\{\frac{g^{max}}{cd(r^{cen}(s_2)+1)}, \frac{g^{max}}{c d(r^{cen}(s_1))}, 0 \right\}.
\end{align}
And write $A_{s_1,s_2}$ as follows:
\begin{definition}
    The census comparison assorter for states $s_1, s_2\in \mathcal{S}$ is defined as:
    $$
        A_{s_1,s_2}(h) := \half + \frac{m_{s_1,s_2} + a^{PES}_{s_1,s_2}(h) - a^{cen}_{s_1,s_2}(h)}{2(z_{s_1,s_2} - m_{s_1,s_2})}.
    $$    
\end{definition}
We now prove that the assorters $\{A_{s_1,s_2}\,|\,s_1,s_2 \in \mathcal{S} \times \mathcal{S}, s_1\neq s_2\}$ are valid for auditing the census' allocation of representatives to the federal states.

\begin{theorem} \label{thm: census assorters}
    Assume that the PES surveyed all households. The assorters $\{A_{s_1,s_2}\,|\,s_1,s_2\in\mathcal{S}\}$, as defined above, are all non-negative and satisfy the following condition: The allocation of representatives according to the census and the PES match iff for all $s_1,s_2\in \mathcal{S}$:
    $$
        \frac{1}{|H|}\sum_{h\in H}{A_{s_1, s_2}(h)} > \half.
    $$
\end{theorem}
\begin{proof}
    We first prove that $A_{s_1,s_2}$ is non-negative for any $s_1,s_2\in \mathcal{S}$, and then    show that the equivalence above holds.

    \begin{claim}
        For any $s_1,s_2\in \mathcal{S}$, $A_{s_1,s_2}$ is non-negative.
    \end{claim}
    \begin{proof}
        Fix two states $s_1,s_2\in \mathcal{S}$. By~\eqref{eq: A min nominator} and by the definition of $z_{s_1,s_2}$:
        \begin{align}
            m_{s_1,s_2} + a^{PES}_{s_1,s_2}(h) - a^{cen}_{s_1,s_2}(h) 
            \geq m_{s_1, s_2} - z_{s_1,s_2}.
        \end{align}
        Meaning that for any $h\in H$:
        \begin{align*}
            A_{s_1, s_2} &= \half + \frac{m_{s_1,s_2} + a^{PES}_{s_1,s_2}(h) - a^{cen}_{s_1,s_2}(h)}{2(z_{s_1,s_2} - m_{s_1,s_2})} \\[1.5ex]
            & \geq \half + \frac{m_{s_1, s_2} - z_{s_1,s_2}}{2(z_{s_1,s_2} - m_{s_1,s_2})} \\
            &= 0,
        \end{align*}
        proving the claim.
    \end{proof}
    \begin{claim} \label{claim: census assorter valid 2}
        The allocation of representatives according to the census and the PES match iff for all $s_1,s_2\in \mathcal{S}$:
        $$
            \frac{1}{|H|}\sum_{h\in H}{A_{s_1, s_2}(h)} > \half.
        $$
    \end{claim} 
    \begin{proof}
        In~\Cref{sec: census household-level assorter}, we developed assorters $a^{PES}_{s_1, s_2}$ such that the allocation of representatives according to the census and the PES match iff for all $s_1,s_2\in \mathcal{S}$:
        $$
            \frac{1}{|H|}\sum_{h\in H}{a^{PES}_{s_1, s_2}(h)} > \half
        $$
        Therefore, to prove this claim, it suffices to prove that for every $s_1, s_2\in \mathcal{S}$:
        $$
            \left(\frac{1}{|H|}\sum_{h\in H}{A_{s_1, s_2}(h)} > \half \right) \Longleftrightarrow \left(\frac{1}{|H|}\sum_{h\in H}{a_{s_1, s_2}(h)} > \half \right).
        $$
        Fix any two federal-states $s_1, s_2 \in \mathcal{S}$. We show that the two inequalities above are equivalent:
        \begin{align*}
            & \frac{1}{|H|}\sum_{h\in H}{A_{s_1, s_2}(h)} > \half \\[1.5ex]
            \Longleftrightarrow & \frac{1}{|H|}\sum_{h\in H}{\left(\half + \frac{m_{s_1,s_2} + a^{PES}_{s_1,s_2}(h) - a^{cen}_{s_1,s_2}(h)}{2(z_{s_1,s_2} - m_{s_1,s_2})}\right)} > \half \\[1.5ex]
            \Longleftrightarrow & \frac{1}{|H|}\sum_{h\in H}{\frac{m_{s_1,s_2} + a^{PES}_{s_1,s_2}(h) - a^{cen}_{s_1,s_2}(h)}{2(z_{s_1,s_2} - m_{s_1,s_2})}} > 0.
            \intertext{Now, using the definition of $m_{s_1,s_2}$ and re-arranging the summation yields the desired  equivalence:}
            \Longleftrightarrow & \frac{1}{|H|}\sum_{h\in H}{\frac{\frac{1}{|H|}\sum_{h'\in H}{a^{cen}_{s_1, s_2}(h')} - \half + a^{PES}_{s_1,s_2}(h) - a^{cen}_{s_1,s_2}(h)}{2(z_{s_1,s_2} - m_{s_1,s_2})}} > 0   \\[1.5ex]
            \Longleftrightarrow & \frac{\frac{1}{|H|}\sum_{h'\in H}{a^{cen}_{s_1, s_2}(h')} + \frac{1}{|H|}\sum_{h\in H}{a^{PES}_{s_1, s_2}(h)} - \frac{1}{|H|}\sum_{h\in H}{a^{cen}_{s_1, s_2}(h)} - \half}{2(z_{s_1,s_2} - m_{s_1,s_2})} > 0   \\[1.5ex]
            \Longleftrightarrow & \frac{\frac{1}{|H|}\sum_{h\in H}{a^{PES}_{s_1, s_2}(h)} - \half}{2(z_{s_1,s_2} - m_{s_1,s_2})} > 0   \\[1.5ex]
            \Longleftrightarrow & \frac{1}{|H|}\sum_{h\in H}{a^{PES}_{s_1, s_2}(h)} > \half,
        \end{align*}
        The last transition relies on the 
        fact that $z_{s_1,s_2}>m_{s_1,s_2}$, which is true since:
        \begin{align*}
            z_{s_1,s_2} \geq \max_{h\in H}{a^{cen}(h)} \geq m_{s_1,s_2}.
        \end{align*}
        This concludes the proof of this claim.
    \end{proof}
    The combination of these two claims completes this theorem's proof.
\end{proof}

\subsubsection{Census RLA Description}
\label{sec: census RLA description}
The algorithm presented next is a slightly altered version of the ALPHA martingale test from~\Cref{sec:alpha test}, when thinking of each household as a ballot whose content is the household's state and its number of residents.

\paragraph*{Adapting the ALPHA Martingale Test to Censuses}

Unlike an election RLA, where the paper-backup ballots are read as the audit is ran, a census RLA is performed after the PES, meaning that the households were re-surveyed before the audit begins. We handle this by sampling households such that from the auditor's perspective, if it doesn't know which households were surveyed by the PES, it receives a previously unsampled household uniformly at random. This is performed according to a subroutine that takes as its input 4 arguments:
\begin{itemize}
    \item $H^1$: the set of households that have yet to be audited.
    \item $H^{cen}$: the list of households according to the census.
    \item $H^{PES}$: the list from which the PES randomly chose the households it surveyed ($\tilde{H}^{PES}$).
    \item $\tilde{H}^{PES}$: a list of households that were surveyed during the PES.
\end{itemize}

\fbox{\begin{minipage}{40em}
  \textbf{Sample\_Household($H^1, H^{cen}, H^{PES}, \tilde{H}^{PES}$):} \\[1.5ex]
  With probability $\frac{|H^{PES} \cap H^1|}{|H^1|}$, sample a household uniformly at random from $\tilde{H}^{PES} \cap H^1$.  \\[1.2ex]
  Otherwise, sample a household uniformly at random from $\left(H \setminus H^{PES}\right)\cap H^1$.
\end{minipage}}

Additionally, recall that instead of pre-setting the risk-limit $\alpha$, the algorithm outputs the smallest value $\alpha$ with which it can approve the representative allocation of the census, as described in the~\hyperlink{census rla guarantee}{Census RLA guarantee}.
This is done by keeping, at all times, the risk-limit of each assertion, $\frac{1}{T^{max}}$. When the audit runs out of households to sample, the procedure outputs the maximal risk-limit across all the assertions (step~\ref{census rla step8}).

Before presenting the algorithm, recall that the variables $\mu, \eta$ and $U$ represent guesses regarding the value that their corresponding assorter will return in the next iteration. $\mu$ is that guess given that the assorter has a mean of exactly $\half$, $\eta$ is the guess assuming that the census and PES agree on all households, and $U$ is a parameter which controls how significantly $T$ changes per iteration. $T$ is the inverse of the risk-limit with which we can approve that its corresponding assertion is true, and $T^{max}$ holds the maximal value of $T$ throughout the audit.

\paragraph*{The Census RLA Algorithm}

\begin{enumerate}
    \item \textbf{Initialization}
    \begin{enumerate}[label*=\arabic*.]
        \item Initialize $H^1=H$ and $H^0=\emptyset$. $H^1$ holds the households that have yet to be audited, and $H^0$ holds the households which were already audited.
        \item For each $(s_1, s_2) \in \mathcal{S}\times \mathcal{S}$ s.t. $s_1\neq s_2$ initialize:
        \begin{itemize}
            \item $T_{s_1, s_2} := 1$.
            \item $T^{max}_{s_1, s_2} :=  1$.
            \item $\mu_{s_1, s_2} := \half$.
            \item $\eta_{s_1, s_2} := \half + \frac{m_{s_1,s_2}}{2(z_{s_1,s_2} - m_{s_1,s_2})}$.
            \item $U_{s_1, s_2} := \half + \frac{m_{s_1, s_2} + \delta}{2(z_{s_1,s_2} - m_{s_1,s_2})}$, where $\delta > 0$. 
        \end{itemize}
        For definitions of $m_{s_1,s_2}$ and $z_{s_1,s_2}$ see~\eqref{eq: census rep margin} and~\eqref{eq: z def}.
    \end{enumerate}
    \item \textbf{Auditing Stage:} While $H^1 \cap \tilde{H}^{PES}\neq \emptyset$, perform:
    \begin{enumerate}[label*=\arabic*.]
        \item\label{census rla step3} Sample a household $h$ using the subroutine $\text{Sample\_Household}(H^1, H^{cen}, H^{PES}, \tilde{H}^{PES})$.
        \item Remove $h$ from $H^1$ and add it to $H^0$.
        \item\label{census rla tmax} For each $s_1, s_2$, update $T_{s_1,s_2}$ and $T^{max}_{s_1,s_2}$:
        \begin{itemize}
            \item $T_{s_1,s_2} \leftarrow T_{s_1,s_2} \frac{1}{U_{s_1,s_2}}\left(A_{s_1,s_2}(h) \frac{\eta_{s_1,s_2}}{\mu_{s_1,s_2}} + (U_{s_1,s_2}-A_{s_1,s_2}(h)) \frac{U_{s_1,s_2} - \eta_{s_1,s_2}}{U_{s_1,s_2} - \mu_{s_1,s_2}}\right)$.
            \item $T^{max}_{s_1,s_2} \leftarrow \max\left\{T^{max}_{s_1,s_2}, T_{s_1,s_2}\right\}$.
        \end{itemize}
        \item For each $s_1, s_2$ update $\mu_{s_1, s_2}, \eta_{s_1, s_2}$ and $U_{s_1, s_2}$ to be:
        \begin{itemize}
            \item $\mu_{s_1, s_2} \leftarrow \frac{\half |H| - \sum_{h'\in H^0}{A_{s_1, s_2}(h')}}{|H_1|}$.
            \item $\eta_{s_1, s_2}  \leftarrow max\left\{\half + \frac{m_{s_1, s_2}}{2(z_{s_1, s_2}-m_{s_1, s_2})}, \mu_{s_1, s_2} + \epsilon\right\}$.
            \item $U_{s_1, s_2} \leftarrow max\{U_{s_1, s_2}, \eta_{s_1, s_2} + \epsilon\}$.
        \end{itemize}
        Where $\epsilon$ is some very small positive meant to ensure that $\mu_{s_1, s_2} < \eta_{s_1, s_2} < U_{s_1, s_2}$. We assume these variables are updated according to the order of their listing above.
        \item\label{census force approve step} For any $s_1, s_2$, if $\mu_{s_1, s_2} < 0$, we must have $\frac{1}{|H|}\sum_{h\in H}{A_{s_1,s_2}(h)}>\half$, so set $T_{s_1,s_2}^{max}=\infty$. This means that we can approve this assertion with risk-limit 0 - this assertion is necessarily true.
    \end{enumerate}
    
    \item \label{census rla step8} \textbf{Output:} The result of the audit is the maximal risk-limit across all assertions:
    $$
        \max_{s_1,s_2 \in \mathcal{S}}\left\{\frac{1}{T^{max}_{s_1,s_2}}\right\}.
    $$
\end{enumerate}

Before proving that this algorithm satisfies the~\hyperlink{census rla guarantee}{census RLA guarantee}, we make one preliminary claim.

\begin{claim} \label{claim: census household sampling}
    If the PES is conducted over a random subset of households from $H^{PES}$, the sample\_household subroutine returns a household which is selected uniformly at random from $H^1$.  
\end{claim}
\begin{proof}
    Fix any call to sample\_household during the audit. During this proof, we denote $H^1$ as the set $H^1$ is while executing this call. Denote by $j$ the number of households from $H^1$ that were surveyed during the PES. 
    
    Examine any unsampled household from the PES household list $h\in H^{PES}\cap H^1$. $h$ will be returned by sample\_household if and only if these 3 events all occur:
    \begin{itemize}
        \item \textbf{$h$ was surveyed during the PES}- there are $j$ households in $H^1$ that were surveyed in the PES, and $|H^1\cap H^{PES}|$ households in $H^1$ that were considered for surveying by the PES. This puts the probability of $h$ being surveyed during the PES at $\frac{j}{|H^1 \cap H^{PES}|}$.
        \item \textbf{sample\_household sampled a household that was surveyed during the PES} - this happens w.p.\ $\frac{|H^1 \cap H^{PES}|}{|H^1|}$.
        \item \textbf{sample\_household chose $h$, given that the two previous events happened} - there are $j$ households in $H^1$ that were surveyed during the PES, so the probability of this occurring is $\frac{1}{j}$.
    \end{itemize}
    The probability of $h$ getting sampled by sample\_household is therefore:
    \begin{align*}
        \frac{j}{|H^1\cap H^{PES}|} \frac{|H^1 \cap H^{PES}|}{|H^1|} \frac{1}{j} = \frac{1}{|H^1|}.
    \end{align*}
    This establishes that the probability of any single household from $H^1\cap H^{PES}$ being returned is $\frac{1}{|H_1|}$. Now, since sample\_household may only return households from $H^1$, and since any household from $H^1 \setminus H^{PES}$ is returned with equal probability, the probability of any household from $H^1$ getting returned must be $\frac{1}{|H^1|}$. This shows that sample\_household returns a household uniformly at random from $H^1$, proving the claim.  
\end{proof}
We can now prove that this audit is a census RLA.
\begin{theorem}
    For any nation with federal-states $\mathcal{S}$, that allocates representatives to its federal-states in proportion to their population using a highest averages method, this suggested census RLA satisfies the~\hyperlink{census rla guarantee}{census RLA guarantee}: For any $0 < \alpha \leq 1$, if running the PES over all households would lead to a different allocation of representatives than the census, then the probability that a census RLA returns a value $\alpha'$ such that $\alpha' \leq \alpha$ is at most $\alpha$.
\end{theorem}
\begin{proof}
    This proof is similar to the proof of correctness for the ALPHA martingale test in~\Cref{thm: alpha martingale test}. The key point here is that from the perspective of the auditor, the households it receives are sampled uniformly at random and without replacement from the set of all households. Given this observation, the census RLA algorithm can be viewed as a regular election RLA, where each ballot corresponds to a household and contains the household's state and number of residents.
    
    Fix a highest-averages allocation method (meaning a monotonically increasing function $d:\bbN \rightarrow \bbN$ and $c_s\in \bbR$ for each $s\in \mathcal{S}$), lists of households $H^{cen}, H^{PES}$ and census results $g^{cen}$ as specified in~\Cref{sec: census RLA model}. Assume that extending the PES such that it surveys all households leads to a different allocation of representatives than the census, and denote the function which returns these full PES results as $g^{PES}:H^{PES}\rightarrow [g^{max}]\cup\{0\}$. Let $\tilde{H}^{PES}$ be a set of randomly selected households of a pre-determined size from $H^{PES}$ which were actually surveyed during the PES, and fix $\alpha \in [0,1]$. We wish to show that when running the census RLA over these inputs, the probability of the census RLA outputting a value $\alpha'$ s.t.\ $\alpha' < \alpha$ is at most $\alpha$. 
    
    By~\Cref{claim: census assorter valid 2}, since the representative allocations according to the census and according to the PES are different, there must be some census-assorter whose mean is at most $\half$. Denote this assorter as $A_{s'_1,s'_2}$:
    \begin{assumption}
        $$
            \frac{1}{|H|}\sum_{h\in H}{A_{s'_1,s'_2}(h)} \leq \half.
        $$
    \end{assumption}

    $T^{max}_{s'_1,s'_2}$ cannot become $\infty$ in step~\ref{census force approve step}, since that would mean, for any $H^0\subseteq H$:
    $$
       \mu_{s_1, s_2} < 0 \Longrightarrow \half |H| - \sum_{h'\in H^0}{A_{s_1, s_2}(h')} < 0 \Longrightarrow \half < \frac{1}{|H|}\sum_{h'\in H^0}{A_{s_1, s_2}(h')} \leq \frac{1}{|H|}\sum_{h'\in H}{A_{s_1, s_2}(h')}
    $$
    contradicting our assumption. Therefore, $T^{max}_{s'_1,s'_2}$ receives its final update in step~\ref{census rla tmax}.
    
    Denote the values $T_{s'_1,s'_2}$ has after every sampled household during the audit as $T^0_{s'_1,s'_2},T^1_{s'_1,s'_2},...,T^q_{s'_1,s'_2}$, where $T^0_{s'_1,s'_2}$ is its initial value and $q$ is any natural number. Since the algorithm outputs: 
    $$
        \max_{s_1,s_2 \in \mathcal{S}}\left\{\frac{1}{T^{max}_{s_1,s_2}}\right\},
    $$ 
    it outputs a value that is smaller than $\alpha$ only if at the end of the audit we have $T^{max}_{s'_1,s'_2} < \frac{1}{\alpha}$. By this and by the fact $T^{max}_{s'_1,s'_2}$ receives its final value at step~\ref{census rla tmax}, to prove this theorem, it suffices to prove that:
    \begin{align} \label{eq: crla goal}
        \Pr\left[\max_{j}\{T^j_{s'_1,s'_2}\} > \frac{1}{\alpha}\right] \leq \alpha.
    \end{align}
    To achieve this, we show that $T^0_{s'_1,s'_2},...,T^q_{s'_1,s'_2}$ is a non-negative supermartingale, and then use Ville's inequality, similarly to~\Cref{thm: alpha martingale test}. Note that the update rules of $T_{s_1,s_2}$ and of $T_k$ are identical in the census RLA and in the ALPHA martingale test respectively, and that the update rules for $U_{s_1,s_2}, \eta_{s_1,s_2}, \mu_{s_1,s_2}$ always maintain $U_{s_1,s_2} > \eta_{s_1,s_2} > \mu_{s_1,s_2}$.  Moreover, just like the ALPHA martingale test samples ballots randomly and without replacement, the census RLA samples households randomly and without replacement, by~\Cref{claim: census household sampling}.
    For these reasons,~\Cref{claim: t non negative} applies here too - given a non-negative assorter, our update rule for $T_{s_1,s_2}$ makes it non-negative itself, meaning that $T^0_{s'_1,s'_2},...,T^q_{s'_1,s'_2}$ are all non-negative.

    Additionally, for the same arguments as in~\Cref{claim: t non increasing}, we have it that for any $i\in[q]$:
    $$
    \bbE[T^i_{s'_1,s'_2}\,|\,T_{s'_1,s'_2}^1,...,T_{s'_1,s'_2}^{i-1}] \leq T_{s'_1,s'_2}^{i-1}.
    $$
    This is true because both $T_{s'_1,s'_2}$ here and $T_1$ in~\Cref{claim: t non increasing} belong to assertions which are false (their assorters have a mean of $\half$ or less), and are updated in the exact same manner. This makes $T^0_{s'_1,s'_2},...,T^q_{s'_1,s'_2}$ a non-negative supermartingale, meaning that by \hyperref[Ville's inequality]{Ville's inequality}~\cite{durrett2019probability}:
    \begin{align*}
        \Pr\left[\max_{j}\{T^j_{s'_1,s'_2}\} > \frac{1}{\alpha}\right] \leq \alpha \cdot T^0_{s'_1,s'_2} = \alpha.
    \end{align*}
    Which proves~\eqref{eq: crla goal} and thereby completes the proof of this theorem.
\end{proof}
\begin{comment}

\subsection{Batch Census RLAs}
\barkar{This (small) section is a nice combination of the two parts of the thesis, but I'm not sure if it's clear or helpful. Not sure if to keep it or cut it.}

Surveying a completely random sample of households can be a costly procedure, as they may be from each other. This leads some nations to perform their PES in batches, where entire blocks or neighborhoods are sampled and surveyed, instead of single households. Auditing a census under these conditions can be seen as a combination of a batch-level RLA and a census-RLA, as it would sample batches of households. An examination of such models is beyond the scope of this work, but is theoretically possible by combining the Batchcomp RLA from~\Cref{sec:batch RLA} method with this census RLA method.

To get assorters for a batch-census-RLA, we can convert the assorters from~\Cref{def: census assorters prelim} to assorters which operate on batches, as described in~\Cref{herreeee}. An audit for A census audit over these assorters would fulfill the census RLA guarantee as long as each single household had an equal probability of being surveyed during the PES.

This results in the assorter:
$$
   A_{s_1,s_2}(H_i) := \half + \frac{m_{s_1,s_2} + a^{PES}_{s_1,s_2}(H_i) - a^{rep}_{s_1,s_2}(H_i)}{2(\max_{j}\{a^{cen}_{s_1,s_2}(H_j)\} - m_{s_1,s_2})}.
$$
\end{comment}

\subsection{Census RLA Simulations} \label{sec: census RLA sim}

This section simulates the suggested census RLA on the Cypriot census and its resulting allocation of representatives to districts in the House of Representatives of Cyprus. Our original intention was to simulate the suggested census RLA method on the US census and its resulting allocation of representatives in the US House of Representatives to the states. This turned out to be infeasible, however, due to the relatively large number of states (50) and representatives (435). Allocating many representatives to many states increases the probability of there being a single representative whose allocation is determined by a very small number of state residents. Therefore, such systems require a very small enumeration error to change the census' allocation, and are therefore more difficult to audit. Using our suggested method on the American setting, a PES which surveys 10\% of households results in a risk-limit of only 0.75. In reality, the US PES surveys around 1\% of households~\cite{usPesSize}.

To show that the census RLA is useful in other cases, we chose to simulate the audit on the House of representatives of Cyprus, where 56 representatives are allocated to 5 districts. Due to the somewhat limited available resources in English regarding the Cypriot census, we view this as a pet-setting for testing our census RLA method, and not as a ready-as-is implementation.

\subsubsection{The House of Representatives of Cyprus} \label{sec: crla technical details}
The House of representatives of Cyprus is its sole legislating body. Officially, The house holds 80 seats, where 56 are to be elected by the Greek Cypriot community and 24 by the Turkish Cypriot community. In 1964, the Turkish-Cypriots withdrew from the political decision-making process, leaving their house seats vacant~\cite{charalambous2008house}.

Since then, the remaining 56 seats of the house are allocated to 5 districts. Currently, the allocation of seats to the districts is amended by law when found necessary, and does not change automatically following a census according to a set method. Our census RLA could be useful when performing these amendments, to ensure that the resulting allocation of seats to districts is sufficiently reliable.

\subsubsection{Data Generation and Technical Details}
The data used to perform this simulation is based on the population census conducted in 2021~\cite{cyprusCensus}. The Statistical Service of Cyprus publicly reports the total number of residents in every district, but not the individual household data, which the census RLA requires. To generate this data, we assumed that the number of residents per household distributes as it does in the United States, as reported by its census~\cite{usCensusHouseholds}. We additionally assumed that 1\% of households do not respond to the census and are counted as if they have no residents.
The per-household data used in these simulations was generated as follows:

\paragraph*{Generating the Census Data}

\begin{enumerate}
    \item The number of households per district was calculated by dividing the district's population by the expected number of residents per household.
    \item The number of residents in each household was drawn from the distribution specified in the US census~\cite{usCensusHouseholds}. 
    \item Due to the randomness involved in the previous step, the real census and our generated one might disagree on the population of the districts. To balance this, the constant of each district ($c_s$ in~\eqref{eq: census repr allocation} at~\Cref{sec: census RLA model}) was set as the difference between the population of the district according to the real census and according to our generated one. With this definition, the allocation of representatives to districts by the real census and by our generated one is necessarily identical.
\end{enumerate}

\paragraph*{Audit Parameters and Other Details}

We allocated representatives to districts using the D'Hondt method. D'Hondt was chosen since it's currently used in the Cypriot elections to allocate seats to political parties. The audit was run assuming that each household holds 15 residents at most, and with $\delta=10^{-10}$. We assumed that the list of households according to the census and the PES match, meaning $H^{PES}=H^{cen}$.
The simulation's code was written in Python and is available in \url{https://github.com/TGKar/Batch-and-Census-RLA}.

\subsubsection{Results}
We present the outputted risk-limit of the census RLA as a factor of the size of the PES. The x-axis shows the share of households that were surveyed of the PES, and each point in the plot represents the audit's outputted risk-limit when using a PES of that specified size. The results are averaged across 10 simulations. 

\subsubsection*{Results When Census and PES Completely Match}
\begin{center}
    \includegraphics[scale=0.42]{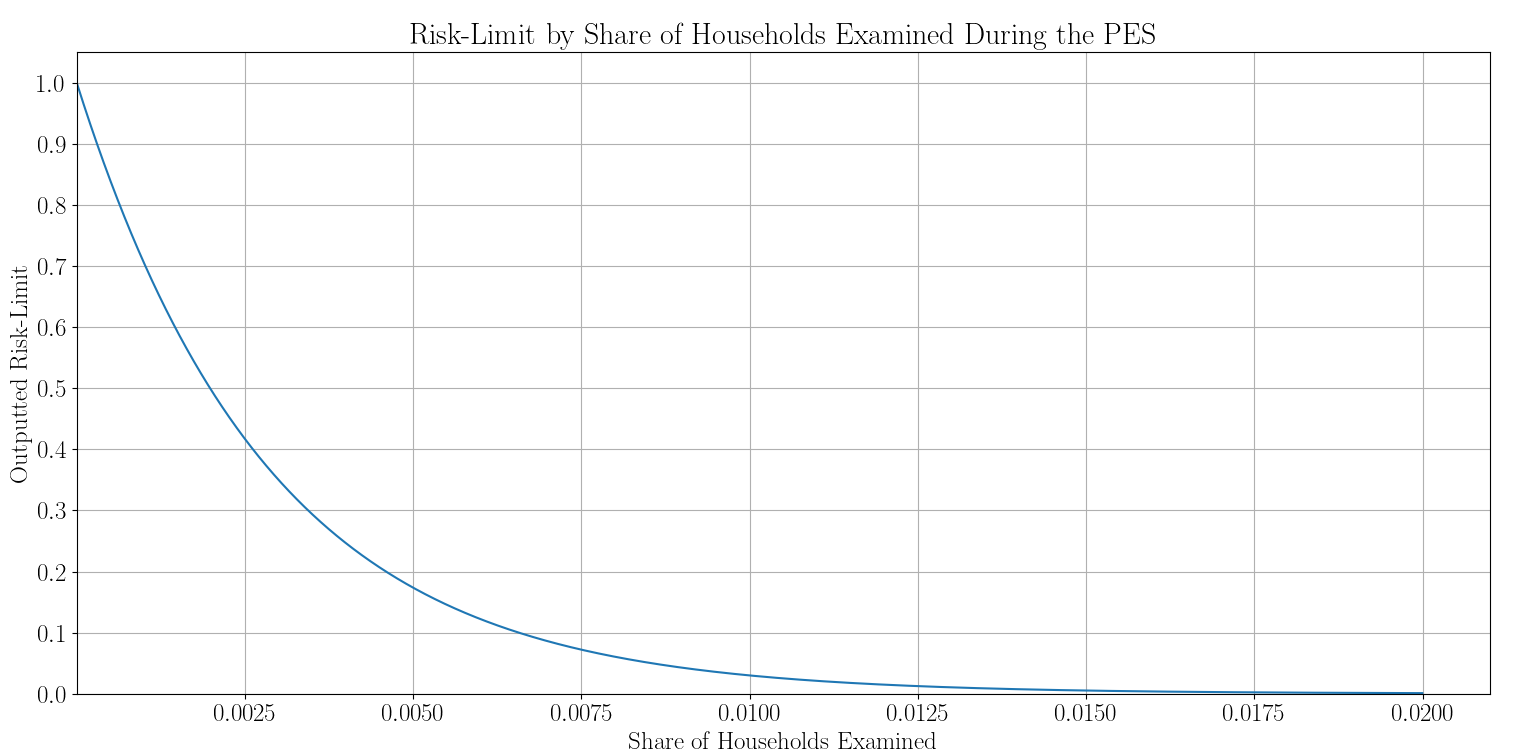}    
\end{center}
The audit's output when the census and PES fully agree on the number of residents in each household. Under these conditions, a PES which samples 0.66\% of households is sufficient for a risk limit of 0.1, and a sample of 0.87\% is sufficient for a risk-limit of 0.05. A PES often surveys around 1\% of households~\cite{hu2022determining}, meaning that our census RLA can confidently approve its resulting allocation of representatives to districts under these specified conditions.

\paragraph*{Results with Small Enumeration Errors}

The next plot shows the same results when the census and PES potentially disagree on 5\% of households. For these 5\% of households, which are selected uniformly at random, the number of residents according to the PES is re-drawn from the distribution of residents per household described in~\Cref{sec: crla technical details}. The following simulations were run over census and full PES results that lead to the same allocation of representatives to states. During the simulated census RLA, the audit only receives the PES results over a subset of randomly selected households.  
\begin{center}
    \includegraphics[scale=0.5]{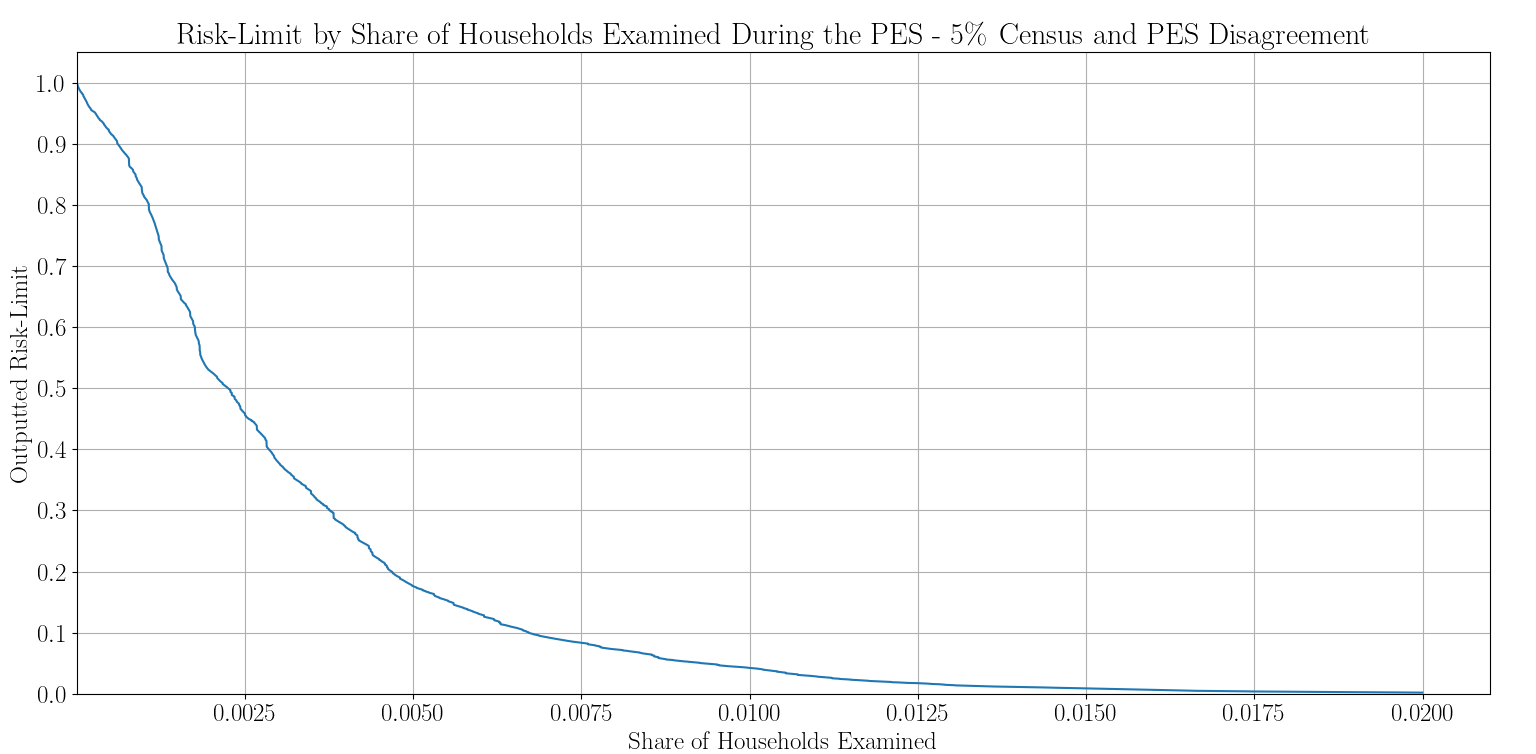}    
\end{center}
These results appear very similar to the previous plot, where the census and PES agreed on the number of residents at all households. With the described rate of disagreement between the PES and the census, a PES which surveys 0.72\% of households is required for the audit to approve the allocation with a risk limit of 0.1, compared to 0.66\% if there were no enumeration disagreements. To get a risk-limit of 0.05, the PES would need to survey 1\% of households, compared to 0.87\% with no enumeration errors.

\subsubsection{Simulation Conclusions} \label{sec: crla conclusions}

\paragraph*{Settings Where Our Method is Applicable}

As mentioned previously, our original goal was to run the census RLA over the US census and House of Representatives, but our method could not confirm such representative allocations with sufficient confidence unless the PES was very large. In nations with less representatives and federal-states, such as Cyprus, our method can approve the census with a relatively low risk-limit with a PES which covers 1\% of households. More generally, our method is applicable when the minimal census enumeration error which leads to a change in the resulting allocation of representatives is relatively large (0.25\% appears to be sufficient for Cyprus based on simulations). When there are many representatives and federal-states, even a small mistake in the census can lead to a wrongful allocation of representatives and auditing the census results requires a larger PES sample.

\paragraph*{Tolerance to Census and PES Disagreements}

Our method appears to be tolerant to small disagreements between the census and the PES results under these specified settings. A disagreement over 5\% of households leads to a small decrease in the audit's resulting risk-limit for any given PES sample size. A PES usually surveys around 1\% of households, which is sufficient for a risk-limit of 0.05 when the census and PES disagree on 5\% of households, compared to a risk-limit of 0.03 when there are no such disagreements. 

\subsubsection*{Recovering From Assertions with a High Risk-Limit}

Even when there are relatively few representatives and states, it's possible that a small enumeration error could lead to a different allocation of representatives. In such cases, a full census RLA would struggle to approve the census results with a sufficiently low risk-limit. If this occurs, the auditing body can examine the risk-limit of the individual assertions ($\frac{1}{T^{max}_{s_1,s_2}}$ for each $s_1,s_2\in \mathcal{S}$) to see which assertions have a higher risk-limit. As seen while simulating an election RLA in~\Cref{sec:batch RLA simulations}, in many cases, one specific assertion is significantly more difficult to approve than others, meaning that its risk-limit is significantly higher. If the risk-limit of all other assertions is sufficiently low, we might still decide that the census' allocation of representatives is reliable, with the exception of a single representative whose state-allocation could not be determined with confidence.

If the audit results in an insufficient risk-limit, we can also examine the risk-limit of each state, i.e.\ the risk-limit that would be outputted if we only wanted to approve that the number of representatives a specific state receives according to the census and according to the PES match. For a state $s$, this risk-limit is the maximal risk-limit of all assertions which involve $s$, meaning:
$$
    \max_{s'\in \mathcal{S}\setminus\{s\}}\left\{\frac{1}{T^{max}_{s,s'}}, \frac{1}{T^{max}_{s',s}}\right\}.
$$

\newpage
\section{Discussion and Further Research}

Throughout this work, we can observe that an election's social choice function and setting can severely limit the efficiency of their RLAs. Systems like the Israeli Knesset elections and the US House of Representatives' allocation of representatives to states are very sensitive to enumeration errors, making it difficult to audit them efficiently. 

When designing new election systems, the ability to audit their social choice function in a risk-limiting manner should be examined in advance. If a system has a sensitive social-choice function, it should be compensated by other means. E.g., it can use a vote tabulation system which returns the specific interpretation it gave each paper-backup ballots and ensure individual ballots can be accessed efficiently, allowing for ballot-comparison audits. If a system cannot be audited efficiently at all, implementing electronic vote tabulations for it should be done with extra care or avoided altogether.

\subsection*{The Batchcomp RLA}

In current literature, different RLA methods are usually compared by showing their simulated performance over some reported election results. Most RLA methods, including the ones presented in this work, are not analyzed in terms of query-complexity. This issue is especially prominent in the field of batch RLAs, where the performance of methods which can audit social choice functions beyond a simple plurality is often not analyzed at all, neither analytically or by using simulations. 

The simulations presented in this work (\Cref{sec:batch RLA simulations}) indicate that Batchcomp provides a noticeable improvement over ALPHA-batch in the limited settings that were tested. Despite this relative success, we cannot definitively say it outperforms existing methods without a clear, rigorous way of analyzing their efficiency.

\subsection*{Auditing the Knesset Elections}

\Cref{sec:batch RLA simulations} shows the difficulties in auditing the Knesset election. Due to the small margins these election results some times have, implementing RLAs for them seems could be problematic, since they will require a full recount when any constraint is close to be unsatisfied. 

Currently, a more practical use for RLAs in Israel could be to approve specific conditions regarding the election results. E.g., confirm that a certain party is above or below the electoral threshold. In such cases, the audit should be simulated in advance to ensure that the number of ballots it would require to read is manageable.

\subsection*{Census RLAs}

The census RLA method appears to be useful in some limited settings, and can be implemented using existing post-enumeration surveys. In systems where our method is currently not sufficient, a census RLA could perhaps aim for a weaker guarantee - that the number of representatives each state should receive according to the PES is close to the number it has according to the census.

The work raises many open questions and potential research directions:

\begin{description}
    \item \textbf{Applying RLAs in Additional Settings:}   Generally speaking, RLAs are relevant and can be applied whenever one wishes to verify the computation of some function over a large number of inputs obtained through potentially error-prone processes. While political elections provide a natural environment for their application, we advocate for their use in a wider range of settings to ensure reliable results. 
    
    As an example of such settings, RLAs could potentially be used to verify that decisions taken based on datasets which were altered in order to satisfy differential privacy are correct according to the real data. This could be achieved by running an RLA in a protected environment (enclave) which holds a subsample of the original, noiseless data. In this setting, the noisy, (differential private) dataset is seen as the reported result, while the noiseless dataset is the true results. An RLA can verify that the results of some computation over the differential private dataset and over the original noiseless dataset are likely to be identical, based on a (hopefully) small random sample from the original dataset. One challenge is to make sure that the very fact that the data passed the test does not hurt the desired differential privacy property.
    
    \item \textbf{Analytical Analysis of the Efficiency of RLAs:} Most recent literature in the field, including this work, focuses on suggesting new RLA algorithms and fitting them to additional electoral systems and settings. There is little to no analytical analysis of the {\em efficiency} and capabilities of many RLA methods. Without a more rigorous analysis, it is not possible to definitively determine which RLA methods are better for which settings. Such analysis could help, for instance, to argue analytically whether Batchcomp is indeed preferable over ALPHA-Batch.

    \item \textbf{Analyzing the Ability to Audit Different Systems:} Future research should analyze how efficiently different social choice functions can be audited. The ''audatibility" of a social choice function might be connected to the its noise stability, i.e., the probability of its outputted winners changing if every vote is changed with some equal, independent probability. If a social choice function has low noise-stability, it's more likely to lead to election results which have small margins, meaning they would be difficult to audit efficiently. The noise-stability of different voting rules has been examined previously in literature~\cite{heilman2022noise, heilman2021three}. Connecting these works to the field of RLAs may aid in determining the potential capabilities of RLAs for different election systems. 

    \item\textbf{Connection Between RLAs and Computational Models:} Thus far, advancements in the field of RLAs were done mostly independently and without connection to computational models. Finding such connections  may inspire new RLA algorithms, or suggest new methods for analyzing the capabilities and efficiency of existing methods. As an example of these connections, RLAs can essentially be viewed as randomized decision trees, where each branch represents a different sequence of paper-backup ballots that can be uncovered during the audit. Viewing RLAs in this manner allows us to analyze their query complexity (number of ballots examined) or instance complexity (best possible performance over specific election results) and to apply existing results from other fields onto RLAs.

    \item \textbf{Unlabeled Instance Complexity and RLAs:}  Future research may find lower bounds for the query-complexity of RLAs by analyzing the randomized unlabeled certificate complexity of the social choice function they operate on, as defined by Grossman, Komargodski and Naor~\cite{grossman2020instance}. The randomized unlabeled certificate complexity is a complexity measure of a function over some specific input. It's defined roughly as the minimal number of queries, in expectation, that any randomized decision tree which computes this function has to perform over the specified input, given a permuted version of it as a certificate. This notion is relevant for RLAs since they are essentially randomized decision trees which calculate a social choice function's output (the true winners) while using the reported election results. In the ballot-level RLA setting, these reported results are given as a reported tally of the votes, which is equivalent to an unlabeled certificate - a randomly permuted version of the paper-backup ballots. Therefore, it appears that an RLA's expected query-complexity over accurate reported results is lower bounded by the election's social choice function's randomized unlabeled certificate complexity over these same results.

Making such an observation, however, requires some adjustments in the definition of the randomized unlabeled instance complexity. RLAs are expected to be efficient even if their certificate is nearly accurate. Meaning, if the reported tally they receive only contains small errors which do not change the election winners, they are still expected to be relatively efficient. Decision trees which are optimal for a specific input may be very inefficient when the certificate is even slightly inaccurate. Thus, the unlabeled instance complexity of a function may be determined by randomized decision trees which would make for bad RLAs, as they may lead to a full recount if the reported results contain negligible mistakes. Without adjusting its definition, lower bounds which rely on this complexity measure may therefore be too loose. % A more interesting lower bound would be the minimal query complexity of randomized decision trees which compute a certain social choice function using a nearly accurate unlabeled certificate.
    
    %\barkar{Should I add something about instance optimality? In this context it means the reported results are not helpful, but this explanation is already pretty long}
    
    \begin{comment}
    Analyzing this complexity measure over an election's social choice function and reported results can provide lower bounds for the query-complexity of RLAs over these results. This is since the certificate that an RLA receives, at least in classical ballot-level RLA settings, is a tally of the votes according to the reported results. This tally is essentially a permutated version of the ballots reportedly cast in the election, which are the supposed input to its social choice function.
    \end{comment}

    \item\textbf{The Expressibility of SHANGRLA:} Future research should attempt to assess the capabilities and limitations of different RLA frameworks, such as SHANGRLA. Such research could, for example, find bounds on the efficiency (number of examined ballots) of a SHANGRLA based RLA given some reported election results, or discover classes of social choice functions which can and cannot be audited using SHANGRLA. Some social choice functions, such as instant runoff voting, do not currently have reductions to SHANGRLA assertions which are both sufficient and necessary for the reported winners of the elections to be correct~\cite{stark2020sets}. Finding clear limitations for SHANGRLA can prevent researchers from trying to develop SHANGRLA based RLAs for systems where it cannot apply.     

    \item\textbf{Batch-Level RLAs Beyond SHANGRLA:} Currently, There are few batch-level RLA algorithms for social choice functions beyond a simple plurality. This work suggests a generic method for converting any SHANGRLA based RLA to a batch-level RLA. Some election systems, such as instant runoff voting, do not currently have a reduction to SHANGRLA assertions that are both sufficient and necessary for the reported winners of the elections to be correct. This raises the need for an even more general batch-level RLA method, which can be used in systems that cannot currently be audited using SHANGRLA. 
    
    \item \textbf{Census RLAs:} This work presents the first RLA for population censuses. It is possible and even likely that other RLA algorithms could be adapted to audit censuses, perhaps with different goals or guarantees than our suggested method.

    Additionally, our census RLA method could potentially be optimized further. This method operates iteratively, by sequentially sampling households and reading their census results, and maintains a probability with which it can approve the census at all times. It is possible that other algorithms, which use the entirety of the census and PES data in one shot, could outperform our method. 
\end{description}

\section{References}
\printbibliography[heading=none]

@misc{StateLegislatureConf,
  title = {A full list of U.S.A. states that conduct risk-limiting audits is available on the National Conference of State Legislature's website},
  howpublished = {\url{https://www.ncsl.org/research/elections-and-campaigns/risk-limiting-audits.aspx}}
}

@article{lindeman2012gentle,
  title={A gentle introduction to risk-limiting audits},
  author={Lindeman, Mark and Stark, Philip B.},
  journal={IEEE Security \& Privacy},
  year="2012",
  volume={10},
  number={5},
  pages={42--49},
  year={2012},
  publisher={IEEE}
}

@article{lindeman2018next,
  title={Next steps for the Colorado risk-limiting audit (corla) program},
  author={Lindeman, Mark and McBurnett, Neal and Ottoboni, Kellie and Stark, Philip B.},
  journal={arXiv preprint arXiv:1803.00698},
  year={2018}
}

@article{blom2019raire,
  title={RAIRE: Risk-limiting audits for IRV elections},
  author={Blom, Michelle and Stuckey, Peter J. and Teague, Vanessa},
  journal={arXiv preprint arXiv:1903.08804},
  year={2019}
}

@article{stark2014verifiable,
  title={Verifiable European elections: Risk-limiting audits for D’Hondt and its relatives},
  author={Stark, Philip B. and Teague, Vanessa and Essex, Aleksander},
  journal={$\{$USENIX$\}$ Journal of Election Technology and Systems ($\{$JETS$\}$)},
  volume={1},
  pages={18--39},
  year={2014}
}

@inproceedings{stark2020sets,
  title={Sets of half-average nulls generate risk-limiting audits: SHANGRLA},
  author={Stark, Philip B.},
  booktitle={International Conference on Financial Cryptography and Data Security},
  pages={319--336},
  year={2020},
  organization={Springer}
}

@inproceedings{grossman2020instance,
  title={Instance Complexity and Unlabeled Certificates in the Decision Tree Model},
  author={Grossman, Tomer and Komargodski, Ilan and Naor, Moni},
  booktitle={11th Innovations in Theoretical Computer Science Conference (ITCS 2020)},
  year={2020},
  organization={Schloss Dagstuhl-Leibniz-Zentrum f{\"u}r Informatik}
}

@article{stark2022alpha,
  title={ALPHA: Audit that Learns from Previously Hand-Audited Ballots},
  author={Stark, Philip B.},
  journal={arXiv preprint arXiv:2201.02707},
  year={2022}
}

@inproceedings{lindeman2012bravo,
  title={BRAVO: Ballot-polling Risk-limiting Audits to Verify Outcomes},
  author={Lindeman, Mark and Stark, Philip B. and Yates, Vincent S.},
  booktitle={EVT/WOTE},
  year={2012}
}

@inbook{durrett2019probability,
  title={Probability: theory and examples},
  author={Durrett, Rick},
  volume={49},
  year={2019},
  pages = {235},
  publisher={Cambridge university press}
}

@inproceedings{blom2021assertion,
  title={Assertion-based approaches to auditing complex elections, with application to party-list proportional elections},
  author={Blom, Michelle and Budurushi, Jurlind and Rivest, Ronald L. and Stark, Philip B. and Stuckey, Peter J. and Teague, Vanessa and Vukcevic, Damjan},
  booktitle={International Joint Conference on Electronic Voting},
  pages={47--62},
  year={2021},
  organization={Springer}
}

@inproceedings{waudby2021rilacs,
  title={RiLACS: Risk limiting audits via confidence sequences},
  author={Waudby-Smith, Ian and Stark, Philip B. and Ramdas, Aaditya},
  booktitle={International Joint Conference on Electronic Voting},
  pages={124--139},
  year={2021},
  organization={Springer}
}

@misc{NcslRLA,
  title = {NCSL Risk-Limiting Audits},
  howpublished = {\url{https://www.ncsl.org/research/elections-and-campaigns/risk-limiting-audits.aspx}}
}

@inproceedings{schurmann2016risk,
  title={A risk-limiting audit in Denmark: A pilot},
  author={Sch{\"u}rmann, Carsten},
  booktitle={International Joint Conference on Electronic Voting},
  pages={192--202},
  year={2016},
  organization={Springer}
}

@misc{american2010american,
  title={American Statistical Association statement on risk-limiting post-election audits},
  author={American Statistical Association},
  year={2010}
}

@book{national2018securing,
  title={Securing the Vote: Protecting American Democracy},
  author={National Academies of Sciences, Engineering, Medicine},
  year={2018},
  pages={29--30},
  publisher={National Academies Press}
}

@misc{center2013presidential,
  title={Presidential Commission on Election Administration},
  author={Center, Bipartisan Policy},
  year={2013}
}

@article{stark2008conservative,
  title={Conservative statistical post-election audits},
  author={Stark, Philip B.},
  journal={The Annals of Applied Statistics},
  volume={2},
  number={2},
  pages={550--581},
  year={2008},
  publisher={Institute of Mathematical Statistics}
}

@article{stark2009risk,
  title={Risk-limiting postelection audits: Conservative $ P $-values from common probability inequalities},
  author={Stark, Philip B.},
  journal={IEEE Transactions on Information Forensics and Security},
  volume={4},
  number={4},
  pages={1005--1014},
  year={2009},
  publisher={IEEE}
}

@inproceedings{spertus2022sweeter,
  title={Sweeter than SUITE: Supermartingale Stratified Union-Intersection Tests of Elections},
  author={Spertus, Jacob V. and Stark, Philip B.},
  booktitle={International Joint Conference on Electronic Voting},
  pages={106--121},
  year={2022},
  organization={Springer}
}

@inproceedings{blom2020random,
  title={Random errors are not necessarily politically neutral},
  author={Blom, Michelle and Conway, Andrew and Stuckey, Peter J. and Teague, Vanessa J. and Vukcevic, Damjan},
  booktitle={International Joint Conference on Electronic Voting},
  pages={19--35},
  year={2020},
  organization={Springer}
}

@article{ansolabehere2018learning,
  title={Learning from recounts},
  author={Ansolabehere, Stephen and Burden, Barry C. and Mayer, Kenneth R. and Stewart III, Charles},
  journal={Election Law Journal: Rules, Politics, and Policy},
  volume={17},
  number={2},
  pages={100--116},
  year={2018},
  publisher={Mary Ann Liebert, Inc. 140 Huguenot Street, 3rd Floor New Rochelle, NY 10801 USA}
}

@misc{knessetElectionMethod,
  title={The Distribution of Knesset Seats (Bader-Ofer Method)},
  howpublished = {\url{https://m.knesset.gov.il/en/about/lexicon/pages/seats.aspx}},
  note={This page wrongly states that the electoral threshold is 2.0\%. It was since changed to 3.25\%}
}

@book{saltman1978effective,
  title={Effective use of computing technology in vote-tallying},
  author={Saltman, Roy G.},
  volume={13},
  year={1978},
  chapter={V},
  publisher={US Department of Commerce, National Bureau of Standards}
}

@article{mccarthy2008percentage,
  title={Percentage-based versus statistical-power-based vote tabulation audits},
  author={McCarthy, John and Stanislevic, Howard and Lindeman, Mark and Ash, Arlene S. and Addona, Vittorio and Batcher, Mary},
  journal={The American Statistician},
  volume={62},
  number={1},
  pages={11--16},
  year={2008},
  publisher={Taylor \& Francis}
}

@article{rivest2006estimating,
  title={On estimating the size of a statistical audit},
  author={Rivest, Ronald L.},
  year={2006}
}

@misc{simon2006end,
  title={An end to
‘faith-based’ voting: universal precinct-based
handcount sampling to check computerized
vote counts in federal and statewide elections},
  author={Simon, J. D. and O’Dell, B.},
  year={2006}
}

@book{norden2007post,
  title={Post-election audits: Restoring trust in elections},
  author={Norden, Lawrence D. and Burstein, Aaron and Hall, Joseph Lorenzo and Chen, Margaret},
  year={2007},
  publisher={Brennan Center for Justice},
  chapter={II}
}

@article{johnson2004election,
  title={Election certification by statistical audit of voter-verified paper ballots},
  author={Johnson, Kenneth C.},
  journal={Available at SSRN 640943},
  year={2004}
}

@misc{un2010pes,
  title={Post Enumeration Surveys
Operational guidelines},
  author={United Nations Secretariat, Department of Economics and Social Affairs, Statistics Division},
  year={2010},
  howpublished = {\url{https://unstats.un.org/unsd/demographic/standmeth/handbooks/Manual_PESen.pdf}}
}

@article{gallagher1991proportionality,
  title={Proportionality, disproportionality and electoral systems},
  author={Gallagher, Michael},
  journal={Electoral studies},
  volume={10},
  number={1},
  pages={33--51},
  year={1991},
  publisher={Elsevier}
}

@misc{denmarkElectorealSystem,
  title={The Electoral System in Denmark - Parliamentary elections},
  author={The Danish Ministry of Social Affairs},
  howpublished = {\url{https://elections.sim.dk/media/10507/the-electoral-system-in-denmark.pdf}}
}

@misc{argentinaConstitution,
  title={Constitution of the Argentine Nation},
  howpublished = {\url{http://www.biblioteca.jus.gov.ar/argentina-constitution.pdf}},
  note={Sections 45-46}
}

@misc{pakistanElectoralSystem,
  title={National Assembly of Pakistan - Composition},
  howpublished = {\url{https://na.gov.pk/en/composition.php}}
}

@misc{cyprusCensus,
  title={Census of Population and Housing 2021: Preliminary Results},
  year={2022},
  author={Statistical Service of the Republic of Cyprus},
  howpublished = {\url{https://www.pio.gov.cy/en/press-releases-article.html?id=27965}}
}

@misc{usCensusHouseholds,
    title={Historical Households Tables},
    author={United States Census Bureau},
    year={2022},
    howpublished = {\url{https://www.census.gov/data/tables/time-series/demo/families/households.html}},
    note={Table HH-4}
}

@misc{bundestagAllocation,
    title={Musterberechnung:
    Sitzverteilung nach dem Fünfundzwanzigsten Gesetz zur Änderung des Bundeswahlgesetzes},
    author={Federal Returning Officer of Germany (Bundeswahlleiter)},
    year={2020},
    howpublished = {\url{https://www.bundeswahlleiter.de/dam/jcr/05c1185a-173f-4bab-80d6-51027c94b1bc/bwg2021_mustersitzberechnung_ergebnis2017.pdf}},
    note={In German. An unofficial explanation in English based on this document is available at \url{https://mikebeneschan.medium.com/the-algorithm-that-fills-germanys-parliament-fa1e10c85917}}
}

@article{hu2022determining,
  title={Determining the sample size of a post-enumeration survey: The case of China, 2020},
  author={Hu, Guihua and Wen, Ting and Liu, Yuhuan},
  journal={Mathematical Population Studies},
  pages={1--31},
  year={2022},
  publisher={Taylor \& Francis}
}

@misc{usPesSize,
    title={Briefing On Post Enumeration Census Results},
    year={2022},
    author={National Governor's Association},
    notes={In Hebrew},
    howpublished={\url{https://www.nga.org/updates/briefing-on-post-enumeration-census-results/\#:\~:text=The\%20Post\%2DEnumeration\%20Survey\%20is,independent\%20of\%20the\%20census\%20operations}}

}

@article{yeminHadash,
    title={Vaadat habchriot layamin hadash: "Hateanot hasrot shachar, novot memetzuka" [The Central Election Committee to the new right party: "The claims are baseless, born of distress"]},
    author={Zeev Kam},
    date={2019-04-16},
    Journal={Kan (Israeli Public Broadcasting Corporation)},
    howpublished={\url{https://www.kan.org.il/item/?itemid=50550}}
}

@report{decisionsGuidelines,
    title={Decisions and guidelines from the elections of the 19th Knesset},
    author={The Central Election Committee of Israel},
    year={2013},
    pages={359--360},
    notes={In Hebrew},
    howpublished={\url{https://bechirot24.bechirot.gov.il/election/Decisions/AllDecisions/PreviousEllections/Documents/decisions_rubinstein_knesset19_2012_2013.pdf}}
}

@article{heilman2022noise,
  title={Noise Stability of Ranked Choice Voting},
  author={Heilman, Steven},
  journal={arXiv preprint arXiv:2209.11183},
  year={2022}
}

@inproceedings{heilman2021three,
  title={Three candidate plurality is stablest for small correlations},
  author={Heilman, Steven and Tarter, Alex},
  booktitle={Forum of Mathematics, Sigma},
  volume={9},
  year={2021},
  organization={Cambridge University Press}
}

@article{charalambous2008house,
  title={The House of Representatives},
  author={Charalambous, Giorgos},
  journal={The Politics and Government of Cyprus. Oxford: Peter Lange},
  pages={143--168},
  year={2008}
}

\end{document}